\documentclass[11pt]{article}
\usepackage[margin=2cm]{geometry}
\usepackage{setspace}

\usepackage{multirow}  % tables
\usepackage[table,xcdraw]{xcolor}  % tables
\usepackage{float} % keep the position of a table
\restylefloat{table} % keep the position of a table
\usepackage{tabularx,ragged2e,booktabs}%,caption
\usepackage[justification=centering]{caption}

\usepackage[flushleft]{threeparttable} % insert footnotes in the tables

\usepackage{appendix}

\usepackage{fancyhdr}

\usepackage[natbibapa]{apacite}
\renewcommand{\&}{and}
\RequirePackage[colorlinks,citecolor=blue,urlcolor=blue]{hyperref}
\RequirePackage{hypernat}

\usepackage{moresize}

\usepackage{amsfonts}
\usepackage{amsmath}
\usepackage{amssymb} % proof - \blacksquare
\usepackage{mathtools}

\usepackage{dsfont} % mathds
\usepackage{scalerel} % scaleobj
\usepackage{upgreek} % fancy greek letters \upalpha, \upbeta,...
\usepackage{bm} % bold upgreek

\usepackage{xifthen}% provides \isempty test

\usepackage{verbatim} % comments
\usepackage{amsthm}

\usepackage{soul} %  a striking line over text \st

\newtheoremstyle{dotless}{}{}{\itshape}{}{\bfseries}{}{ }{}

\newtheorem{theorem}{Theorem}

\theoremstyle{definition}	
\newtheorem{definition}{Definition}
\numberwithin{definition}{section} 

\theoremstyle{informaldefinition}	

\numberwithin{informaldefinition}{section}

\newtheorem{lemma}{Lemma}
\numberwithin{lemma}{section}
\numberwithin{theorem}{section} % this is relevant to both sections and appedices

\newcommand{\E}{\mathbb{E} }
\newcommand{\be}{{\bm \epsilon}}

\newcommand\independent{\protect\mathpalette{\protect\independenT}{\perp}}
\newcommand\ci{\protect\mathpalette{\protect\independenT}{\perp}}
\def\independenT#1#2{\mathrel{\rlap{$#1#2$}\mkern2mu{#1#2}}}

\usepackage{xparse}

\newif\iffirstitem
\firstitemtrue   
\newcommand\myitem[1]{%
	\iffirstitem
	\firstitemfalse
	\else
	; %
	\fi
	\citeauthor{{#1}}, \citeyear{{#1}}}

\NewDocumentCommand\mycite{>{\SplitList{,}}m}
{%
	\ProcessList{#1}{\myciteitem}%
	\firstitemtrue
}

\newif\iffirstitem
\firstitemtrue   
\newcommand\myciteitem[1]{%
	\iffirstitem
	\firstitemfalse
	\else
	; %
	\fi
	\citeauthor{{#1}}'s (\citeyear{{#1}})}

\NewDocumentCommand\myciteopt{>{\SplitList{,}}m}
{%
	(\ProcessList{#1}{\myitem}, to mention a few)%
	\firstitemtrue  
}

\usepackage{xstring} % if then else

\newcommand\close[1]{\left( {#1} \right)}
\newcommand\braket[1]{\left[ {#1} \right]}
\newcommand\curly[1]{\left \{{#1} \right\}}
\title{Context-dependent Causality\\ (the Non-Monotonic Case)}
\author{Nir Billfeld\thanks{University of Haifa. nirbillfeld@gmail.com}\hspace{0.5em} Moshe Kim\thanks{University of Haifa. kim@econ.haifa.ac.il.}\thanks{We thank Boaz Nadler from the Weizmann Institute of Science for his insightful comments.}} % parallel computer lab of , without whose help we would have been still running these simulations as we are circulating this paper.}}% \hspace{0.5em} Boaz Nadler\thanks{Weizmann Institute of Science. boaz.nadler@weizmann.ac.il}}%\\\\\scalebox{0.8}{The Weizmann Institute of Science}}
\date{\today}

\usepackage{float} 
\floatstyle{plaintop}
\restylefloat{table}
\theoremstyle{cond}
\numberwithin{cond}{section}

\newtheorem{asu}{Assumption}
\theoremstyle{asu}
\numberwithin{asu}{section}
\newtheorem{assumption}{Assumption}
\newtheorem{subasu}{Assumption}[assumption]

\let\oldproofname=\proofname
\renewcommand{\proofname}{\rm\bf{\oldproofname}} % bold proof

\usepackage{empheq}
\usepackage[most]{tcolorbox}

\usepackage[ruled, lined, linesnumbered, commentsnumbered, longend]{algorithm2e} %ruled,vlined

\usepackage{tcolorbox}
\begin{document}
	
	\maketitle

	%\renewenvironment{abstract}{%
		%       \hfill\begin{minipage}{0.95\textwidth}
			%               \rule{\textwidth}{1pt}}
		%       {\par\noindent\rule{\textwidth}{1pt}\end{minipage}}
	%
	
	% why context is importance
	%       what is  it
	% the problematic in the existing literature: monotonicity
	% why monotonicity is a problem
	% The proposed identifiabilti stargent benfints  from the universality...
	% my model not only solve context it also solve monotnicitiy.
	% 3-4 gans
	
	\begin{abstract}
We develop a novel identification strategy as well as a new estimator for  context-dependent causal inference in non-parametric triangular models with non-separable disturbances. Departing from the common practice, our analysis does not rely on the strict monotonicity assumption. 
Our key contribution lies in leveraging on \textbf{\textit{diffusion models}} to formulate the structural equations as a system evolving from noise accumulation to account for the influence of the latent context (confounder) variable on the outcome. 
Our identifiability strategy involves a system of Fredholm integral equations expressing the distributional relationship between a latent context variable and a vector of observables.	 
These integral equations involve an unknown kernel and are governed by a set of structural form functions, inducing a \textbf{\textit{non-monotonic inverse problem}}. 
We prove that if the kernel density can be represented as an infinite mixture of Gaussians, then there exists a unique solution for the unknown function. This is a significant result, as it shows that it is possible to solve a non-monotonic inverse problem even when the kernel is unknown.	 
On the methodological front we leverage on a novel and enriched \textit{Contaminated} Generative Adversarial (Neural)  Networks (\textbf{CONGAN}) which we provide as a solution to the non-monotonic inverse problem. 			
%     Consequences are propagated by actions taken in a specific and often latent contextual platform. Overlooking this feature when analyzing causality in triangular models may introduce wrong causal relations. Thus, we develop a novel identification strategy as well as a new estimator for context-dependent non parametric triangular models in the presence of non-separable disturbances. Unlike the common practice our analysis does not rely on the strict monotonicity assumption. Our identifiability strategy involves a system of \textit{Fredholm} integral equations expressing the distributional relationship between a latent context variable and a vector of observables. These integral equations involve an unknown kernel and are governed by a set of structural form functions, inducing \textit{a non-monotonic inverse problem}. This very formulation facilitates the establishment of uniqueness of the interventional distribution by employing the inverse Radon transform. On the methodological front we leverage on a novel and enriched Contaminated Generative Adversarial (Neural)  Networks (CONGAN) providing a solution to the non-monotonic inverse problem. Monte-Carlo simulations generating $100$ data sets involving important economic fundamentals, each consists of $10,000$ observations, exhibit high accuracy. Accuracy is measured by applying Sliced-Wasserstein-Distance test for similarity of the estimated synthetic counterfactual distribution to the real distribution benchmark. Jensen-Shannon-Divergence criterion attests for stability. 
	\end{abstract}
	
	\def\keywords{\vspace{.5em}
		{\textit{Keywords}:\,\relax%
	}}
	\def\endkeywords{\par}
	
	%    \begin{Keywords} %Counterfactual \sep%
		%    \end{Keywords} 
	\keywords Counterfactual; Diffusion; Synthetic distribution; Unrestricted support; Contaminated Generative Adversarial Networks; Context; Non-monotonic inverse problem; Unknown kernel function; Neural Networks; Fredholm integral equation
	\vspace{20em}
	%------------------------------
	\pagebreak
	
	\renewcommand{\theenumi}{(\roman{enumi})}
	\renewcommand{\labelenumi}{\theenumi}			
	\section{Introduction}  
	Consequences are propagated by actions taken in a specific and often latent contextual platform. As such, context is an intrinsic part of a causal model, playing the role of  unobserved heterogeneity. The common practice in the treatment of unobserved heterogeneity is to impose various kinds of structural form restrictions on the relationships among actions and consequences. A widely used structural restriction is the monotonicity assumption, which implies a one-to-one relationship between either the action or the outcome variable and its random disturbance (given the control variables). A key limitation of this assumption however, is that it rules out any source of uncertainty regarding the unobservables when the observables are held fixed. Namely, monotonicity forces a degenerate conditional distribution of the unobservables given the observables. Consequently, the theoretical model is restricted to a subset of structural form functions and thus, may poorly identify and capture the true underlying causal relationship between action and outcome. 
	
	We depart from the conventional reliance on monotonicity to achieve identifiability. Yet, we keep the general architecture similar to the conventional modeling. The present model is a nonparametric triangular model with non-separable disturbances, which includes a latent context variable. Context is modeled as a confounder with unobserved heterogeneity given the observables, challenging the conventional causal inference. 
	The main challenge in achieving identifiability is aggravated due to the absence of an adjustment set (observed control variables) to be conditioned on in order to satisfy some conditional independence assumption. This identifiability relies on formulating the structural equations as a system evolving from noise accumulation, known as a  \textbf{\textit{diffusion process}} to account for the influence of the latent context (confounder) on the outcome. This allows us to achieve identifiability without the limitations of monotonicity assumptions. The key result of the offered identifiability approach introduced here importantly emanates from the fact that any possible set of structural functions belonging to an equivalence class of triangular models generating the data yields the same do-interventional counterfactual distribution. This result is achieved by expressing the distributional relationship between a latent context variable and a vector of observables through a system of Fredholm integral equations.
	The aforementioned set of equations is governed by the set of generator functions and an \textit{unknown kernel function}, inducing a \textbf{\textit{non-monotonic inverse problem}}. 
	The role of these generator functions is two-fold: (i) to characterize 
	an admissible kernel function induced by these generator functions and (ii) to ensure that the estimator of the unknown quantity is a continuous function of the data for any given kernel function in the equivalence class. We establish that the interventional distribution is invariant to the choice of the kernel if it belongs to any strongly complete family of densities, e.g., an infinite mixture of Gaussians. This is a significant result, as it shows that it is possible to solve a non-monotonic inverse problem even when the kernel is unknown.	 
	
	The framework offered is general and may have a significant impact on a wide variety of economic and other phenomena inclusive of those intrinsically exhibiting non-monotonicity in their confounding variables.	The relationship between the equivalence class of structural form functions and the causal distribution is formed by the establishment of a newly introduced augmented Fourier series expansion. This novel expansion method is designed to characterize an infinite system of linear regression equations satisfying cross-equation restrictions, induced by an equivalence class of non-separable structural form functions. These cross-equation restrictions are derived from the common triangular model assumptions, without reliance on any variant of monotonicity.

	The existing literature using either local average treatment effects or triangular models architecture impose various functional form restrictions, such as conditional quantile restrictions \citep{chesher2003identification}, strict monotonicity in the case of average local treatment effects \citep{angrist1994identification} as well as in the case of triangular models \citep{imbens2009identification,hoderlein2017corrigendum}. Alternative models assume restricted support domain (locality) \citep{hoderlein2009identification,angrist1996identification, altonji2005cross, heckman2005structural}. These restrictions guarantee the identifiability of the latent confounder up to a strictly monotonic transformation from conditional quantiles of the observables. 
	The main limitation of such treatment however, is to potentially narrow down the support domain of the latent space. This can be detrimental to causal inference as it eliminates a subset of the control (placebo) group, when employing e.g., the common practice of partial means \citep{newey1994kernel}, interventions \citep{pearl2019interpretation} and local average treatment effect \citep{angrist1996identification}. 
	%In many causal structural models, as in the classical and the influential causal model of \cite{imbens2009identification}, all the disparity between association and causation is related to an identifiable unique confounder, a control variable affecting both the action as well as the outcome (common cause). The very identifiability of this confounder rules out any randomness from being incorporated in the causal model, as this confounder must be determined uniquely by the action and the observable covariates. Technically, this is achieved by the imposition of strict monotonicity  of the observed outcome or action on the latent confounder \myciteopt{blundell2014control, heckman1985alternative, hoderlein2009identification,imbens2009identification}. 
	Monotonicity is rather not an innocuous restriction as it is at odds with many phenomena in general and those in economics in particular. The potential bias propagated by relying on the monotonicity assumption has long been estimated, (e.g., \cite{klein2010heterogeneous}). However, there is a lack of alternative identifiability, inference and an appropriate estimator which are not monotonicity-based.
	Monotonicity is a strong and perhaps unpalatable for many basic economic phenomena such as demand, supply and consumption behaviour \citep{hoderlein2016testing, hoderlein2007identification}, multiple equilibria in strategic behavior \citep{myerson1999nash}, search behavior \citep{chetverikov2019testing, agarwal2020searching} as well as in prospect theory \citep{tversky1979analysis}.\footnote{A recent paper \citep{dembo2021ever} documents strong violations of monotonicity as a feature of expected utility theory.}% as well as the disjunction fallacy of beliefs in the experimental literature \citep{tversky1974judgment}. Additionally, as will be shown here, monotonicity inflicts degeneracy onto context and propagates restriction on the support domain. This reduces entropy and may misspecify the true underlying causal and counterfactual relations. 
	
	In our present framework, we depart from the above mentioned restrictions by allowing for the posterior distribution of the latent context to be a function of the covariates and of the outcome. This generalization incorporates cases in which the context is also a function of the outcome and not only the covariates, unlike the control variable approach \myciteopt{mammen2012nonparametric,mammen2016semiparametric,chernozhukov2020}.
	%We develop a non-parametric context-dependent causal inference allowing for the possibility of different latent contextual regimes to be incorporated in the triangular model. The context covariate as is dealt with here is in fact a latent nonlinear confounder. The inference of context dependent causality, which exploits the unrestricted support domain of the latent context, is the main contribution of our paper. This yields a rather generic non-parametric, context-dependent and monotonicity-free causal relation, if and when it is present in the data. All these features have an important ramification for the proper synthetic counterfactual distribution.
	The presence of a latent context covariate affects the cause and the outcome simultaneously and thus, is the source of endogeneity in our model.
	Thus, the endogeneity can be eliminated by allocating each one of the cause and outcome a different and independent realization of the context. 
	Building upon this idea, we develop a novel identification strategy as well as a new estimator for triangular models in the presence of non-separable disturbances. Unlike the common practice, our approach does not rely on the strict monotonicity assumption, but rather on an integral equation characterized by generator functions. These generator functions induce an equivalence class of triangular models, while exploiting the entire support domain of the latent space in a non-monotonic manner. This insight solves the double-hurdle problem of relying on monotonicity as well as on a restricted support domain \citep{chernozhukov2020} which may produce inferior results as is shown in our simulations.
	%\textcolor{red}{Simply put, any given set of generator functions induces an inverse problem, in which the unknown function is the joint distribution of outcome, action and context, whereas the known function is the joint distribution of the observables. We introduce a \textit{full rank} assumption for the inverse problem, which is both necessary and sufficient condition for determining uniquely each of these unknown functions from the observables. Further, we show that all of these recovered functions yield the same counterfactual distribution of outcome and action.}%yielding data with the identical joint distribution as those of the observables 		 
	
	The new paradigm presented here yields a synthetic counterfactual distribution, which is entropy-advantageous. 
	We provide a novel feed-forward multi-layer Neural network-based estimator.  
	%The main building blocks in the proposed methodology is the detachment of the action variable from the context, that is obscured in the data, due to their joint dependence. In performing this newly introduced methodology we achieve the aforementioned randomness through generator functions produced by Neural networks yielding synthetic counterfactual distribution for the purpose of causal inference. 
	This is performed by leveraging on a novel extension namely, a \textit{\textbf{Contaminated}} version we develop and apply to the Generative Adversarial Networks model (\textbf{CONGAN}) \citep{goodfellow2014generative}. The advantage of the proposed practice is that it enables the Neural network to be trained on the entire support domain of the latent space unlike the common practice, which uncovers counterfactual relationships by averaging available observable data on a restricted support.
	
	The paper proceeds as follows. Section \ref{Sec:Discussion} gives an overview of the common practice of causal inference in triangular models. Section \ref{Sec:Theoretical:model} presents the diffusion problem formulation and its underlying assumptions. %Section \ref{Sec:integral} presents an interventional distribution througan integral equation.
	 Section \ref{Sec:alternative_representation} shows an alternative formulation of the triangular model as an integral equation to simplify the analysis. Section \ref{sec:example} provides an illustrative example for identifiability which does not rely on monotonicity assumption. Section \ref{sec:completeness} characterizes strong completeness of sufficient statistic as a building block for identifiability.
	Section \ref{Sec:Identification} presents the identifiability by relying on strong completeness. The estimator is presented in section \ref{Sec:Estimation}. %Asymptotic results and efficiency are established in section \ref{Sec:Asymptotic}. 
	In section \ref{Sec:Simulations} Monte-Carlo simulations are used to verify our theoretical model performance. Section \ref{Sec:Conclusion} concludes.

	\section {Discussion and literature review}\label{Sec:Discussion}
	%Consequently, LAR does not focus on answering the question of what would have been the expected counterfactual outcome of a specific individual if it was subjected to a lottery assigning the value of the treatment variable, which can be of importance for policymakers.
	Various identification strategies have been designed to uncover counterfactual relationships in triangular models, consisting of non-additively separable disturbances, in the absence of experimental data. The ultimate goal of these approaches is to mimic a ``natural experiment'' \citep{rosenzweig2000natural} by detaching a random variable $X$ (an observed action or treatment) from the unobservables affecting some other random variable $Y$ (outcome). The prominent question they answer is what would have been the expected counterfactual outcome of $Y$ if the values of the independent variable $X$ were assigned independently of the disturbances. Generally, in order to achieve this goal, monotonicity is imposed either on the unobservables in the first stage regression  \myciteopt{imbens2009identification,blundell2014control, heckman1985alternative,hoderlein2009identification,chernozhukov2020}; or on the unobservables in the outcome's structural equation in the second stage \citep{chernozhukov2007instrumental, kim2020partial}. In other cases, the monotonicity requirement is imposed on the instrumental variable \citep{angrist1994identification,hoderlein2017corrigendum,heckman2018unordered}.  However, this assumption is at odds with many actual phenomena as discussed earlier. Moreover and perhaps more importantly, it imposes severe limitations on the set of identifiable causal relationships in nonparametric inference, while potentially missing the true causal relationship. The unattended issue in the existing literature is that the direction and strength of causal relationship is context-blind.
	
	An alternative identifiability strategy for causal inference not relying on monotonicity is obtained by assuming ignorability which is known also as an unconfoundedness assumption \citep{athey2016recursive} or a conditional independence on observables \citep{farrell2021deep}. This approach limits the scope of causal analysis by requiring a directly observed set of control variables, which are rarely available in reality, in order to introduce conditional independence. The present study departs from this approach by embracing a generative covariate mechanism \citep{mammen2012nonparametric,mammen2016semiparametric} producing a synthetic sample of the same distribution present in the data. This approach does not necessitate the unconfoundedness assumption, yet allows for the estimation of the counterfactual distribution at any level of accuracy desired (given amount of data available and level of complexity). This methodology also departs from the popular quantile structural estimation, which is intrinsically monotonic \citep{chernozhukov2013inference}.

	%	\st{Our main contribution in this paper is in that we do not rest identifiability on a degenerate latent covariate distribution as is commonly done in the literature. Degeneracy of the latent covariate distribution implies that the latent covariates must be deterministically determined given all of the observables. This however, precludes a conditional unobserved heterogeneity, limiting the viability of correct identifiability. Additionally and importantly,} 
	We contribute to the existing literature in two specific aspects.  The theoretical contribution amounts to developing a new causal identification paradigm, which is non-parametric, monotonicity-free, context-dependent and universal with respect to the identifiable function set. Our new paradigm yields a counterfactual distribution exploiting the unrestricted support domain of the latent space, which is essential for the practice of intervention and counterfactual inference. On the methodological front, we develop a novel estimation strategy referred to here as a Contaminated Generative Adversarial Networks (CONGAN). This approach consists of two generators and one classifier: counterfactual samples generator and a contaminator generator. The classifier determines whether the sample is real or synthetic as in the conventional GAN \citep{goodfellow2014generative}. The main difference is in the unique role of each one of the generators in the causal inference. The counterfactual samples generator intends to mimic the causal relationship between $X$ and $Y$, which would be the case if $X$ and the unobservables were distributionally independent. However, it cannot be trained directly using the observed data, in which $X$ and the latent context variable are distributionally jointly dependent. Consequently, the role of the contaminator is to minimize the disparity between the observed and the counterfactual samples. As such, the contaminator is a nuisance Neural network which is essential for training the counterfactual generator. Our new paradigm enables to mimic the distribution of $Y$ given $X$ that would have been obtained under random assignment of the latent context variable (independently of $X$). We emphasize at this stage that context is not just a semantic concept but rather an identifiability instrument, in that it permits rendering $Y$ a different contextual regime from the one governing $X$. Once the regime governing $Y$ does not longer comove with $X$, one mimics the effect of an exogenous variation in $X$ on $Y$. The benefit of the contaminated Neural network is in that it precludes the need to invert the generators in order to account for the posterior distribution of the latent context variable given the observables, which is computationally cumbersome. 
	%\st{A further and cardinal aspect is that our architecture can tolerate undetectable variations in the latent space, which do not affect the observed distribution. This, in and by itself, seems prohibitive for identifiability. This is so, because identifiability utilizes the latent space.  We provide a theoretical result showing that the aforementioned discrepancy can be resolved by enabling the model to be immune to non-monotonic transformations, which enables the identifiability of the counterfactual relationships.} 
	Our CONGAN application utilizes multi-layer feed-forward Neural networks in which, unlike in other series estimators (e.g., sieve), the basis functions themselves are data-driven by  combining simple functions. Such a flexible combination is known to be capable of approximating any measurable function to any desired degree of accuracy \citep{hornik1989multilayer}.\footnote{We note that a specific variant of GANS, referred to as Wasserstein GANS with a penalized gradient, has been recently employed for conducting monte-carlo simulations of potential outcomes under unconfoundedness assumption \citep{athey2020}. Recent work criticizes this penalized Wasserstein GANS due to disregarding significant parts of the support domain in the data \citep{wei2018improving}. Note that the utilization of the correct support is at the core of counterfactual analysis.} In the ensuing section we attend to the theoretical triangular model formulation. %optimization challenges \citep{bahat2020explorable, gulrajani2017improved} as well as 

	%%%%%%%%%%%%%%%%%%%%%%%%%%%%%%%%%%%%%%%%%%%%%%%%%%%%
	%%%%%%%%%%%%%%%%%%%%%%%%%%%%%%%%%%%%%%%%%%%%%%%%%%%%
	\section{Diffusion-based non-monotonic triangular models}
	%	\section{Problem formulation}
	\label{Sec:Theoretical:model} 
	Consider the following triangular simultaneous equations model, pioneered by \citet{heckman1985alternative}. 
	The model includes three observed random variables $X,Y,Z\in\mathbb{R}$, 
	and two unobserved random disturbances $\upeta\in\mathbb{R}$ and 
	$\bm{\epsilon}:=(\epsilon_1,...,\epsilon_d)\in\upchi^d\subset\mathbb{R}^{d}$. The triplet $(Z,\bm{\epsilon},\upeta)$ generates a pair of observed quantities $X,Y\in\mathbb{R}$, as an infinite Gaussian mixtures governed by unknown location as well as variance functional parameters $(\mu_{1,i}, \mu_{2,i})$ and $(\sigma_{1,i}, \sigma_{2,i})$. There are unknown mixing proportion parameters $(\alpha_{1,i}, \alpha_{2,i})$.  Let $L$ be the number of components in the Gaussian mixtures of $(X,Y)$, which approaches infinity, depicted as,
	\begin{align}
		& X = h(z,\eta) \sim \sum_{i=1}^{L} \alpha_{1,i}\mathcal{N}\left(\mu_{1,i}(z) + \rho_{12}\sigma_{1,i}(z)\left(  \frac{\eta-\mu_{\eta}}{\sigma_\eta}\right), \sigma_{1,i}^2(z)(1-\rho_{i,12}^2)\right) \label{X:Structural}
	\end{align}
	and $Y = g(X,\upepsilon)$ satisfying,
	\begin{align}
		& g(X,\upepsilon)|\upeta=\eta\sim \sum_{i=1}^{L} \alpha_{2,i}\mathcal{N}\left(\mu_{2,i}(X) + \rho_{i,13}(X)\sigma_{2,i}(X)\left(  \frac{\eta-\mu_{\eta}}{\sigma_\eta}\right), \sigma_{2,i}^2(X)(1-\rho_{i,13}^2(X))\right)\label{Y:Structural}.
	\end{align}	
	The random variable $Y$ is the outcome variable, whereas $X$ represents an action or an endogeneous variable. %The entire process can only be characterized up to a normalization of $\upeta$.
	The random variable $Z$ is an exogenous  covariate also known as an instrumental variable.
	The random variable $Z$ is assumed to be continuous with an unknown density $f_Z(z)$,
	and is further assumed to satisfy  
	$Z\independent (\upeta,\bm{\epsilon})$. 
	The disturbances $(\upeta,\bm{\epsilon})$ are assumed to be continuous and possibly jointly dependent. 
	We denote their unknown joint and marginal distributions 
	by $F_{\upeta,\bm{\epsilon}}$, $F_\upeta$ and $F_{\bm{\epsilon}}$, respectively. 

Equivalently, the Gaussian Cholesky decomposition parameterized by the functional parameters $\left\{\upbeta_i\right\}$ with $\upbeta_i(x):=\rho_{i,13}^2(x)$  and a constant $\uptheta_i:=\rho_{12}^2$ can be used to simplify the analytical formula of the data Gaussian mixture data generation process in Eqs. \eqref{X:Structural} and \eqref{Y:Structural},
% ,the mixture Gaussian data generation process nests an infinite diffusion models \citep{ho2020denoising} 
\begin{align}
&	X = \sum_{i=1}^{L}\alpha_{1,i}\left\{\mu_{1,i}(Z) + \sigma_{1,i}(Z)\left(\sqrt{1-\uptheta_i}  \Phi^{-1}(\bm\nu_{i,x}) + \sqrt{\uptheta_i}\left(  \frac{\upeta-\mu_{\eta}}{\sigma_\eta}\right)\right)\right\},\nonumber\\
&	Y = \sum_{i=1}^{L}\alpha_{2,i}\left\{\mu_{2,i}(X) + \sigma_{2,i}(X)\left(\sqrt{1-\upbeta_i(X)}  \Phi^{-1}(\bm\nu_i) + \sqrt{\upbeta_i(X)}\left(  \frac{\upeta-\mu_{\eta}}{\sigma_\eta}\right)\right)\right\},\nonumber
\end{align}
with independent uniformly distributed random variables $\bm\upnu_i$ and $\bm\upnu_{i,x}$.
We define $\alpha_{i,j}^*:=\alpha_{1,i}\alpha_{2,j}$ to express the observed outcome as a function of $(X,Z)$,% in which each component nests a diffusion model \citep{ho2020denoising},
\begin{align}
	Y = \sum_{\substack{i,j=1}}^{L}\alpha_{i,j}^*\left\{\mu_{2,i}(X) + \sigma_{2,i}(X)\left(\sqrt{1-\upbeta_i(X)}  \Phi^{-1}(\bm\nu_i) + \sqrt{\upbeta_i(X)}\left(  \frac{\Psi_{j}(X,Z)+\bm\epsilon_0-\mu_{\eta}}{\sigma_\eta}\right)\right)\right\}, \label{Y:ReducedForm}
\end{align}
where $\Psi_{j}(X,Z)$ for any given noise $\bm\epsilon_0$ is just a one of the infinite potential values of $\upeta$ in the $j$'th component of the mixture for any given $(X,Z)$. Eq. \eqref{Y:ReducedForm} is a reduced-form representation of $Y$ without any endogenous predictor variables. This formulation shares similarities with diffusion models \citep{ho2020denoising} by capturing the influence of noise on $Y$.
%    Let $\mathcal{B}$ be the Borel $\sigma$-algebra in $\mathbb{R}$. Denote $\mu_{Q}$ the distribution of $Q$, which is the probability measure $\mu_{Q}:\left\{\mathcal{A}\in\mathbb{R} \left|Q^{-1}(S)\in\mathcal{F} \right.\right\}\mapsto [0,1]$  $\mathcal{B}\subset\mathcal{P}(\mathbb{R})\mapsto [0,1]$ a probability measure on $\mathbb{R}$. A map assigning to any set $\mathcal{A}\in\mathcal{B}$ the probability of the corresponding event in $\mathcal{A}$. 

We attend to the following definition of the support of a continuous random vector $Q$ in $\mathbb{Q}^n$ with a probability density function $f_{Q}$,	% the support of a random variable is the support of its probability measure defined by $\mu_{Q}:\left\{\mathcal{A}\in\mathbb{R} \left|Q^{-1}(S)\in\mathcal{F} \right.\right\}\mapsto [0,1]$
    \begin{definition}[Support]
    	$\text{Supp}(Q):=\text{The closure of the set of points for which } f_{Q}(q)\ne 0$.
%    	$\text{Supp}(Q):=\left\{q\in\mathbb{R}^n\left| \mu_{Q}\left(\mathbb{B}(q,r) > 0\right)\quad\forall r > 0\right.  \right\}$, such that $\mathbb{B}(q,r)$ stands for a ball centered at $q$ with a radius $r$.
    \end{definition}
	We consider a non-parametric setting, whereby the unknown function $h$ 
	is assumed to be continuous, but otherwise of general structural form. The unknown function
	$g$ can be continuous or piecewise constant, in the latter case leading to a discrete  outcome $Y$. 
	%	and $g$ are assumed to be continuous but otherwise of general structural form. 
	As in \cite{{imbens2009identification}}
	the functions $h$ and $g$ may be nonseparable in their respective disturbances $\upeta$ and $\bm{\epsilon}$. 
	The model (\ref{X:Structural})-(\ref{Y:Structural}) implies that $\upeta$ is a \textit{confounder} or common cause, affecting both $X$ directly via \eqref{X:Structural}, 
	and $Y$ indirectly through the dependence of
	$\bm{\epsilon}$ on $\upeta$.
	Finally, following previous work \citep{imbens2009identification}, we assume the following standard condition holds: 
	\begin{asu}[Common support]\label{asu:common}
		The conditional random variable $\upeta|X=x$ has a \textbf{fixed support} independent of $x$.
		Namely, 
		$\text{Supp}(\upeta|X=x) = \text{Supp}(\upeta),$ $\forall x$. 
	\end{asu}
	%\textcolor{red}{	
		%	\begin{asu} \label{asu:indep}
			%		The random variables $X$ and $\upepsilon$ are conditionally independent given $\upeta$, namely 
			%		$X\independent \bm\epsilon|\upeta$. This implies that $\upeta$ is the only source of joint dependence between $X$ and $\bm\epsilon$.
			%	\end{asu}	

	For future use, we present two important properties of the above triangular model. 
	First, note that Eq. (\ref{X:Structural}) together with the assumption that $Z$ is independent of $(\upeta,\bm \epsilon)$ imply that $X\ci \bm\epsilon|\upeta$. The second property is
	stated in the following auxiliary lemma (see proof in appendix \ref{Appendix:Proof}). 
\begin{lemma}\label{lemma:eps:support}
	Under assumption \ref{asu:common} it holds that,
	\begin{align}
		\text{Supp}(\left.\bm\epsilon\right|X=x)=\text{Supp}(\bm\epsilon) \quad \forall x\in\text{Supp}(X).\label{eq:Eps:Support}
	\end{align}	
\end{lemma}
		Given $n$ 
		i.i.d. triplets $(x_i,y_i,z_i)$ from the above triangular model, Eqs. (\ref{X:Structural})-(\ref{Y:Structural}), the goal is to estimate the counterfactual distribution of the outcome $Y$ if $X$ were set to the value $x$.
		Depending on the particular application, this quantity is related to policy effect or treatment effect. 
		In the notation of causal inference \citep{pearl2019interpretation}, the above quantity of interest is known as \textit{intervention}. It is given by   
		\begin{align}
			T(y|x) = P\left(Y\le y|\text{do}(X=x)\right):= \mathbb{E}_{\bm\epsilon\sim F_{\bm\epsilon}}\left[\mathds{1}\left\{g(x,\bm\epsilon)\le y\right\}\right].
			\label{eq:do_X}
		\end{align}
		It is worth emphasizing that under the above triangular model and its assumptions, even exact knowledge of
		the joint distribution of $(X,Y,Z)$ does not uniquely determine the quantity $T(y|x)$ \citep{imbens2007nonadditive}. 	%Note that the quantity $T(y|x)$ cannot be directly estimated from the observed triplets $(x_i,y_i,z_i)$. 
		The reason is that even if $g$ were known, under (\ref{X:Structural})-(\ref{Y:Structural}) the random variables $X$ and $\bm\epsilon$ might be jointly dependent. 
		In contrast, in Eq. (\ref{eq:do_X}) the expectation is over the marginal distribution of $\bm\epsilon$, effectively treating $\bm\epsilon$  as independent of $X$. 
		
		\begin{comment} 
		\textcolor{red}{
		1. %Clearly, even with perfect knowledge of the joint distribution of $(X,Y,Z)$ the quantity $T(y|x)$ is not well posed CITE
		\\
		%		It is worth emphasizing that the triangular model with exact knowledge of
		the joint distribution of $(X,Y,Z)$ does not uniquely determine the quantity $T(y|x)$. 
		\\
		%        2. As we describe below, the existing literature imposed various additional assumptions that lead to well posedness of the problem of estimating $T(y|x)$.
		\\   
		3. Most works imposed the following two standard conditions, which we shall also assume hold in our work. Describe assumptions 3.1 3.2, 
		\\
		4. These two assumption are still not sufficient for identifiability. EXAMPLE / CITATION ? 
		\\
		5. previous works additional assumptions sufficient for identifiability. $H_c$, control variable...
		\\
		6.         
		}
		\end{comment}

		To guarantee identifiability of $T(y|x)$, several previous works imposed an additional assumption
		of {\em strict monotonicity} of $h(z,\eta)$ in the second argument $\eta$, required to hold
		for all $z$ in its support, e.g., \cite{blundell2014control,imbens2009identification}. 
		Specifically, \citet{imbens2009identification}
		defined the following random variable $V:=F_{X|Z}(X)$
		and the following functional, both of which can be estimated from the observed data, 
		\begin{align}
			H_c(y|x) = \int_{0}^1 F_{Y|X=x,V=v}(y)dv
			\label{eq:H_c}.
		\end{align}
		%By definition, the conditional random variable $V|Z=z\sim U[0,1]$ is uniform for any $z$, and thus $V\sim U[0,1]$ is also uniform. 
		The strict monotonicity assumption implies that the random variable $V$ is a
		one-to-one mapping of the unobserved disturbance $\upeta$. Furthermore, it can be shown that
		$V$ is a {\em control variable} and that $H_c(y|x) = T(y|x)$. Then,  an estimate of $H_c$ directly yields an estimate of $T(y|x)$. 
		
		The identifiability strategy described above is based on the ability to uniquely determine $\upeta$,
		up to a strictly monotonic transformation, for any pair of observed $(X,Z)$. As we now describe, we significantly depart from this approach. Instead, we leverage on the known axiom of continuity and building block of expected utility theory and post the following assumption:\footnote{We are reminded that continuity is an essential part of the axiom of expected utility theory guaranteeing non-intersecting indifference curves. See  \cite{afriat1967} for the well-known result showing that data cannot be treated as being generated by a utility function, if there are large deviations from rationalizability.  \mycite{afriat1967} theorem tells us that continuity is utmost necessary for any finite data set to be rationalizable. } %make the following assumption, which is weaker than strict monotonicity. %Further, we show that	our approach generalizes the strict monotonicity assumption, in the sense that if strict monotonicity holds, then these two assumptions hold as well. 
		%
		%these assumptions are automatically satisfied under strict monotonicity and hence are a generalization of it. 
		
		%%%%%%%%%%%%%% ASSUMPTION %%%%%%%%%%%%%%%%%%%%\
	\begin{asu}[Continuity of $X|Z$ and $\bm{\epsilon}|\upeta$]
	\begin{subasu}\label{asu:continuity}
		For any $z\in \text{Supp}(Z)$, the random variable $X|Z=z$ is continuous. In particular, the set of values of $\upeta$ where $\partial h(z,\eta)/\partial \eta=0$ is of measure zero.				
	\end{subasu}
	\begin{subasu}\label{asu:cond:continuity}
		The random vector $\bm{\epsilon}|\upeta$ is a $d$-dimensional continuous random vector for any $\upeta$.		
	\end{subasu}			
\end{asu}

\begin{asu}[Completeness]
	The distribution of the random vector $(X,\upeta)|Z$ belongs to a complete family of distributions with respect to $Z$.
\end{asu}\label{asu:bounded}
\begin{asu}[Square integrability]\label{asu:Integrability}
	The density $f(y|X=x,Z=z)$ is square integrable, and is known.
\end{asu}

		%Assumption \ref{asu:continuity} implies that changes in $\eta$ affect the resulting action $X$. Under strict monotonicity of $h$ in $\upeta$, the action $X$ must either increase or decrease as $\upeta$ changes for any fixed value of $Z$. Thus, strict monotonicity implies this assumption.%\ref{asu:continuity} implies continuity without restricting the function to be monotonic. 
		%%%%%%%%%%%%%% ASSUMPTION %%%%%%%%%%%%%%%%%%%%
		\begin{comment}
		\begin{asu}
		For any function $u(\upeta,x)$ with finite expectation with respect to $\upeta$, if 
		\begin{equation}
		\label{eq:E_w_eta_x}
		\mathbb{E}_{\upeta|X=x,Z=z}\left[u(\upeta,x)
		\right]=0  \quad \forall (x,z)\in\text{Supp}(X,Z),
		\end{equation}
		then $\mathbb{E}_{\upeta}\left[u(\upeta,x)\right]=0$ for all $x\in \text{Supp}(X)$.
		\end{asu}
		\end{comment}		
		%\textcolor{red}{		Denote $h\in\mathcal{C}_x$ and $h^{\dagger}\in\mathcal{C}_x$ the true (normalized) and synthetic generator functions of the random variables $X$ and $X^{\dagger}$, respectively. Denote the conditional densities of $\upeta$ given any pair $(x,z)\in\text{Supp}(X,Z)$ induced by $h$ and $h^{\dagger}$, respectively,}
		%\begin{align} 
		%\pi_{x,z}^*(\eta)= \Pr[\upeta=\eta\,|\,X=x,Z=z] \quad \text {and}\quad  \pi_{x,z}^{\dagger}(\eta)= \Pr[\upeta=\eta\,|\,X^{\dagger}=x,Z=z]. 
		%\pi_{x,z}^*(\eta)= f_{\upeta|X=x,Z=z}(\eta) \quad \text {and}\quad  \pi_{x,z}^{\dagger}(\eta)= f_{\upeta|X^{\dagger}=x,Z=z}(\eta). 
		%\label{eq:pi:dagger:definition}
		%\end{align}    
	{
		
		The first key contribution of our work is to show that under the 
		triangular model with assumptions \ref{asu:common}-\ref{asu:Integrability}, 
		the intervention $T(y|x)$ is identifiable.
	Consequently, the fact that additive and monotonic triangular models are identifiable, are special cases of our more general identifiability result  \citep{imbens2007nonadditive,imbens2009identification,blundell2014control,heckman1985alternative}.

	We establish identifiability in a population setting, where assume to have observed an infinite number of triplets $(x_i,y_i,z_i)$ from the triangular model \eqref{X:Structural}-\eqref{Y:Structural}. Hence, in what follows we assume that the joint distribution 
	$F_{X,Y,Z}$ as well as various conditional distributions, such as $F_{Y|X=x,Z=z}$ are all
	perfectly known. Our approach to prove identifiability proceeds as follows. First, 
	in Section \ref{Sec:integral}, we derive an 
	integral equation relating the intervention 
	of Eq. (\ref{eq:do_X}) to known distributions
	of $(X,Y,Z)$. In general, there may be an infinite number of solutions to this integral equation, which reflects the fact that in the original triangular model, the unknown functions $h$ and $g$ are not identifiable. Yet, in Section \ref{Sec:Identification} we prove that {\em any}
	solution of this integral equation gives the same intervention, hence proving identifiability of the intervention.

		%%%%%%%%%%%%%%%%%%%%%%%%%%%%%%%%%%%%%%%%%%%%%%%%%%%%%%%%%%%%%%%%%%%%%%%%%%
		%%%%%%%%%%%%%%%%%%%%%%%%%%%%%%%%%%%%%%%%%%%%%%%%%%%%%%%%%%%%%%%%%%%%%%%%%%

		\section{An alternative representation for the intervention}\label{Sec:alternative_representation}
		
		A key preliminary step is to derive an {\em alternative} representation
for the intervention $T(y|x)$. 
	{Consider the following cumulative distribution function (CDF),}  
		\begin{equation}
			H(y|x) =  \mathbb{E}_{\upeta}[  F_{Y|X=x,\upeta}(y)] =  
			\int F_{Y|X=x,\upeta=\eta}(y)f_{\upeta}(\eta)d\eta.
			\label{eq:def_H}
		\end{equation}
	Note that the function inside the integral , with $\upeta$ attaining all possible values in its support, is well defined by assumption \ref{asu:common}. 	%
	The following lemma is key to our proposed approach, both for proving identifiability and for inference. 
		\begin{lemma}[Functional equivalence]
			\label{lem:H_T}
			Under the triangular model with assumption \ref{asu:common}, 
			the function $H(y|x)$ defined in Eq. \eqref{eq:def_H} satisfies that 
			$\forall\hspace{0.2em}(x,y)$ in the relevant support,
			\begin{eqnarray}
				H(y|x) = T(y|x).
				\label{eq:H_T}
			\end{eqnarray}	
		\end{lemma}
		
		\begin{proof}
			Using the definition of $Y$ in Eq. \eqref{Y:Structural} the function 
			$F_{Y|X=x,\upeta}(y)$ may be equivalently written as follows
			\begin{equation}
				F_{Y|X=x,\upeta}(y) = \mathbb{E}_{\bm{\epsilon}|X=x,\upeta}\left[\mathds{1}\left\{Y\le y\right\}|X=x,\upeta\right] = \mathbb{E}_{\bm{\epsilon}|X=x,\upeta}\left[
				\mathds{1}\left\{g(x,\bm{\epsilon})<y\right\}\right].
			\end{equation}
			By conditional independence $X\ci\epsilon|\upeta$, the random vector $\bm{\epsilon}|X=x,\upeta$ has the same distribution as 
			$\bm{\epsilon}|\upeta$. Inserting the resulting expression back into Eq. \eqref{eq:def_H}
			gives that
			\begin{equation}		
				H(y|x) = \mathbb{E}_{\upeta\sim F_{\upeta}}\left[\mathbb{E}_{\bm{\epsilon}|\upeta}\left[\mathds{1}\left\{g(x,\bm{\epsilon})\le y\right\}\right]\right].
			\end{equation}
			By the law of total expectation, the RHS is simply
			$\mathbb{E}_{\bm \epsilon}[\mathds{1}\left\{g(x,\bm\epsilon)\le y\right\}]$, which
			by Eq. \eqref{eq:do_X} is $T(y|x)$. %$\hfill\blacksquare$
		\end{proof}
		
		By Lemma \ref{lem:H_T}, instead of estimating $T(y|x)$ which involves a $d$-dimensional integration over $\bm\epsilon$, one may instead estimate the function $H(y|x)$ that depends on a univariate integral w.r.t. $\upeta$.  
The specific expectation operator in Eq. (\ref{eq:def_H}) is known as the \textit{partial means} estimator \citep{newey1994kernel}, since the expectation is taken w.r.t. the unconditional distribution $F_{\upeta}$ rather than the conditional one $F_{\upeta|X=x}$. 
As such, $H(y|x)$ is a {\em counterfactual} distribution.

\section{An Example for a Triangular Data Generation Process}\label{sec:example}

Let's examine the non-monotonic case, $h(z,\eta)\ne F_{X|Z=z}^{-1}(F_{\upeta}(\eta))$.

\begin{align}
	& \bm\Sigma_{y,\eta}(X,Z):=\left[\begin{matrix}
		\sigma_y^2(X) & \sigma_2(x)\sigma_{\eta}(X,Z)\rho_{13}(X) \\
		\sigma_2(x)\sigma_{\eta}(X,Z)\rho_{13}(X) & \sigma_{\eta}^2(X,Z)
	\end{matrix}\right]\nonumber
\end{align}

The coefficient $\rho_{13}(x)$ determines the conditional correlation between $Y$ and $\upeta$ given $X=x, Z=z$.

\begin{align}
	& Z \sim F_Z, \quad \upeta\sim N(\mu_{\eta}, \sigma_{\eta}^2) \nonumber \\
	& X|Z=z, \upeta=\eta\sim N\left(\mu_1(z) + \rho_{12}\sigma_1(z)\left(  \frac{\eta-\mu_{\eta}}{\sigma_\eta}\right), \sigma_x^2(z)(1-\rho_{12}^2)\right)\nonumber\\
	& Y|X=x,Z=z\sim N\left(\mu_2(x) + \rho_{13}(x)\sigma_2(x)\left(  \frac{\eta(x,z)-\mu_{\eta}}{\sigma_\eta}\right), \sigma_y^2(x)(1-\rho_{13}^2(x))\right)\nonumber
\end{align}
Recall that $\upeta\ci Z$. Yet, $\upeta\not\ci Z|X$. Similarly, $Y\ci Z|X,\upeta$.
In the equation of $X$, each value of $Z$ characterizes a specific conditional Gaussian (given $\upeta$). In the equation of $Y$, each value of $X$ characterizes a specific conditional Gaussian (given $\upeta$).
We employ Gaussian Cholesky decomposition to simplify the analytical formula of the data generation process,
\begin{align}
	& X(z,\eta, \bm\varepsilon_x) = \mu_1(z) - \rho_{12}\sigma_1(z)\frac{\mu_{\eta}}{\sigma_\eta} +   \sigma_1(z)\rho_{12}\frac{\upeta}{\sigma_\eta} + \sigma_1(z)\sqrt{1-\rho_{12}^2}  \Phi^{-1}(\bm\varepsilon_x) \nonumber\\
	& Y(x,z,\bm\nu) = \mu_2(x) - \rho_{13}(x)\sigma_2(x)\frac{\mu_{\eta}}{\sigma_{\eta}} +   \sigma_2(x)\rho_{13}(x)\frac{\upeta(z,x)}{\sigma_{\eta}} + \sigma_2(x)\sqrt{1-\rho_{13}^2(x)}  \Phi^{-1}(\bm\nu)\label{Diffused_outcome}
\end{align}
\begin{align}
	&\upeta\frac{\sigma_1(z)\rho_{12}}{\sigma_\eta} = x - \mu_1(z) + \rho_{12}\sigma_1(z)\frac{\mu_{\eta}}{\sigma_\eta}  - \sigma_1(z)\sqrt{1-\rho_{12}^2}  \Phi^{-1}(\bm\varepsilon_x) \nonumber
\end{align}
\begin{align}
	&\upeta =\mu_{\eta}  + \frac{\sigma_\eta}{\sigma_1(z)\rho_{12}}(x  - \mu_1(z)) - \frac{\sigma_{\eta}\sqrt{1-\rho_{12}^2}}{\rho_{12}} \Phi^{-1}(\bm\varepsilon_x) \nonumber
\end{align}
\begin{align}
	& \sigma_{\eta}(x,z):=\frac{\sigma_{\eta}\sqrt{1-\rho_{12}^2}}{\rho_{12}} \nonumber \\
	& \mu_{\eta}(x,z):=\mu_{\eta} +  \frac{\sigma_{\eta}}{\sigma_1(z) \rho_{12}}\left(x-\mu_1(z)\right) \quad \mu_{\eta}=0 \quad (\text{a normalization}) \nonumber \\
	& \upeta(z,x) = \mu_{\eta}(x,z)- \sigma_{\eta}(x,z)\Phi^{-1}(\bm\varepsilon_x). \nonumber
\end{align}
where $\upeta(x,z):=\upeta|X=x, Z=z \sim N\left(\mu_{\eta}(x,z), \sigma_{\eta}^2(x,z)\right)$.
The unknown parameters characterizing the latent confounder $\upeta$ are: $(\rho_{12}, \rho_{13}(\cdot), \mu_{\eta}, \sigma_{\eta})$. The rest of the parameters can be directly measured from the joint distribution of the observables $(X,Y,Z)$.

Identifying $\rho_{12}$ from the covariance with observable.
\begin{align}
	& \mathbb{E}_{\upeta}\left[\mathbb{COV}\left[X,Z\right]|\upeta\right] = \mathbb{COV}\left[\mu_1(z), Z\right] + \sqrt{1-\rho_{12}^2} \mathbb{COV}\left[\sigma_1(z) \Phi^{-1}(\bm\varepsilon_x), Z\right] \nonumber
\end{align}
This gives,
\begin{align}
	& \sqrt{1-\rho_{12}^2}=\frac{\mathbb{COV}\left[X,Z\right]-\mathbb{COV}\left[\mu_1(z), Z\right] }{\mathbb{COV}\left[\sigma_1(z) \Phi^{-1}(\bm\varepsilon_x), Z\right]} \nonumber
\end{align}
Given $\rho_{12}$, the coefficient $\rho_{13}(x)$ can be identified by a linear regression of the form,
\begin{align}
	& Y(x,z,\bm\nu) = \overbrace{\mu_2(x)}^{\mathbb{E}[Y|X=x]} - \rho_{13}(x)\left[  \frac{\sigma_2(x)}{\sigma_1(z)\rho_{12}}\overbrace{(x-\mu_1(z))}^{x-\mathbb{E}[X|Z=z]}\right] + \epsilon \nonumber
\end{align}
The counterfactual outcome variable
\begin{align}
	& Y^{\text{CF}}(\text{Do}(X=x)) = \mu_2(x) + \rho_{13}(x)\sigma_2(x)\left(  \frac{\eta-\mu_{\eta}}{\sigma_\eta}\right) + \sigma_2(x)\sqrt{1-\rho_{13}^2(x)}  \Phi^{-1}(\bm\nu) \label{Diffused_CF_outcome}		
\end{align}
Define $\upbeta:=\rho_{13}^2(x)$. For any given pair $(x,z)$, $Y$ can be reconstructed from a noise propagated by $\bm\nu$,
\begin{align}
	& Y(x,z) = \mu_2(x) + \sqrt{\upbeta}\sigma_2(x)\left(  \frac{\eta(x,z)-\mu_{\eta}}{\sigma_\eta}\right) + \sigma_2(x)\sqrt{1-\upbeta}  \Phi^{-1}(\bm\nu). \label{Diffused_CF_outcome:beta}		
\end{align}
The non-contaminated output is defined as follows,
$$Y_0:=\Phi^{-1}(\bm\nu).$$ The observed output is a contaminated variant of $Y_0$ is known as a diffusion process \citep{ho2020denoising},\footnote{Diffusion models excel at generating high-quality, diverse data for various applications, including creative content generation. They achieve this by adding controlled noise to the data and subsequently learning to reverse the process, effectively capturing underlying structure in datasets.}
% the structure of observed $Y$ in Eq. \eqref{Diffused_outcome} is known 
\begin{align}
	\frac{Y(x,z)-\mu_2(x)}{\sigma_2(x)}:=\sqrt{1-\upbeta}Y_{0}+\sqrt{\upbeta}\left(\frac{\eta(x,z)-\mu_{\eta}}{\sigma_\eta}\right).\label{eq:Y_t} 
\end{align}
 %+ \frac{\sqrt{\upbeta}}{\sqrt{1-\beta}}\left(\frac{\eta-\mu_{\eta}}{\sigma_\eta}\right)
It can be seen that the posterior distribution of $\eta$ given $x$ and $z$ does not play any role in $Y_0$. In the present diffusion model we only need to recover the parameters $(\mu_2(x), \sigma_2(x))$ characterizing $Y(x,z)$ for any give pair $(x,z)$. This requires to identify the nuisance parameters $\beta$ and $\rho_{12}$ as depicted above. Eq. \eqref{Diffused_CF_outcome} gives us the direct linkage between the identified parameters and the counterfactual output. Namely, by substituting $\eta(x,z)$ in Eq. \eqref{eq:Y_t}  with a random draw from the prior distribution of $\upeta$ we get a counterfactual sample from $Y^{\text{CF}}$.
%The only dependence structure between $Y^{\text{CF}}$ and $\upeta$ for a given $X$ is  functional dependence through $\mu_2(x)$, $\rho_{13}(x)$ and $\sigma_2(x)$, without updating the posterior distribution of $\upeta$.

The next building block is formulating a two-regime model describing the relationship among members in a complete family of distributions, facilitating generalizing
 the identifiability achieved in the simple example above.
\section{Completeness of Sufficient Statistic}\label{sec:completeness}
		
		\begin{definition}[Completeness]
		Let $X$ be a random variable with density functions parameterized by $\theta$. A statistic $T(X)$ is said to be complete if, for every measurable function $g$, the following holds:
		\[ \mathbb{E}[g(T(X))] = 0 \text{ for all } \theta \implies g(T(X)) = 0 \text{ almost everywhere.} \]
		\end{definition}		
		In simpler terms, if the expected value of any function of the statistic is zero for all possible values of the parameter $\theta$, then the function itself must be almost everywhere equal to zero.

		%\section{Strong Completeness of Sufficient Statistic}
		Now, we attend to a stronger variant of completeness, bridging between regimes nested in the same model.
		
%		\begin{informaldefinition}[Strong completeness]
	\textbf{Informal Definition (Strong completeness)}.
		The concept of strong completeness applied to the relationship between regimes (subsets of a parameter space), introduced by \cite{alamatsaz1983completeness}, asserts the invariance of expected values of functions with respect to a parameter in some regime (for some fixed setting of the remaining parameters), implying invariance with respect to this parameter in any other regime (other fixed settings of the remaining parameters). % ($\theta_2^{\prime}\in\Theta_2$) ($\forall \theta_2^{\prime\prime}\in\Theta_2$).
%		\end{informaldefinition}		
		\begin{figure}[H]{Two Mutually Exclusive Regimes Nested in a Single Model}
	% $\left\{(\theta_1,\theta_{2,0})\in\Theta_1\right\}$
	\begin{center}
		\begin{tikzpicture}[scale=1.5]
			% Circles
			\draw[fill=blue, opacity=0.3] (0,0) circle (1);
			\draw[fill=green, opacity=0.3] (2,0) circle (1);
			\draw[fill=yellow, opacity=0.2] (1,0) circle (2);			
			% Labels
			\node[align=center, text=black] at (0,0) {Regime A \\ $\theta_2=\theta_2^{\prime}$ \\(\small randomized)};
			\node[align=center, text=black] at (2,0) {Regime B \\ $\theta_2=\theta_2^{\prime\prime}$\\(\small observational) };
			\node[align=center, text=black] at (1,-1.5) {A Two-regime Model\\ $\theta_1\in\Theta_1$};
			
			% Arrows
			%			\draw[->, thick] (0.5,0.5) -- (1.5,0.5);
			\draw[->, thick] (1.5,-0.5) -- (0.5,-0.5);
		\end{tikzpicture}
	\end{center}		
\end{figure}	
In the figure above we present a two-regime model parameterized by $(\theta_1,\theta_2)$ belonging to a strongly complete family of distributions with respect to $\theta_1$. Under strong completeness, family of regime $B$ (each member is a different realization of $\Theta_1$) is informative for policy makers that want to infer about family of regime $A$ because these two families share a common parameter space $\theta_1$. We next attend to a formal definition of strong completeness.
		\begin{definition}[Strong Completeness]
		A family of distributions $F:=\left\{ F_{Y|\theta_1, \theta_2} : \theta_1\in \Theta_1\right\}$ parameterized by $(\theta_1, \theta_2)$ of $d$-dimensional random vectors is said to be strongly complete with respect to $\theta_1$, if for every function $m:\mathbb{R}^d \to\mathbb{R}$ satisfying,
		\[ \mathbb{E}_{Y\sim F_{Y|\theta_1, \theta_2}}\left[m(\bm Y)\right] = 0 \quad \forall\theta_1 \in \Theta_1, \]
		we have that,
		\[ \mathbb{E}_{ Y\sim F_{Y|\theta_1, \theta_2}}\left[\mathds{1}\left\{m(\bm Y)=0\right\}\right] = 1 \quad \forall(\theta_1,\theta_2) \in \Theta_1\times \Theta_2. \]
	    \end{definition}
		Consequently, in the presence of a model related to a strongly complete family of distributions with respect to common parameters $\theta_1\in\Theta_1$, invariance to these common parameters in one regime $\theta_2=\theta_2^{\prime\prime}$ holds also in the rest of regimes $\theta_2^{\prime}\in\Theta_2$. This enables to inform policy decisions in different real-world settings (regime $A$).

		%%%%%%%%%%%%%%%%%%%%%%%%%%%%%%
		%%%%%%%%%%%%%%%%%%%%%%%%%%%%%%%
		\section{Identifiability}\label{Sec:Identification} 
		In this section we prove the identifiability of the intervention $T(y|x)$. This is done in two steps. First, we alleviate the issue of the triangular model dependence on multiple unknown quantities. The most prominent one being the dependence of the disturbances on the endogenous covariate. This is implemented by employing an equivalent representation for the structural functions (lemmas \ref{Lemma:Decoupling} and \ref{lemma:Invariant:Do} to follow) in a canonical form, which is invariant to the specification of the unknown quantities (joint-distribution of the disturbances). Second, we use  canonical form representation to characterize an equivalence class of normalized-triangular models and establish that identifiability is universal in that it holds for any structural continuous function in the equivalence class (theorem \ref{Theorem:Identifiability}).

		%%%%%%%%%%%%%%%%%%%%%%%%%%
		\subsection{Normalized triangular models via error-decoupling mechanism}
		
		%\end{document}
		%%%%%%%%%%%%%%%%%%%%%
%		A major identifiability challenge in the original model is related to the fact that the joint and marginal distributions of disturbances are unknown. Nevertheless, this challenge can be alleviated by adapting a canonical form representation for the triangular model in Eqs. \eqref{X:Structural} and \eqref{Y:Structural}, which is invariant to the specification of the aforementioned unknown distributions. This strategy is referred to as a normalization \citep{schennach2014entropic}.
		%		\textcolor{red}{Move here discussion of the fact that in original model distributions of disturbances are unknown. Explain that normalized means that the distribution of disturbances is known, in fact it is uniform. Give motivation for the lemma that is presented}
		
		In what follows, ${X}^*\overset{d}{=}X$ indicates that the two random variables $X^*$ and $X$  have the same distribution. Using this notation, the triangular model is reformulated via an error-decoupling mechanism
		with disturbances that are uniformly distributed.
		To this end, given random variables $\upomega$ and $\nu_1,\ldots,\nu_d$, define the following transformed random variables,
		\begin{equation}
			\Lambda_0 = \Lambda_0(\upomega):=F_{\upeta}^{-1}(\upomega),
			\label{eq:Lambda_0}
		\end{equation} 
		and recursively for $1\le j\le d$, 
		\begin{equation}
			\Lambda_j = \Lambda_j(\upomega,\nu_1,...,\nu_j):=
			F_{\epsilon_j|(\Lambda_0(\upomega),\Lambda_1(\upomega,\nu_1),\ldots,\Lambda_{j-1}(\upomega,\nu_1,\ldots,\nu_{j-1}))}^{-1}(\nu_j).
			\label{eq:Lambda_j}		
		\end{equation}
		The notations in Eqs. \eqref{eq:Lambda_0}-\eqref{eq:Lambda_j} are used to perform error-decoupling in the following lemma (see proof in appendix \ref{Appendix:Proof}):
		\begin{lemma}
			\label{Lemma:Decoupling}(Error-decoupling)
			Let $\upomega\sim U[0,1]$ and $\bm{\nu}:=(\nu_1,\ldots,\nu_d) \sim U[0,1]^d$
			be independent of $\upomega$. 
			Consider the following triangular model,
			where $\Lambda_0$ and $\Lambda_j$ are defined in Eqs. (\ref{eq:Lambda_0}) and (\ref{eq:Lambda_j}),  
			\begin{eqnarray}
				X^*&=&h^*(Z,\upomega):=h\left(Z,\Lambda_0(\upomega)\right)
				\label{eq:tilde_X},\\
				Y^*&=&g^*(X^*,\upomega,\bm{\nu}):=g\left(X^*,\Lambda_1(\upomega,\nu_1),\ldots,\Lambda_d(\upomega,\nu_1,\ldots,\nu_d)\right). \label{eq:tilde_Y}
			\end{eqnarray}
			Then, $(X^*, Y^*, Z) \overset{d}{=}(X,Y,Z)$.
		\end{lemma}						
		There are two major differences
		between the normalized triangular model Eqs.  \eqref{eq:tilde_X}-\eqref{eq:tilde_Y}
		and the original Eqs. \eqref{X:Structural}-\eqref{Y:Structural}. 
		The first difference is that the disturbances $\upomega$ and $\bm \nu$ are {\em independent}
		and in fact uniformly distributed. 
		The second difference, or the price to pay for this decoupling, is that now the disturbance $\omega$ appears explicitly both in the equations for $X$ and for $Y$, namely inside the functions $h^*$ and $g^*$. Nevertheless, this price is negligible as it alleviates a major challenge in the original triangular model of Eqs. \eqref{X:Structural}-\eqref{Y:Structural}. This challenge stems from the fact that $\upeta$ and $\bm \epsilon$ may in general be dependent, in an unknown fashion. 
		%In contrast, in the new model the unknown quantity $F_{Y|X=x,\upeta=\eta}(y)$ can be represented independently of the joint distribution functions %$F_{\bm\epsilon,\upeta}$ and $F_{\bm\epsilon,X}$ as follows,% since 
		%\[
		%F_{Y|X=x,\upeta=\eta}(y)=\int_{[0,1]^d} \mathds{1}\left\{g^*(x,F_{\upeta}(\eta),\bm\nu)\le y\right\} d\bm\nu.
		%\]	 
		%This is so because Eqs. \eqref{eq:tilde_X}-\eqref{eq:tilde_Y}, which  yield the same distribution for the observables, where 
		%$\upeta$ is replaced by $\omega$ which is distributed $U[0,1]$. More importantly, the unknown dependence between $(\upeta,\bm\epsilon)$ is now %replaced by two independent disturbances $\omega$ and the random vector $\bm\nu$. 
		
		%This error-decoupling is of crucial importance in what will follow.  
		%Before describing this, we next formulate
		%the quantity of interest, the interventional distribution $T(y|x)$ in terms of the equivalent model \eqref{eq:tilde_X}-\eqref{eq:tilde_Y}. This is presented in the following lemma, 	whose proof appears in Appendix \ref{Appendix:Proof}. 
		\begin{lemma}[Normalized Interventional Distribution]\label{lemma:Invariant:Do}
			Suppose that assumption \ref{asu:cond:continuity} holds. Namely, $\bm{\epsilon}|\upeta$ is a $d$-dimensional continuous random vector for any $\upeta$. 
			In terms of the decoupled system as in Eqs. (\ref{eq:tilde_X})-(\ref{eq:tilde_Y}), the quantity of interest (\ref{eq:do_X})
			admits the following form,   
			\begin{equation}
				H(y|x) = \mathbb{E}_{\substack{\upomega\sim U[0,1]\\ \bm\nu\sim U[0,1]^d}}\braket{\mathds{1}\curly{g^*(x,\upomega,\bm\nu)\le y}}.
				\label{eq:Tyx_nu}
			\end{equation}
		\end{lemma}
%		Eq. \eqref{eq:Tyx_nu} plays an important role in the establishment of the equivalence class of structural functions in next section. This equivalence class is a building-block in achieving identifiability.
		
		\subsection{Identifiability of the interventional distribution}
		
		%	\textcolor{red}{Work on roadmap for this section, and relation to error-decomposition of previous section}
		
		Recall that our primary goal is to estimate the intervention $T(y|x)$ of (\ref{eq:do_X}). %		A key question is whether this goal is at all feasible. Namely, do the observables $(X,Y,Z)$ uniquely determine the function $T(y|x)$? The identifiability of $T(y|x)$ is important because following lemma \ref{lem:H_T} it implies that $H(y|x)$, defined in Eq. \eqref{eq:def_H}, is also uniquely determined.		
	The key difficulty is that observing $(X,Y,Z)$ does not uniquely determine the functions ${h}^*$ and ${g}^*$. We alleviate this issue by introducing a set of functions, which are all indistinguishable with respect to the observations. To this end, let $\tilde{F}$ be an arbitrary continuous distribution. Define		
		\begin{align}
			&\widetilde{\mathcal{C}}:=\left\{(\widetilde{h},\widetilde{g}): \widetilde{X}=\widetilde{h}(Z,\widetilde{\upomega}), \, \widetilde{Y}=\widetilde{g}(\widetilde{X},\widetilde{\upomega},\bm\nu),\hspace{0.5em} \left(\widetilde{X},\widetilde{Y},Z\right)\overset{d}{=}\left(X^*,Y^*,Z\right), \,  \widetilde{\upomega}\sim U[0,1], \bm\nu\sim U[0,1]^d \right\}.
			\label{eq:equivalence:class}
	\end{align}
	where $\bm\nu$ and $\widetilde{\upomega}$ are independent. 	
In other words, the set $\widetilde{\cal C}$ is the {\em equivalence class} of 
		the given triangular model. Namely, all pairs of structural functions $\widetilde{h}$ and $\widetilde{g}$ give rise to the same distribution for the observables $(X,Y,Z)$. 
For any pair $(\tilde{h},\tilde{g})$ there is a corresponding function 
\begin{equation}
	\tilde{H}(y|x) = \mathbb{E}_{\substack{\tilde{\upomega}\sim U[0,1]\\ \bm\nu\sim U[0,1]^d}}\braket{\mathds{1}\curly{\tilde{g}(x,\tilde{\upomega},\bm\nu)\le y}}.
	\label{eq:Tyx_nu}
\end{equation}	
Note that any pair of structural functions $(\widetilde{h}, \widetilde{g})\in\mathcal{\widetilde{C}}$  constitutes a distribution, 
\begin{align}
\mathcal{F}_{X,\upomega|Z=z}^{\widetilde{h}}(x,\omega):=\mathbb{E}\left[\mathds{1}\left\{\widetilde{h}(z,\upomega)<x, \upomega<\omega\right\}\right], \label{eq:F:X:omega}
\end{align}
as well as a counterfactual distribution,
\[
F_{X,Y|Z=z}^{\text{CF}}(x,y):=\mathbb{E}\left[\mathds{1}\left\{\widetilde{h}(z, \upomega)<x, \widetilde{g}(x, \upomega,\bm\nu) < y\right\}\right],\] 
such that the counterfactual density is $f_{X,Y|Z=z}^{\text{CF}}(x,y):=\mathbb{E}\left[\delta(\widetilde{h}(z, \upomega)-x) \delta(\widetilde{g}(x, \upomega,\bm\nu)-y)\right]$.  In next steps, this density of interest is shown to be identifiable.
\begin{lemma}[Fourier Transform of Counterfactual Density]\label{lemma:Fourier:CF:representation}
Under assumptions \ref{asu:common}-\ref{asu:Integrability} the Fourier transform of the counter factual distribution admits the representation as a Fredholm integral equation of the first kind for any pair of frequencies $(\xi_1, \xi_2)$ and value of $z$,
	\begin{align}
		\Gamma_{\xi_1, \xi_2}(z;(\widetilde{h}, \widetilde{g}))=\mathbb{E}_{\substack{\upomega\sim U[0,1]\\X\sim F_{X|Z=z}}}\left[\int \exp(-i\xi_1 \cdot X) \exp(-i\xi_2 \cdot \widetilde{g}(X, \upomega, \bm\nu))d\bm\nu\right].\label{eq:box:counter:X:Y}
	\end{align}
\end{lemma}
A sufficient condition for identifiability in equivalence class $\mathcal{\widetilde{C}}$ is to uniquely determine the Fourier transform of the counterfactual density from the joint distribution of the triplets $(X,Y,Z)$, namely, having that,
\begin{align}
	\Gamma_{\xi_1, \xi_2}^{\text{CF}}(z; (\widetilde{h}, \widetilde{g})) = \Gamma_{\xi_1, \xi_2}^{\text{CF}}(z; (h^*, g^*)) \quad \forall (z,\xi_1, \xi_2).\label{eq:coeff:obs:pair:CF:Fourier}
\end{align} 
We rely on the following lemma as a building block in establishing that this condition is satisfied under assumptions \ref{asu:common}-\ref{asu:Integrability}, by formulating the observational equivalence condition in terms of Fourier transform.
\begin{lemma}[Fourier Transform Coefficients]\label{lemma:Fourier:representation}
	Under assumptions \ref{asu:common}-\ref{asu:Integrability}
	the following Fredholm integral equation of the first kind uniformly holds for any pair $(\widetilde{g}, \widetilde{h})\in\mathcal{\widetilde{C}}$ constituting a density $\mathcal{F}_{X,\upomega|Z=z}^{h^*}(x,\omega))$ as in Eq. \eqref{eq:F:X:omega},
	\begin{align}%[box=\tcbhighmath]{equation}
		\Gamma_{\xi_1, \xi_2}(z)=   
		\mathbb{E}_{(X,\upomega)\sim \mathcal{F}_{X,\upomega|Z=z}^{\widetilde{h}}}\left[\int \exp(-i\xi_1 \cdot X) \exp(-i\xi_2 \cdot \widetilde{g}(X, \upomega, \bm\nu))d\bm\nu\right] \quad \forall (z, \xi_1, \xi_2) ,\label{eq:coeff:obs:pair:Fourier}
	\end{align}
    where the LHS is a specific Fourier transform coefficient of the known density $f_{X,Y|Z=z}$ for frequencies $(\xi_1, \xi_2)$. 
\end{lemma}
The following theorem \ref{theorem:coefficients} states that for any two pairs of functions $(g^*,h^*)$ and  $(\widetilde{g}, \widetilde{h})$ satisfying observational equivalence (Eq. \eqref{eq:coeff:obs:pair:Fourier}), counterfactual equivalence also holds. 
\begin{theorem}[Coefficients comparison]\label{theorem:coefficients}
	Let $m(x,y):= \exp(-i\xi_1 \cdot x) \exp(-i\xi_2 \cdot y)$. Suppose that the following observational equivalence holds for two mechanisms $(h^*, g^*)$ and $(\widetilde{h}, \widetilde{g})$ generating $(X^*,Y^*)$ and $(\widetilde{X},\widetilde{Y})$, respectively, in regime $\rho\in\left\{0, 1\right\}$,
	\begin{align}%[\right]\mathbb{E}_{(\upeta,\upeta^{\prime})\sim F_{\eta,\eta^{\prime}}(\cdot,\cdot;\rho=1)}} \left[  \right]
	& \mathbb{E}\left[ m(\widetilde{X}, \widetilde{Y})|Z=z,\rho=1\right]=\mathbb{E}\left[ m(X^*, Y^*)|Z=z,\rho=1\right] = \varphi(z)  \quad \forall z, \label{eq:obs:equiv}
\end{align}%	$g_1(x, \eta) - g_2(x, \eta) = 0$
such that, $(X^*,Y^*)|Z, \rho$ and $(\widetilde{X},\widetilde{Y})|Z, \rho$ are strongly complete mixtures of Gaussians  with respect to $Z$. Then,  
	\[\mathbb{E}\left[ m(\widetilde{X}, \widetilde{Y})|Z=z,\rho=0\right]=\mathbb{E}\left[ m(X^*, Y^*)|Z=z,\rho=0\right] \quad\forall z.\]
	%\[\int m_1(x,\omega^{\prime}, \nu)d\nu  =  \int m_2(x,\omega^{\prime}, \nu)d\nu  \text{ almost everywhere}.\]
	This implies that in any regime $\rho$ the conditional expectation of $m(\cdot,\cdot)$ is invariant to the choice of the mechanism. % is invariant to the choice of the kernel ($\mathcal{F}_{\upomega|X=x,Z=z}^{\widetilde{h}}$ and $\mathcal{F}_{\upomega|X=x,Z=z}^{h^*}$ in Eq. \eqref{eq:coeff:obs:pair:Fourier})% As such, they are uniquely determined from the observables. 
\end{theorem}
The theorem below relies on theorem \ref{theorem:coefficients} to show that Eq. \eqref{eq:coeff:obs:pair:CF:Fourier} holds. Simply put, it states that the interventional distribution is invariant with respect to the choice of the structural form functions in the equivalence class $\widetilde{\mathcal{C}}$ (see proof in appendix \ref{Appendix:Proof}).
\begin{theorem}[Identifiability]\label{Theorem:Identifiability}
	Given assumptions \ref{asu:common}-\ref{asu:Integrability},  the intervention function is uniquely determined from the joint distribution of the triplets $(X,Y,Z)$. Namely, $\forall (\widetilde{g}, \widetilde{h})\in\widetilde{\mathcal{C}}$ it holds that,
	\begin{equation}
		H(y|x) = \tilde H(y|x)		
				\label{eq:H_equals_tilde_H}
	\end{equation}
%	
%	\begin{align}
%		&F_{Y|\text{Do}(X=x)}(y)=\frac{1}{f_{X|Z=z}(x)} \int_{\xi_1}\int_{\xi_2}\mathbb{E}_{\upomega,\upomega^{\prime},\bm\nu}\left[\Psi_{\xi_1}\left(-\widetilde{g}\left(\widetilde{h}(z,\upomega),\upomega^{\prime},\bm\upnu\right)\right)\Psi_{\xi_2}\left(-\widetilde{h}(z,\upomega)\right)\right]\psi_{\xi_3}(x)\Psi_{\xi_2}(y)d\xi_3 d\xi_2.
%	\end{align}
\end{theorem}

\begin{comment}
So far, we have formed the relationship between the observed and the interventional distributions through a system of Fredholm equations with an unknown kernel. In the ensuing section we develop an augmented Fourier series expansion method to show that the moments equality in Eq. \eqref{eq:moments} holds. Then, we show that the interventional distribution is uniquely determined as a linear combination of these moments (in theorem \ref{Theorem:Identifiability} to follow), implying that any pair of structural functions belonging to equivalence class $\widetilde{\mathcal{C}}$ satisfying Eq. \eqref{eq:int:star} yields the same interventional distribution. This will be done without reliance on the monotonicity assumption.
\end{comment}
		The aforementioned may give rise to the following query: can a monotone data generation process  produce an equivalent distribution to the one propagated by a non-monotonic data generation process? The answer to such query is definitely negative (as is proven here in lemma \ref{lemma:Auxiliary}, Appendix \ref{Appendix:Proof}). This is so because any monotonic data generation process lacks sufficient variation in the distribution of the outcome variable due to the restriction on any value of context producing the same distribution of $Y$ within a fixed conditional quantile of $X$. %In others words, a degenerate context, which may fail causal inference thereby leading to the wrong counterfactual distribution (see next section).
		
		So far we have established identifiability for the interventional distribution. This is a very strong result in that it is universal, neither relying on any structural form restrictions nor on any assumptions regarding the true measure (prior to normalization) of the disturbances. Thus, it lends itself to nonparametric treatment. Building on the above identifiability result of the counterfactual distribution, we next develop the synthetic counterfactual machinery. Theorem \ref{theorem:coefficients} establishes that the kernel can be unknown in Fredholm integral equations, whereas Theorem \ref{Theorem:Identifiability} establishes the identifiability of the counterfactual distribution based on the insight of theorem \ref{Theorem:Identifiability}. % the building blocks necessary to evaluate the interventional distribution from the observed triplets of $(X,Y,Z)$.
		\subsection{Synthetic counterfactual machinery}
		%\textcolor{red}{discuss the next steps} 
		%Simply put, the same dependence structure is used to map the input vector $(Z,\omega,\nu)$ into the output vector $(X,Y)$. %In definition \ref{Definition:Equivalence:Class} we rely on the straightforward equality $F_{Y|X,Z}(y|x,z)=\frac{\frac{\partial}{\partial x}F_{Y,X|Z}(y,x|z)}{f_{X|Z}(x|z)}$, implying that $F_{Y|X,Z}(y|x,z)$ is defined uniquely from the expressions $\frac{\partial}{\partial x}F_{Y,X|Z}(y,x|z)$ and $f_{X|Z}(x|z)$ in both of which $\omega$ is uniformly distributed on $[0,1]$ (independently of $X$), while $X$ is held fixed at an arbitrarily chosen values $x$. 
		Following theorem \ref{Theorem:Identifiability}, the interventional distribution $H(y|x)$ is the same for any sequential generator functions in the equivalence class $(h^*,g^*)\in\widetilde{\mathcal{C}}$. This result naturally lends itself for an arbitrarily chosen sequential generator functions $(\widetilde{h},\widetilde{g})$ for defining,
		\begin{align}
			& F_{Y,X|Z=z}^{\text{O}}(y,x|\widetilde{g},\widetilde{h}):=P\left(Y\le y,X\le x|Y=\widetilde{g}(X,\upomega,\upnu), X=\widetilde{h}(z,\upomega)\right)\label{equ:FYX:givenZ}\\
			&=\int_{0}^1\int_{[0,1]^d}\mathds{1}\left\{\widetilde{g}(\widetilde{h}(z,\omega),\omega,\bm\nu)\le y\right\}\mathds{1}\left\{\widetilde{h}(z,\omega)\le x\right\}d\bm\nu d\omega,\nonumber
		\end{align}
		which is  the observed (O) conditional generated density function of $Y$ and $X$ given $Z=z$. Similarly, the observed (O)  conditional generated density function of $X$ given $Z=z$ is defined as,
		\begin{align}
			& F_{X|Z=z}^{\text{O}}(x|\widetilde{h}):=P\left(X\le x|X=\widetilde{h}(z,\upomega)\right)=\int_{0}^1\mathds{1}\left\{\widetilde{h}(z,\omega)\le x\right\}\hspace{0.25em} d\omega. \label{equ:X:CF:givenZ}
		\end{align}
		Note that the key innovation embedded in \eqref{equ:FYX:givenZ} is in that the entire support domain $[0,1]\times[0,1]^d$ of the latent space (random disturbances) is used to characterize the conditional joint distribution of $Y$ and $X$ given $Z=z$.
		This is superior to the concept widely used in the literature for characterizing $F_{Y|X,\upomega}(y|x,\omega)$, the conditional distribution of $Y$ given $X=x$ and the actual latent space (up to normalization), which is restricted by the available data. The present treatment in \eqref{equ:FYX:CF:givenZ} however, conveys an entropy superior way to float the interrelationships among the unobservables and the observables, as it leads to the following general characterization of the unobserved counterfactual (CF) conditional joint distribution of $Y$ and $X$ given $Z=z$ (by letting $\widetilde{\upomega}\sim U[0,1]$ satisfying $\widetilde{\upomega}\ci\upomega$):
		\begin{align}
			& F_{Y,X|Z=z}^{\text{CF}}(y,x|\widetilde{g},\widetilde{h}):=P\left(Y\le y,X\le x|Y=\widetilde{g}(X,\widetilde{\upomega},\upnu), X=\widetilde{h}(z,\upomega)\right)\label{equ:FYX:CF:givenZ}\\
			&=\int_{0}^1\int_{0}^1\int_{[0,1]^d}\mathds{1}\left\{\widetilde{g}(\widetilde{h}(z,\omega),\widetilde{\omega},\bm\nu)\le y\right\}\mathds{1}\left\{\widetilde{h}(z,\omega)\le x\right\}d\bm\nu d\omega \hspace{0.25em} d\widetilde{\omega}. \nonumber
		\end{align}
		The increase in entropy is due to the presence of two independent contextual covariates $\upomega$ and $\widetilde{\upomega}$ in \eqref{equ:FYX:CF:givenZ} introducing independent variations to $X$ and to $Y$. This yields the following counterfactual result:
		\begin{align}
			& \mathbb{E}\left[\mathds{1}\left\{Y\le y\right\}|\text{do}(X=x); \widetilde{g},\widetilde{h}\right] = \mathbb{E}_{z\sim F_Z}\left[\left.\frac{\frac{\partial}{\partial x}F_{Y,X|Z}^{\text{CF}}(y,x;\widetilde{g},\widetilde{h})}{\frac{\partial}{\partial x}F_{X|Z}^{\text{O}}(x;\widetilde{g},\widetilde{h})}\right|Z=z\right].\label{equ:Expected:FY:CF:givenX}
		\end{align} 
		The result in \eqref{equ:Expected:FY:CF:givenX} is propagated by the unconstrained support domain of \eqref{equ:FYX:CF:givenZ} unlike in models employing the additive control function approach \citep{heckman1985alternative}, the non-additive control function approach \citep{newey1994kernel,imbens2009identification,blundell2014control,hoderlein2009identification}, the generated covariates \citep{mammen2012nonparametric}, the non-separable counterfactual distribution based on quantiles \citep{chernozhukov2013inference}, the treatment of unobservables in \cite{schennach2014entropic} and the average causal effects by employing do-interventions \citep{pearl2019interpretation}. Further discussion pertaining to the effect of limited empirical support on counterfactual analysis is detailed in section \ref{Sec:Restricted:Support}. %We emphasize that such characterization of the relationship among $X$, $Y$ and $Z$ benefits from additional information embedded in the reconstruction of the data generation process (up to normalization). By doing so, we depart from the common practice of integrating out conditional expectations applied to the observed data \myciteopt{blundell2014control, imbens2009identification, newey1994kernel,pearl2019interpretation}.

		In fact, the context in and by itself makes it feasible to allocate different values for the counterfactual data set of $X$ and $Y$. By doing so, we can generate random draws, independent of each other, from the unconstrained support of the context and the action variables, unlike the case when the observed data set is used. This is demonstrated in Figure \ref{figure:CounterFactualRectSupport}:
		\begin{figure}[H]\centering\caption{\label{figure:CounterFactualRectSupport}Counterfactual Machinery vs. Partial Means  \\ (Unrestricted vs. Restricted support)}\vspace{2mm}
			\includegraphics[width=16cm,height=5cm]{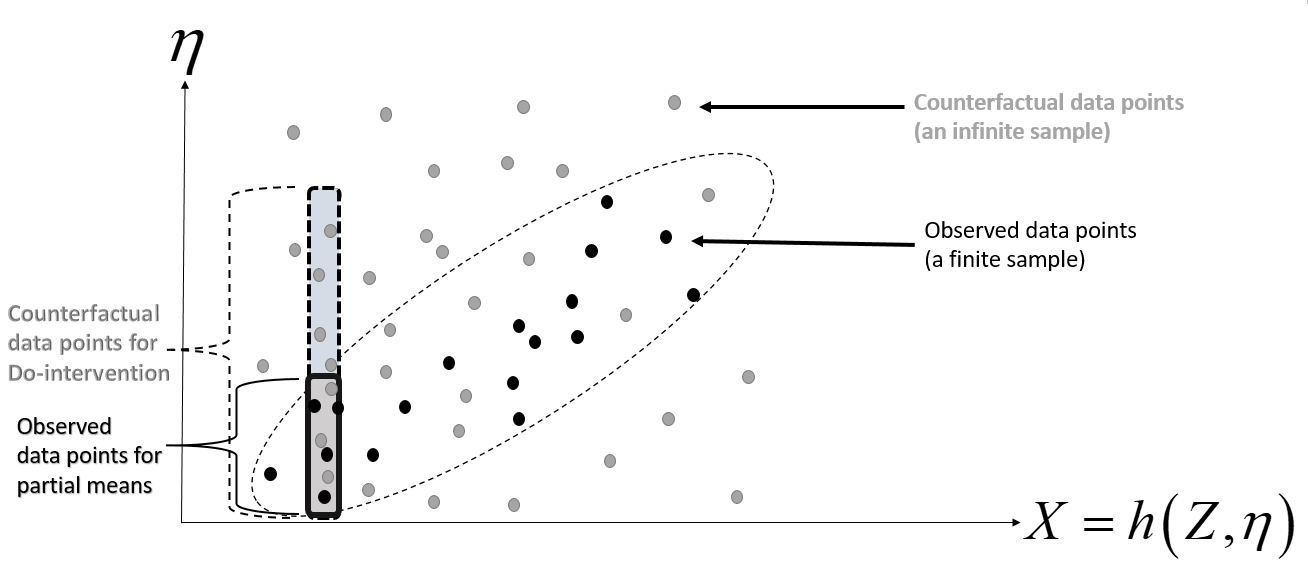}
		\end{figure}
		The observed data set is denoted by the black dots and the counterfactual data set is denoted by the gray dots. By construction, the former is randomly drawn from the joint support of $X$ and $\upeta$, while the latter is randomly drawn from the rectangular support $\text{Supp}(X)\times\text{Supp}(\upeta)$. Consequently, the partial means estimate of $\mathbb{E}[Y|\text{do}(X=x)]$ is calculated by using the observed data $X$ and $Y$ in which $X$ and $\upeta$ are jointly dependent as in Figure \ref{figure:CounterFactualRectSupport}.
		
		Our theoretical result (theorem \ref{Theorem:Identifiability}) implies that both the kernel as well as the integral equation can be completely characterized by generator functions belonging to $\widetilde{\mathcal{C}}$ and being incorporated in our proposed estimator (in section \ref{Sec:Estimation} to follow). To show that, recall the Dirac delta function $\updelta(\cdot)$ and denote,		
		\[
		f_{Y,X|Z=z}(y,x) = \int_{0}^{1}\left\{\int_{[0,1]^d}\updelta\left(y-\widetilde{g}\left(\widetilde{h}(z,\omega),\omega,\bm\nu\right)\right) d\bm\nu\right\} \updelta\left(x-\widetilde{h}(z,\omega)\right) d\omega.
		\]
		
		Next, we propose our estimator and justify the type of Neural network architecture we develop, which involves sequential generator functions. The sequential generator functions estimator is proposed to mimic the underlying data generation process as in a triangular model, where by its very definition, the structural functions are formulated in a sequential manner.

		\section{Estimation: contaminated GAN (CONGAN)}\label{Sec:Estimation}
		%Neural network assumptions \citep{singh2018nonparametric}:
		
	The formulation of the system of Fredholm equations in \eqref{eq:integral:normalized} (in Appendix \ref{Sec:integral}) is used only for achieving the identifiability of $H(y|x)$ and not for estimation. This disparity between identification and estimation is related to the fact that such a system of equations might be ill-posed. Namely, it may provide an estimator which is a discontinuous function of the data and thus, inconsistent \citep{darolles2011nonparametric, newey2003instrumental}. This issue is resolved here by constructing a novel contaminated GAN (CONGAN) estimator, which is a smooth function of the data. 	
	The benefit of the contaminated Neural network is in that it precludes the need to invert the generators in order to account for the posterior distribution of the latent context variable given the observables, which is computationally cumbersome. 
		
		%So far, Kolmogorov's superposition theorem \citep{kolmogorov1957} has been used as a building-block to characterize the set of triangular models for which the interventional distribution is identifiable. Its implication is that the problem is well posed. The present task is to estimate the do-interventional distribution in order to depict the empirical causal relation. To that end, we employ neural networks configuration deriving a universal approximation property up to an arbitrary degree of accuracy \citep{hornik1989multilayer}.
		
		Consider a probability space $(\Upomega,\mathcal{F},\mathbb{P})$, a measurable output spaces $(\mathcal{X}, \mathcal{Y})$ for the random vector $(X,Y)$, an observed input space  $\mathcal{Z}$ for the random variable $Z$ and  a latent input space $(\mathcal{U,V})$ for the random vector $(\upomega,\bm\nu)$. The set of paired measurable functions $\mathcal{G}:(\mathcal{X},\mathcal{U},\mathcal{V})\mapsto\mathcal{Y}$ and $\mathcal{H}:(\mathcal{Z},\mathcal{U})\mapsto\mathcal{X}$ map from an input space (observed and latent space) to the output space. These functions are termed \textit{generator functions} designated to produce random samples. We consider also $\mathcal{D}$ a set of measure functions $d:(\mathcal{X},\mathcal{Y},\mathcal{Z})\mapsto[0,1]$ referred to as \textit{discriminators}. Their role is classifying the input space as either real or synthetic.  Let $\mu$ be a measure on $(\mathcal{X}, \mathcal{Y})$, we make the following assumption  \citep{singh2018nonparametric}:
		\begin{asu}\label{Generators:assumption}
			$\forall(\widetilde{h},\widetilde{g})\in\mathcal{(G,H)}$ the distributions of $\widetilde{g}(X, \omega,\nu)$ and $\widetilde{h}(Z, \omega)$ are absolutely continuous with respect to $\mu$. The joint distribution of $(Y,X,Z)$ is absolutely continuous with respect to $\mu$.\\  
		\end{asu} %\subasu \label{asu1i} The joint distribution of (Y,X,Z) is absolutely continuous with respect to µ.\\        
		
		We define the contaminated GAN (CONGAN) approach as an extension of the original GAN \citep{goodfellow2014generative} as follows:
		\begin{align}
			V(d,(\widetilde{h},\widetilde{g})) = \left\{\mathbb{E}_{\substack{X,Y,Z\sim F_{X,Y,Z}}}\left[\log\left(d(X,Y,Z)\right)\right] + \mathbb{E}_{\substack{\upomega\sim U[0,1]\\\nu\sim U[0,1]\\Z\sim F_Z}}\left[\log\left(1-d\left(\widetilde{h}(Z,\omega), \widetilde{g}\left(\widetilde{h}(z,\omega),\omega,\bm\nu\right)\right),Z\right)\right]\right\}\label{Cont:GAN}
		\end{align}
		The key difference between the original GAN and our proposed CONGAN is entirely related to the sequential construction of the generator functions (the generator function $\widetilde{h}$ belongs to the input of the generator function $\widetilde{g}$) both share the noise $\omega$ as appears under the \textit{contamination} braces in \eqref{Cont:GAN}. Namely, \eqref{Cont:GAN} is reduced to the original GAN when $\widetilde{h}(Z,\omega)$ is substituted with $X$ and the expectation of the CON braces is taken w.r.t $X$, $Z$ and $\bm\nu$ rather than w.r.t $\bm\nu,\omega$ and $Z$.
		
		The optimal discriminator assigns values close to $1$ given data from the joint distribution of $(X,Y,Z)$ and assigns values close to zero given data from the distribution obtained by the generator functions $(\widetilde{h},\widetilde{g})$. Thus, the optimal $d$ given a fixed pair of generator functions $(\widetilde{h},\widetilde{g})$ is:
		\begin{align}
			d=\underset{d\in\mathcal{D}}{\arg\max}\hspace{0.25em} V(\cdot,(\widetilde{h},\widetilde{g}))
		\end{align}
		Then, the optimal pair of generator functions $(\widetilde{h},\widetilde{g})$ is the one solving the $\min\max$ problem: 
		\begin{align}
			(\widetilde{h},\widetilde{g})=\underset{(\widetilde{h},\widetilde{g})\in\mathcal{(G,H)}}{\arg\min}\hspace{0.25em}\left\{\underset{d\in\mathcal{D}}{\arg\max}\hspace{0.25em} V(d,(\widetilde{h},\widetilde{g}))\right\}
		\end{align}
		
		Define $f_{Y,X|Z}^{\text{O}}(y,x|z;\widetilde{g},\widetilde{h}):=\frac{\partial^2}{\partial x\partial y}F_{Y,X|Z}^{\text{O}}(y,x|z;\widetilde{g},\widetilde{h})$ and suppose that assumption \ref{Generators:assumption} holds. Then following Theorem 2.2 in \cite{singh2018nonparametric},
		\begin{align} 
			V(d,(\widetilde{h},\widetilde{g})) = \int_{x,y,z}\left\{\log\left(d(x,y,z)\right)f_{Y,X|Z}(y,x|z) + \log\left(1-d(x,y,z)\right)f_{Y,X|Z}^{\text{O}}(y,x|z;\widetilde{g},\widetilde{h})\right\}f_z(z)dx \hspace{0.2em}dy\hspace{0.2em} dz
		\end{align}
		$\forall \hspace{0.2em}(\widetilde{h},\widetilde{g})\in\mathcal{(G,H)}$ and $\forall \hspace{0.2em}d\in\mathcal{D}$. Hence, for any fixed pair of generators $(\widetilde{h},\widetilde{g})\in\mathcal{(G,H)}$ and $d\in\mathcal{D}$ maximization of $V(d,(\widetilde{h},\widetilde{g}))$ has the form:
		\begin{align}
			d(x,y,z) =  \frac{f_{Y,X|Z}(y,x|z)}{f_{Y,X|Z}(y,x|z) + f_{Y,X|Z;\widetilde{g},\widetilde{h}}^{\text{O}}(y,x|z)}, \hspace{1em} \forall \hspace{0.2em}(x,y,z) \in \text{Supp}(X,Y,Z)
		\end{align}
		This yields:
		\begin{align}
			& \underset{d\in\mathcal{D}}{\max}\hspace{0.25em} V(d,(\widetilde{h},\widetilde{g})) = \int_{x,y,z}\left\{f_{Y,X|Z}(y,x|z) \log\left(\frac{f_{Y,X|Z}(y,x|z)}{f_{Y,X|Z}(y,x|z) + f_{Y,X|Z}^{\text{O}}(y,x|z;\widetilde{g},\widetilde{h})}\right)\right\}f_z(z)dx \hspace{0.2em}dy\hspace{0.2em} dz\label{Divergence:loss:function}\\
			& +     \int_{x,y,z}\left\{f_{Y,X|Z}(y,x|z)\log\left(\frac{f_{Y,X|Z}^{\text{O}}(y,x|z;\widetilde{g},\widetilde{h})}{f_{Y,X|Z}(y,x|z) + f_{Y,X|Z}^{\text{O}}(y,x|z;\widetilde{g},\widetilde{h})}\right)f_z(z)dx \hspace{0.2em}dy\hspace{0.2em} dz\right\}\nonumber
		\end{align}
		The result above represents the Nash equilibrium concept, in the sense that both the discriminator as well as the generator players reach their best responses in maximizing their value functions given the other player action \citep{goodfellow2014generative}. In practice $f_z(z)$, $f_{Y,X|Z}(y,x|z)$ and $f_{Y,X|Z}^{\text{O}}(y,x|z;\widetilde{g},\widetilde{h})$ are estimated by a kernel density estimator as a function of some unknown bandwidth $\bm{\upsigma}=(\upsigma_x,\upsigma_y,\upsigma_z)$. This bandwidth is a parameter of the model satisfying the aforementioned Nash equilibrium. Consequently, it obviates the known problematic cross-validation, which requires formidable successive multiple estimations of the likelihood function.
		
		Next we attend to the characterization of our proposed Neural network to be used in the contaminated GANS methodology. For doing so, we denote a feed-forward Neural network consisting of $L$ layers (such that there are $L-1$ hidden layers). In what follows we explain in detail the formulation of the sequential generator functions in this Neural network, presented in the following figure:
		%\subsection{Sequential generator functions}
		%For brevity of exposition, the graphic representation of the network is as depicted in figure \ref{figure:NeuralNetworkSequentialGenerators}:
		\begin{figure}[H]\centering\caption{\label{figure:NeuralNetworkSequentialGenerators}Neural Network with sequential generators ($\widetilde{h}$ and $\widetilde{g}$)}\vspace{2mm}
			\includegraphics[width=8cm]{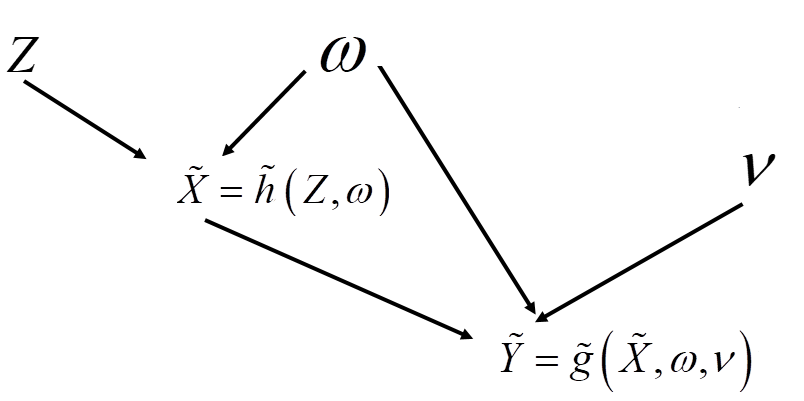}
		\end{figure}
		Figure \ref{figure:NeuralNetworkSequentialGenerators} depicts the sequential generation of $\widetilde{X}$ and $\widetilde{Y}$ induced by a triangular model with error-decoupling. %First, $Z$ and a normalization of the latent context covariate $\upomega$, defined on the support $[0,1]$, generate $\widetilde{X}$. Second, the generated $\widetilde{X}$, the same latent context and a random disturbance term $\nu$, defined on the support $[0,1]$, generate $\widetilde{Y}$. 
		By adopting this sequential data generation process, in which $\widetilde{g}$ takes the output of $\widetilde{h}$ as one of its arguments, one can exploit the unrestricted support domain of the latent space in order to construct the observed distribution in \eqref{equ:FYX:givenZ} as well as the counterfactual distribution in \eqref{equ:FYX:CF:givenZ}. This latent space consists of $\upomega$ and the random disturbance $\nu$. 
		Next, we characterize explicitly the Neural network architecture to estimate these functions.
		\subsection{Multi-layer neural network with sequential generator functions}
		Denote the input space of its $\ell$'th layer in the generator functions of $Y$ and $X$ by the column vectors $\bm{\mathcal{Y}}_{\ell-1}$ and $\bm{\mathcal{X}}_{\ell-1}$ of sizes $T_{\ell-1}^Y\times 1$ and $T_{\ell-1}^X\times 1$, respectively. The parameters associated with these covariate vectors are referred to as weights (slopes) and biases (intercepts) which are denoted by the matrices $\bm{W}_{\ell}^{Y}$ and $\bm{W}_{\ell}^{X}$ of sizes $(T_{\ell}^Y+1)\times (T_{\ell-1}^Y+1)$ and  $(T_{\ell}^X+1)\times (T_{\ell-1}^X+1)$, respectively. Let $\widetilde{g}$ and $\widetilde{h}$ be the generator functions of $Y$ and $X$, respectively, with parameters $\bm{\upbeta}=(\bm{W}_{1}^X,...,\bm{W}_{L}^X)$ and $\bm{\uptheta}=(\bm{W}_{1}^Y,...,\bm{W}_{L}^Y)$, defined as:
		% of the $\ell$'th layer  for $\ell=1,...,L^{\upchi}-1$. %such that $T_{0}^{\upchi}$ denotes the original input space of the generator of $\upchi$. The Neural network is defined as follows:%By construction $\widetilde{g}$ takes three arguments $(X$,$\omega$, and $\nu)$ and $\widetilde{h}$ takes two arguments $(Z$, and $\omega)$. Consequently, the number of covariates in the input space of the first layer are $(T_0^Y,T_0^X)=(3,2)$. $T_{\ell-1}^{\upchi}$ is the size of the $\ell$'th hidden layer for $\upchi\in\left\{X,Y\right\}$
		\begin{align}
			\widetilde{g}(X,\omega,\nu,\bm{\uptheta})=\bm{\mathcal{Y}}_{L},  \hspace{1em}\bm{\mathcal{Y}}_{0} = [X,\omega,\nu], \hspace{1em} \bm{\mathcal{Y}}_{\ell} = \upvarphi\left(\bm{W}_{\ell}^Y \bm{\mathcal{Y}}_{\ell-1}\right), \hspace{1em} \ell = 1,...,L
		\end{align}
		and
		\begin{align}
			\widetilde{h}(Z,\omega;,\bm{\upbeta})=\bm{\mathcal{X}}_{L},  \hspace{1em}\bm{\mathcal{X}}_{0} = [Z,\omega], \hspace{1em} \bm{\mathcal{X}}_{\ell} = \upvarphi\left(\bm{W}_{\ell}^X \bm{\mathcal{X}}_{\ell-1}\right), \hspace{1em} \ell = 1,...,L
		\end{align}
		where $\upvarphi(\cdot):\mathbb{R}\to C\subset\mathbb{R}$ is some differentiable activation function to be applied element-wise. The input space of the first layer is $(T_0^Y, T_0^X)=(3,2)$ and the output of the last layer is $(T_L^Y, T_L^X)=(1,1)$.

		Real and synthetic data are denoted by $\left\{(X_i,Y_i,Z_i)\right\}_{i=1}^N$ and $\left\{(\widehat{X}_i,\widehat{Y}_i,Z_i)\right\}_{i=1}^N$, respectively, with  $\widehat{X}_i:=\widetilde{h}(Z_i,\omega_{0i},\bm{\upbeta})$ and $\widehat{Y}_i:=\widetilde{g}(\widehat{X}_i,\omega_{0i}, \nu_{i};\bm{\uptheta})$.

		% bad sentence: the message is not clear
		\vspace{1em}
		It should be emphasized that by adopting a Shannon-divergence minimization criterion function, we depart from the practice of Wasserstein GAN (WGAN) \citep{athey2020}. Although WGANS can be trained more easily, it requires the utilization of a penalty function to alleviate various optimization difficulties \citep{gulrajani2017improved}, which may limit the support domain of the generated data distribution \citep{wei2018improving}.\footnote{This limitation stems from the fact that this gradient penalty term is influenced only by the sampled points without examining a significant parts of the support domain \citep{wei2018improving}.} As we deal with counterfactual analysis we aspire to utilize the entire support domain of the data distribution function. Thus, we adopt a Shannon-divergence minimization criterion function, in which the density is evaluated at each point in the data to preserve the original distribution support. We define the following density estimators:
		\begin{align}
			\widehat{p}_{\bm{\upsigma}}(x,y,z) = \frac{1}{N}\sum_{i=1}^N  \frac{1}{\upsigma_x}K\left(\frac{X_i-x}{\upsigma_x}\right)\frac{1}{\upsigma_y}K\left(\frac{Y_i-y}{\upsigma_y}\right)\frac{1}{\upsigma_z}K\left(\frac{Z_i-z}{\upsigma_z}\right)
		\end{align}
		\begin{align}
			\widehat{q}_{\bm{\upsigma}}(x,y,z) = \frac{1}{N}\sum_{i=1}^N  \frac{1}{\upsigma_x}K\left(\frac{\widehat{X}_i-x}{\upsigma_x}\right)\frac{1}{\upsigma_y}K\left(\frac{\widehat{Y}_i-y}{\upsigma_y}\right)\frac{1}{\upsigma_z}K\left(\frac{Z_i-z}{\upsigma_z}\right)
		\end{align}
		
		Utilizing the smoothing Jensen-Shannon Divergence emanating from \eqref{Divergence:loss:function} \citep{sinn2018non}:
		\begin{align}
			\mathcal{K}_n(\bm{\uptheta},\bm{\upbeta},\bm{\upsigma}) = \frac{1}{N}\sum_{i=1}^N\log\frac{\widehat{p}_{\bm{\upsigma}}(X_i,Y_i,Z_i)}{\widehat{p}_{\bm{\upsigma}}(X_i,Y_i,Z_i)+\widehat{q}_{\bm{\upsigma}}(X_i,Y_i,Z_i)} + \frac{1}{N}\sum_{i=1}^N\log\frac{\widehat{q}_{\bm{\upsigma}}(\widehat{X}_i,\widehat{Y}_i,Z_i)}{\widehat{p}_{\bm{\upsigma}}(\widehat{X}_i,\widehat{Y}_i,Z_i)+\widehat{q}_{\bm{\upsigma}}(\widehat{X}_i,\widehat{Y}_i,Z_i)}\label{CONT:GAN:LOSS}
		\end{align}
		where $\bm{\upsigma}:=(\upsigma_x,\upsigma_y,\upsigma_z)$ is the bandwidth vector. The optimal generator's strategy is to minimize the divergence:
		\begin{align}
			(\widehat{\bm{\uptheta}}^{(\text{iter+1})},\widehat{\bm{\upbeta}}^{(\text{iter+1})})=\underset{(\bm{\uptheta},\bm{\upbeta})}{\arg\min}\hspace{0.2em} \mathcal{K}_n(\bm{\uptheta},\bm{\upbeta},\bm{\upsigma}^{(\text{iter})})
		\end{align}

		We attend to the generation of simulated data for the purpose of causal counterfactual inference by using the estimated structural functions. The synthetic data is constructed from the synthetic (estimated) structural functions $\widetilde{g}$ and $\widetilde{h}$, whereas the real data is constructed from the original structural model $g$ and $h$ in \eqref{X:Structural} and in \eqref{Y:Structural}, respectively. The latter will be used as a benchmark for verifying the accuracy of the estimate obtained from the synthetic data.
		\subsection{Synthetic counterfactual machine vs. partial means estimator \citep{newey1994kernel}}\label{Sec:Restricted:Support}
		The following equations map an input vector $(Z_i,\upomega_{i},\widetilde{\upomega}_{i},\nu_i)$ to an output vector  $(X_{\text{real,i}}^{\text{CF}}, Y_{\text{real,i}}^{\text{CF}})$ in order to generate the \textit{real} counterfactual (CF) data set from the structural functions $g$ and $h$:% in \eqref{X:Structural} and \eqref{Y:Structural}, respectively
		\begin{align}
			& X_{\text{real,i}}^{\text{CF}} = h^*\left(Z_i,\omega_{i}^{\text{real}}\right), \hspace{1em} \omega_{i}^{\text{real}}, \widetilde{\omega}_{i}^{\text{real}}, \nu_i^{\text{real}} \sim U[0,1].\\
			& Y_{\text{real,i}}^{\text{CF}} = g\left(X_{\text{real,i}},F_{\bm\epsilon|\upeta}^{-1}\left(\nu_i^{\text{real}}|F_{\upeta}^{-1}\left(\widetilde{\omega}_{i}^{\text{real}}\right)\right)\right).
		\end{align}
		The following equations map an input vector $(Z_i,\upomega_{i},\widetilde{\upomega}_{i},\nu_i)$ to an output vector  $(X_{\text{synt,i}}^{\text{CF}}, Y_{\text{synt,i}}^{\text{CF}})$ in order to generate the \textit{synthetic (synt)} counterfactual (CF) data set from the synthetic (estimated) structural functions $\widetilde{g}$ and $\widetilde{h}$:
		\begin{align}
			& X_{\text{synt,i}}^{\text{CF}} =  \widetilde{h}\left(Z,\omega_{i}^{\text{synt}};\bm{\upbeta}\right), \hspace{1em} \omega_{i}^{\text{synt}}, \widetilde{\omega}_{i}^{\text{synt}}, \nu_i^{\text{synt}} \sim U[0,1]\\
			& Y_{\text{synt,i}}^{\text{CF}} =  \widetilde{g}\left(X_{\text{synt,i}},\widetilde{\omega}_{i}^{\text{synt}}, \nu_i^{\text{synt}};\bm{\uptheta}\right).
		\end{align}
		Similarly, the following equations map an input vector $(Z_i,\upomega_{i},\nu_i)$ to an output vector  $(X_{\text{real,i}}^{\text{O}}, Y_{\text{real,i}}^{\text{O}})$ in order to generate the \textit{real} observed (O) data set from the structural functions $g$ and $h$:
		\begin{align}
			& X_{\text{real,i}}^{\text{O}} = h^*\left(Z_i,\omega_{i}^{\text{real}}\right), \hspace{1em} \omega_{i}^{\text{real}}, \nu_i^{\text{real}} \sim U[0,1]\\
			& Y_{\text{real,i}}^{\text{O}} = g\left(X_{\text{real,i}},F_{\bm\epsilon|\upeta}^{-1}\left(\nu_i^{\text{real}}|F_{\upeta}^{-1}\left(\omega_{i}^{\text{real}}\right)\right)\right),
		\end{align}
		and the following equations map an input vector $(Z_i,\upomega_{i},\nu_i)$ to an output vector  $(X_{\text{synt,i}}^{\text{O}}, Y_{\text{synt,i}}^{\text{O}})$ in order to generate the \textit{synthetic (synt)} observed (O) data set from the synthetic (estimated) structural functions $\widetilde{g}$ and $\widetilde{h}$:
		\begin{align}
			& X_{\text{synt,i}}^{\text{O}} =  \widetilde{h}\left(Z,\omega_{i}^{\text{synt}};\bm{\upbeta}\right), \hspace{1em} \omega_{i}^{\text{synt}}, \nu_i^{\text{synt}} \sim U[0,1]\\
			& Y_{\text{synt,i}}^{\text{O}} =  \widetilde{g}\left(X_{\text{synt,i}},\omega_{i}^{\text{synt}}, \nu_i^{\text{synt}};\bm{\uptheta}\right).
		\end{align}	
		%\subsubsection{An estimator for the expected outcome from observed data sets}
		The estimated conditional expected outcomes given the action in the \textit{real} data is:
		\begin{align}
			& \widehat{\mathbb{E}}\left[Y_{\text{real}}|X=x\right] = \widehat{\mathbb{E}}\left[Y_{\text{real}}^{\text{O}}|X_{\text{real}}^{\text{O}}=x\right] = \frac{\frac{1}{N}\frac{1}{\upsigma_x}\sum_{i=1}^N Y_{\text{real,i}}^{\text{O}}K\left(\frac{X_{\text{real,i}}^{\text{O}}-x}{\upsigma_x}\right)}{\frac{1}{N}\frac{1}{\upsigma_x}\sum_{i=1}^N K\left(\frac{X_{\text{real,i}}^{\text{O}}-x}{\upsigma_x}\right)}.\label{real:O}
		\end{align}
		The estimated conditional expected outcomes given the action in the \textit{synthetic (synt)} data is:
		\begin{align}
			& \widehat{\mathbb{E}}\left[Y_{\text{synt}}|X=x\right] = \widehat{\mathbb{E}}\left[Y_{\text{synt}}^{\text{O}}|X_{\text{synt}}^{\text{O}}=x\right] \frac{\frac{1}{N}\frac{1}{\upsigma_x}\sum_{i=1}^N Y_{\text{synt,i}}^{\text{O}}K\left(\frac{X_{\text{synt,i}}^{\text{O}}-x}{\upsigma_x}\right)}{\frac{1}{N}\frac{1}{\upsigma_x}\sum_{i=1}^N K\left(\frac{X_{\text{synt,i}}^{\text{O}}-x}{\upsigma_x}\right)}.\label{synthetic:O}
		\end{align}
		
		%\subsubsection{A counterfactual estimator for the expected outcome from counterfactual data sets}
		The estimated counterfactual expected outcomes given the action in the \textit{real} data is:
		\begin{align}
			& \widehat{\mathbb{E}}\left[Y_{\text{real}}|\text{do}(X_{\text{real}}=x)\right] = \widehat{\mathbb{E}}\left[Y_{\text{real}}^{\text{CF}}|X_{\text{real}}^{\text{CF}}=x\right] = \frac{\frac{1}{N}\frac{1}{\upsigma_x}\sum_{i=1}^N Y_{\text{real,i}}^{\text{CF}}K\left(\frac{X_{\text{real,i}}^{\text{CF}}-x}{\upsigma_x}\right)}{\frac{1}{N}\frac{1}{\upsigma_x}\sum_{i=1}^N K\left(\frac{X_{\text{real,i}}^{\text{CF}}-x}{\upsigma_x}\right)}.\label{real:CF}
		\end{align}
		
		The estimated counterfactual expected outcomes given the action in the \textit{synthetic (synt)} data is:
		\begin{align}
			& \widehat{\mathbb{E}}\left[Y_{\text{synt}}|\text{do}(X_{\text{synt}}=x)\right] = \widehat{\mathbb{E}}\left[Y_{\text{synt}}^{\text{CF}}|X_{\text{synt}}^{\text{CF}}=x\right] \frac{\frac{1}{N}\frac{1}{\upsigma_x}\sum_{i=1}^N Y_{\text{synt,i}}^{\text{CF}}K\left(\frac{X_{\text{synt,i}}^{\text{CF}}-x}{\upsigma_x}\right)}{\frac{1}{N}\frac{1}{\upsigma_x}\sum_{i=1}^N K\left(\frac{X_{\text{synt,i}}^{\text{CF}}-x}{\upsigma_x}\right)}.\label{synthetic:CF}
		\end{align}
		
		%\subsubsection{A counterfactual estimator for the expected outcome from the observed data sets}
		The partial means estimator \citep{newey1994kernel} is presented here as a benchmark, as it is applied to the \textit{observed} data set, unlike the counterfactual expression applied to the \textit{counterfactual} dataset in \eqref{real:CF}: %The observed data is integrated over $\upeta$ to mimic the counterfactual (do-interventional) distribution:
		\begin{align}
			& \widehat{\widehat{\mathbb{E}}}\left[Y_{\text{real}}|\text{do}(X_{\text{real}}=x)\right] =\frac{1}{N}\sum_{i=1}^N \widehat{\mathbb{E}}\left[Y_{\text{real}}^{\text{O}}|X_{\text{real}}^{\text{O}}=x, \upomega^{\text{real}}=\omega_{i}^{\text{real}}\right]\label{Partial:Means}
		\end{align}
		with 
		\begin{align}
			& \widehat{\mathbb{E}}\left[Y_{\text{real}}^{\text{O}}|X_{\text{real}}^{\text{O}}=x, \upomega^{\text{real}}=\omega\right] = \frac{\frac{1}{N}\frac{1}{\upsigma_x}\frac{1}{\upsigma_{\omega}}\sum_{i=1}^N Y_{\text{real,i}}^{\text{O}}K\left(\frac{X_{\text{real,i}}^{\text{O}}-x}{\upsigma_x}\right)K\left(\frac{\omega_i^{\text{real}}-\omega}{\upsigma_{\omega}}\right)}{\frac{1}{N}\frac{1}{\upsigma_x}\frac{1}{\upsigma_{\omega}}\sum_{i=1}^N K\left(\frac{X_{\text{real,i}}^{\text{O}}-x}{\upsigma_x}\right)K\left(\frac{\omega_i^{\text{real}}-\omega}{\upsigma_{\omega}}\right)}.
		\end{align}
		\begin{comment}
		and similarly, 
		\begin{align}
		& \widehat{\widehat{\mathbb{E}}}\left[Y_{\text{synt}}|\text{do}(X_{\text{synt}}=x)\right] = \frac{1}{N}\sum_{i=1}^N \widehat{\mathbb{E}}\left[Y_{\text{synt}}^{\text{O}}|X_{\text{synt}}^{\text{O}}=x, \upomega^{\text{synt}}=\omega_{i}^{\text{synt}}\right]
		\end{align}
		with 
		\begin{align}
		& \frac{1}{N}\sum_{i=1}^N \widehat{\mathbb{E}}\left[Y_{\text{synt}}^{\text{O}}|X_{\text{synt}}^{\text{O}}=x, \upomega^{\text{synt}}=\eta\right] = \frac{\frac{1}{N}\frac{1}{\upsigma_x}\frac{1}{\upsigma_{\eta}}\sum_{i=1}^N Y_{\text{synt,i}}^{\text{O}}K\left(\frac{X_{\text{real,i}}^{\text{O}}-x}{\upsigma_x}\right)K\left(\frac{\omega_i^{\text{synt}}-\eta}{\upsigma_{\eta}}\right)}{\frac{1}{N}\frac{1}{\upsigma_x}\frac{1}{\upsigma_{\eta}}\sum_{i=1}^N K\left(\frac{X_{\text{synt,i}}^{\text{O}}-x}{\upsigma_x}\right)K\left(\frac{\omega_i^{\text{synt}}-\eta}{\upsigma_{\eta}}\right)}.
		\end{align}
		\end{comment}
		We note that $\widehat{\widehat{\mathbb{E}}}\left[Y_{\text{real}}|\text{do}(X_{\text{real}}=x)\right]$ and  $\widehat{\mathbb{E}}\left[Y_{\text{real}}|\text{do}(X_{\text{real}}=x)\right]$ are different empirical objects. In the former $Y_{\text{real}}^{\text{O}}$, available only in theory (as $\upeta$ is unknown), is generated from the same realization of $\upomega$ generated $X_{\text{real}}^{\text{O}}$, while in the latter $Y_{\text{real}}^{\text{CF}}$ is generated by $\widetilde{\upomega}$ (another realization of $\upomega$) which is independent of $X_{\text{real}}^{\text{O}}$. Consequently, $\widehat{\mathbb{E}}\left[Y_{\text{real}}|\text{do}(X_{\text{real}}=x)\right]$ is expected to be more accurate by not being limited to the the joint support of $\upomega$ and $X_{\text{real}}^{\text{O}}$.
		
		Equation \eqref{Partial:Means} utilizes $\upomega:=F_{\upeta}(\upeta)$ which is unobserved in the data. In case of strict monotonicity the following holds $F_{\upeta}(\upeta)=F_{X|Z}(X|Z)$ \citep{imbens2009identification}. Thus, one can construct $V_{\text{real,i}}:=\widehat{F}_{X|Z}(X_{\text{real,i}}^{\text{O}}|Z_i)$ to estimate $\upeta$ (up to normalization) and employ the partial means estimator as follows:
		\begin{align}
			& \widehat{\widehat{\mathbb{E}}}\left[Y_{\text{real}}|\text{do}(X_{\text{real}}=x)\right] =\frac{1}{N}\sum_{i=1}^N \widehat{\mathbb{E}}\left[Y_{\text{real}}^{\text{O}}|X_{\text{real}}^{\text{O}}=x, V_{\text{real}}=V_{\text{real,i}}\right]\label{Partial:Means:Imbens:Newey}
		\end{align}
		This is the benchmark used for comparison of our results. The reasons for choosing equation \eqref{Partial:Means:Imbens:Newey} for comparison are: (i) to show that our counterfactual goes beyond conditional quantiles of observables captured by the transformation above (violation of monotonicity) and further (ii) to demonstrate our approach performance in the presence of restricted support which is an acute problem in empirical research regardless of the presence or absence of monotonicity \citep{chernozhukov2020}.

		\section{Simulations}\label{Sec:Simulations}
		In what follows we verify our proposed estimator's performance by employing monte-carlo simulations. The section has two rules: the first is to verify our model performance when the underlying model is characterized by intrinsically non-monotonic structural equations for $X$ and $Y$. The second is to use various examples representing important and widely used economic models of supply and demand which lie in the heart of the economic science, such as: the ``backward-bending'' supply curve of labor \citep{Hanoch1965}; the Almost Ideal Demand System (AIDS)  \citep{deaton1980almost}; the Constant Elasticity of Substitution (CES) production function \citep{Sato1967} as well as the Translog utility function \citep{christensen1975}, which may also affect the demand and supply curves. A specific model of demand and supply functions have been discussed in \cite{blundell2014control} under the assumption of strict monotonicity in the unobserved non-additively-separable error terms. We however, attend to a more general (and non-monotonic) functional form for various economic phenomena.

		The various widely used economic models discussed above are estimated from data sets generated by the following data generation processes: 
		\begin{enumerate}%[\bf (i).]
			\item  Trans-log utility functions: \label{Translog:model}
			\begin{align}
				& h_1(Z,\upeta) =  \upalpha_0 + \upalpha_1 \log(Z)  + \upalpha_2\log(\upeta ) + \upalpha_3\log^2(Z)+ \upalpha_4\log^2(\upeta) +  \upalpha_5\log(Z)\log(\upeta)\\
				& g_1(X,\bm\epsilon) =  \upbeta_0 + \upbeta_1 \log(X)  + \upbeta_2\log(\epsilon ) + \upbeta_3\log^2(X)+ \upbeta_4\log^2(\epsilon) +  \upbeta_5\log(X)\log(\epsilon)
			\end{align}
			\item  Almost Ideal Demand System (AIDS) functions: \label{AIDS:model}
			\begin{align}
				&  h_2(Z,\upeta) =  h_1(Z,\upeta)+ \upalpha_6 (Z\upeta)^{\uprho}\\
				& g_2(X,\bm\epsilon)  =  g_1(X,\bm\epsilon)+ \upbeta_6 (X\epsilon)^{\uprho}
			\end{align}
			\item  Constant Elasticity of Substitution (CES) production functions: \label{CES:model}
			\begin{align}
				&  h_3(Z,\upeta) =  \upalpha_1\left[\upalpha_2 Z^{-\uprho}+\upalpha_3 \upeta^{-\uprho}\right]^{-\upalpha_4/\uprho}\\
				& g_3(X,\bm\epsilon)  =  \upbeta_1\left[\upbeta_2 X^{-\uprho}+\upbeta_3 \epsilon^{-\uprho}\right]^{-\upbeta_4/\uprho}
			\end{align}     
			\item  Backward-bending supply curves for $\upalpha_1=0, \upalpha_2=1$ and  $\upbeta_1=0, \upbeta_2=1$:  \label{Hanoch:model}
			\begin{align}
				&  h_4(Z,\upeta) = \exp(Z\upeta - \upalpha_1 Z)-\upalpha_2 (Z\upeta-\upalpha_1 Z)\\
				& g_4(X,\bm\epsilon)  =  \exp(X\epsilon-\upbeta_1 X)-\upbeta_2(X\epsilon-\upbeta_1 X)
			\end{align}     
			\item A generalized cobb-Douglas production function (Cobb-Douglas is obtained for $\upalpha_1=\upalpha_2=0$ and $\upbeta_1=\upbeta_2=0$): \label{CobbDouglas:model}
			\begin{align}
				&  h_5(Z,\upeta) =  \upalpha_1 Z + \upalpha_2\upeta +\upalpha_3 Z^{\upalpha_4}\upeta^{\upalpha_5}\\
				& g_5(X,\bm\epsilon)  =  \upalpha_1 X + \upalpha_2\epsilon +\upbeta_3 X^{\upbeta_4}\epsilon^{\upbeta_5}
			\end{align}     
			\item Nonlinearly non-additively-separable error terms \label{tanh:model}
			\begin{align}
				&  h_6(Z,\upeta) =  \upalpha_1\tanh(\upalpha_2 Z+\upalpha_3 \upeta)+\upalpha_4\tanh(\upalpha_5 \upeta)\\ 
				& g_6(X,\bm\epsilon)  =  \upbeta_1\tanh(\upbeta_2 Z+\upbeta_3 \epsilon)+\upbeta_4\tanh(\upbeta_5 \epsilon) 
			\end{align}     
			with $\tanh(x)=\frac{\exp(x)-\exp(-x)}{\exp(x)+\exp(-x)}$ commonly known as hyperbolic tangent.        
		\end{enumerate}
		The error terms are generated as follows: $\upeta, \nu\sim N(0,1)$ such that $\upeta\ci\nu$,
		\begin{align}
			\epsilon = \upgamma_1\upeta\nu + \upgamma_2\upeta + \upgamma_3\nu, \hspace{1em} \upgamma_1,\upgamma_2,\upgamma_3\in\left\{-3,-2, -1, 0, 1, 2, 3\right\}
		\end{align}

		We consider various specifications for the structural equations of $X$ and $Y$, each one is chosen from the following functions: the Trans-log utility functions in \ref{Translog:model}; the Almost Ideal Demand System (AIDS) functions in \ref{AIDS:model}; the Constant Elasticity of Substitution (CES) production functions in \ref{CES:model}; the backward-bending supply curves in  \ref{Hanoch:model}; a Cobb-Douglas production function in \ref{CobbDouglas:model} and the hyperbolic tangent model in \ref{tanh:model}. For sake of generality, for each of these model specifications, we verify our model performance under different Neural network architectures, such that for each given architecture we generated $100$ data sets. These data sets are used to calculate the means and the standard error of our estimates. 
		Such depth of simulations can be carried out only by parallel computing that can handle very demanding computations at a reasonable time in order to explore many specifications and variations of data generation processes.

		Each of tables \ref{tab:CES:Triang:Average}-\ref{tab:BackBending:Average} (appendix \ref{Appendix:Tables}) summarizes the results obtained from Monte-Carlo simulations generating $100$ data sets, each consisting of $10,000$ observations from specified structural model (data generation process) includes the following columns: a specific quantile of the action variable (column ($a$));  the mean value (over data sets) of the action variable in this quantile (column ($b$)); the mean (over data sets) of the estimated expected counterfactual outcome (given the action variable) using the real and synthetic data in (column ($c$)) and (column $(d)$), respectively; the mean (over data sets) of the estimated expected outcome (given the action variable) using the real data in (column $(e)$) and (column $(f)$), respectively; 
		the control variable results \citep{imbens2009identification} represent the counterfactual expected outcome that would be obtained by evaluating the latent context variable through the imposition of strict monotonicity (column $(g)$); and the partial means results \citep{newey1994kernel} represent the counterfactual expected outcome that would be obtained if the context were known (column $(h)$);
		
		We quantify the proximity of the synthetic data set generated by the synthetic generator functions $\widetilde{g}$ and $\widetilde{h}$ to the real data set generated by the ``true'' generator functions, $g$ and $h$. For tractability, our analysis disentangles the restricted support and divergence from strict monotonicity issues by using the following four criteria: 
		\begin{enumerate}
			\item \textbf{Observed Similarity} between $\widehat{\mathbb{E}}\left[Y_{\text{real}}^{\text{O}}|X_{\text{real}}^{\text{O}}=x\right]$ and $\widehat{\mathbb{E}}\left[Y_{\text{synt}}^{\text{O}}|X_{\text{synt}}^{\text{O}}=x\right]$ in \eqref{real:O} and \eqref{synthetic:O}, respectively.
			It  measures the resemblance between the estimates of the \textit{observed} expected outcome of $Y$ (at each quantile of $X$) in the synthetic and real data sets.\label{item:Observed:Similarity} 
			
			\item \textbf{Counterfactual Similarity} between  $\widehat{\mathbb{E}}\left[Y_{\text{real}}^{\text{CF}}|X_{\text{real}}^{\text{CF}}=x\right]$ and $\widehat{\mathbb{E}}\left[Y_{\text{synt}}^{\text{CF}}|X_{\text{synt}}^{\text{CF}}=x\right]$ in \eqref{real:CF} and \eqref{synthetic:CF}, respectively.
			It measures the resemblance between the estimates of the \textit{counterfactual} expected outcome of $Y$ (at each quantile of $X$) in the synthetic and real data sets.\label{item:CF:Similarity}
			
			\item \textbf{Strict monotonicity} by comparing the estimated counterfactual results in \eqref{Partial:Means}, which does not rely on monotonicity, to the partial means estimate under monotonicity in \eqref{Partial:Means:Imbens:Newey}. Both of these estimators utilize the restricted empirical support only.
			The degree of non-monotonicity is quantified in these simulations by comparing the results that would have been obtained if $\upomega$ were known (the partial means estimator) to the actual results obtained by using $\widehat{\upomega}=$, an estimate of it (the control variable estimator). Consequently, a negligible disparity between the estimates of the control variable and the partial means indicate that the monotonicity assumption could have been imposed.\footnote{Namely, the control variable estimator, $\widehat{V}_i=\widehat{F}_{X|Z}(x_i|z_i)$, is an estimate of the realization of $\upomega$ in the $i$th observation, whereas  the partial means estimator, $V_i=F_{\upeta}(\eta_i)$, is the actual realization of $\omega$ in the $i$th observation. The imposition of strict monotonicity is required as in practice $\omega$ is unknown and thus it must be evaluated.} \label{item:MONO:Similarity}
			\item \textbf{Restricted support} by comparing our real data counterfactual estimate in \eqref{real:CF} utilizing the unrestricted support to the partial means in \eqref{Partial:Means} utilizing the restricted support. Both of these estimators do not rely on monotonicity. However, the partial means estimator is available only in theory in which latent context is treated as known. 
			Any disparity between the partial means and the real counterfactual expected outcomes are related to discrepancies between the empirical and theoretical support of the context and action variable. 
			\label{item:MEANS:Similarity}              
		\end{enumerate}

		In the theoretical model we have shown that the importance of allowing for measure-preserving transformations when only distributional similarity is required. That is, we are interested in measuring the similarity between the empirical distributions of two sequences of estimates collected from the Monte Carlo simulations. Each sequence belongs to a different estimator. Thus, in the empirical application we opted for the Sliced-Wasserstein Distance (SWD) \citep{manole2019minimax}, because it is designated for the comparison of two empirical distributions in a non-parametric manner. For robustness, these distributions must be compared over quantiles in the interior of the distribution to avoid the imposition of parametric assumptions regarding the tail behaviour of each distribution. Thus, we specify a scalar $\Lambda\in(0,0.5)$ for which $[\Lambda,1-\Lambda]$ denotes an interval of quantiles in the interior of the distribution. Using the aforementioned (SWD) methodology we append an asterisk and a dagger to indicate sequences of estimates with similar empirical distributions (over the Monte Carlo simulations) to those of the real counterfactual and the real observed results, respectively. Consequently, we perform these similarity tests by comparing distributions using only quantiles in the interval $[0.05,0.95]$ at significance level of $95\%$.
		
		Summary statistics of results are presented in tables \ref{tab:CES:Triang:Average}-\ref{tab:BackBending:Average} (appendix \ref{Appendix:Tables}). Entries represent average (standard error) of estimates obtained from $100$ samples of $10,000$ observations each. The overall picture emerging from the simulations, as appears in the tables, points to the superiority of our proposed casual inference. This is apparent from the similarity (in the SWD sense) between the real (column $(e)$) and synthetic (column ($f$)) counterfactual results, which appears in the various tables representing different known economics models. In particular, when juxtaposed on the much different and inferior results of existing models as shown in columns ($g$) and ($h$). Basically, existing causality models hardly successfully reconstruct the real counterfactual results. Further, even when choosing a monotonic structural form like the CES in table \ref{tab:CES:Triang:Average} (appendix \ref{Appendix:Tables}), results when applying conventional casual inference (column ($g$) and ($h$)) are still inaccurate due to an improper restricted support domain. As can be be seen, our model produces real and synthetic counterfactual results (in bold) which obey SWD similarity. For instance, in the 80th quantile of $X$, we report an estimate of (standard error) (table \ref{tab:CES:Triang:Average}, column ($e$)) 13.94 (0.344) for the real counterfactual and (table \ref{tab:CES:Triang:Average}, column ($f$)) 14.13 (0.553) for the synthetic counterfactual. The indication that the empirical support domain is restricted is based on the fact that the real counterfactual estimates are SWD different from the partial means estimates (as described in criterion \ref{item:MEANS:Similarity}). Just as an example, in the 90th quantile of $X$, we report an estimate of (table \ref{tab:CES:Triang:Average}, column ($e$)) 15.73 (0.728) for the real counterfactual and (table \ref{tab:CES:Triang:Average}, column $(h)$) 20.43 (0.960) for the partial means. The high similarity between the average estimates of the control variable (column $(g)$) and the partial means (column $(h)$) indicates that indeed there is no statistical evidence for the violation of strict monotonicity (as described in criterion \ref{item:MONO:Similarity}). For instance, in the 5th quantile of $X$, we get the estimates (standard error) of (table \ref{tab:CES:Triang:Average}, column $(g)$) 6.98 (0.338) for the control variable and (table \ref{tab:CES:Triang:Average}, column $(h)$) 6.96 (0.276) for the partial means; in the 25th quantile of $X$, the estimate is (table \ref{tab:CES:Triang:Average}, column $(g)$) 8.52 (0.213) for the control variable and (table \ref{tab:CES:Triang:Average}, column $(h)$) 8.51 (0.188) for the partial means; in the 75th quantile of $X$, it is (table \ref{tab:CES:Triang:Average}, column $(g)$) 15.58 (0.338) for the control variable and (table \ref{tab:CES:Triang:Average}, column $(h)$) 15.59 (0.276) for the partial means. 
		
		In table \ref{tab:tranglog:Average} (appendix \ref{Appendix:Tables}) we apply the simulations to the known Translog function which is non-monotonic. The synthetic counterfactual results are not statistically different than the real counterfactual results. For instance, in the 25th quantile of $X$, we report an estimate of (table \ref{tab:tranglog:Average}, column ($e$)) 16.63 (0.175) for the real counterfactual and (table \ref{tab:tranglog:Average}, column ($f$)) 16.61 (0.407) for the synthetic counterfactual. However, the common practice results are extremely inaccurate in magnitudes (columns $(g)$ and $(h)$). In fact, the non-monotonicity is expressed by the disparity between the control function and the partial means estimates. Just as an example, in the 20th quantile of $X$, the estimate is (table \ref{tab:tranglog:Average}, column $(g)$) 20.95 (1.087) for the control variable and (table \ref{tab:tranglog:Average}, column $(h)$) 14.71 (0.857) for the partial means. In this specific quantile there is also a disparity between the empirical and restricted supports, captured by the gap between real counterfactual and the partial means estimates. For instance, in the 15th quantile of $X$, it is (table \ref{tab:tranglog:Average}, column ($e$)) 15.88 (0.190) for the real counterfactual and (table \ref{tab:tranglog:Average}, column $(h)$) 13.53 (0.959) for the partial means.
		
		The entries in table \ref{tab:tanhtanh:Average} (appendix \ref{Appendix:Tables}) the estimates of the hyperbolic tangent form where the strict monotonicity is violated. Applying the our offered model produces very accurate results and exhibits insignificant statistical difference between the real and synthetic interventional distributions (columns $(e)$ and $(f)$). However, the common practice results are extremely inaccurate in both magnitudes and signs (columns $(g)$ and $(h)$). The violation of monotonicity is expressed by a strong disparity between the control function and the partial means estimates. For instance, the 20th quantile of $X$, we report an estimate of (table \ref{tab:tanhtanh:Average}, column $(g)$) -0.18 (0.291) for the control variable and (table \ref{tab:tanhtanh:Average}, column $(h)$) -0.56 (0.189) for the partial means. There is also a disparity between the empirical and restricted supports, captured by the gap between real counterfactual and the partial means estimates. E.g., in the 15th quantile of $X$, we report an estimate of (table \ref{tab:tanhtanh:Average}, column ($e$)) -0.17 (0.038) for the real counterfactual and (table \ref{tab:tanhtanh:Average}, column $(h)$) -0.70 (0.212) for the partial means.
		
		The entries in table \ref{tab:AIDS:Average} (appendix \ref{Appendix:Tables}) exhibit estimates of the well-known demand function (AIDS). This very well-known functional form is intrinsically non-monotonic. Applying the our offered model produces very accurate results and exhibits insignificant statistical difference between the real and synthetic interventional distributions (columns $(e)$ and $(f)$). For instance, in the 65th quantile of $X$, we report an estimate of (table \ref{tab:AIDS:Average}, column ($e$)) 16.71 (0.188) for the real counterfactual and (table \ref{tab:AIDS:Average}, column ($f$)) 16.65 (0.348) for the synthetic counterfactual. However, as can be seen results largely demonstrate the failure to mimic the real counterfactual distribution when applying conventional casual inference (column ($g$) and ($h$)). For instance, in the 90th quantile of $X$, we report an estimate of (table \ref{tab:AIDS:Average}, column ($e$)) 22.79 (0.669) for the real counterfactual and (table \ref{tab:AIDS:Average}, column $(h)$) 26.58 (0.975) for the partial means; For this quantile, the synthetic counterfactual result (table \ref{tab:AIDS:Average}, column ($f$)) is 21.66 (0.758). For instance, in the 55th quantile of $X$, we report an estimate of (table \ref{tab:AIDS:Average}, column $(g)$) 21.16 (1.101) for the control variable and (table \ref{tab:AIDS:Average}, column $(h)$) 16.67 (0.748) for the partial means; For this quantile, the real counterfactual result is (table \ref{tab:AIDS:Average}, column ($e$)) 16.06 (0.195) and the synthetic counterfactual result is (table \ref{tab:AIDS:Average}, column ($f$)) 16.09 (0.344). There is also a disparity between the empirical and restricted supports, captured by the gap between real counterfactual and the partial means estimates.
		
		The entries in table \ref{tab:BackBending:Average} (appendix \ref{Appendix:Tables}) represent those emerging from the estimation of backward-bending supply curve. In this instance, monotonicity is intrinsically violated. This is attested to the disparity between the control function and the partial means estimates in both magnitude as well as in sign, which in reality may lead to wrong public policy. For instance, in the 35th quantile of $X$, we report an estimate of (table \ref{tab:BackBending:Average}, column $(g)$) -0.33 (0.135) for the control variable and (table \ref{tab:BackBending:Average}, column $(h)$) 0.33 (0.091) for the partial means. Applying our offered model produces very accurate results and exhibits insignificant statistical difference between the real and synthetic interventional distributions (columns $(e)$ and $(f)$). Just as an example, in the 60th quantile of $X$, the real counterfactual result (table \ref{tab:BackBending:Average}, column $(e)$) is 0.61 (0.040) and the synthetic counterfactual result (table \ref{tab:BackBending:Average}, column $(f)$) is 0.59 (0.093). Further, below the 40’th quantile it can be seen that there is a disparity between the empirical and restricted supports, captured by the gap between real counterfactual and the partial means estimates. In another instance, in the 25th quantile of $X$, we report an estimate of (table \ref{tab:BackBending:Average}, column ($e$)) 0.47 (0.033) for the real counterfactual and (table \ref{tab:BackBending:Average}, column $(h)$) 0.21 (0.087) for the partial means. The synthetic counterfactual result is (table \ref{tab:BackBending:Average}, column ($f$)) 0.44 (0.109) indicating that it is not statistically different than the real result.
		
		To sum up, the observed and counterfactual similarity criteria \ref{item:Observed:Similarity} and \ref{item:CF:Similarity}  tables \ref{tab:CES:Triang:Average}-\ref{tab:BackBending:Average} (appendix \ref{Appendix:Tables}) developed in this paper demonstrate the high accuracy of the synthetic counterfactual machinery relative to the real data counterfactual result. In terms of sensitivity to the violation of the strict monotonicity assumption, as captured by criterion \ref{item:MONO:Similarity}, our proposed estimator produces accurate results regardless of the absence or presence of monotonicity and outperforms control variable existing approach. In terms of sensitivity to the violation of unrestricted support domain, as captured by criterion \ref{item:MEANS:Similarity}, our proposed estimator outperforms both the partial means as well as the control variable estimators, by mimicking the results obtained for the real data counterfactual outcome. It emphasized the results generated by the aforementioned existing models produce results that are inaccurate both in magnitude and sign of the estimates.\footnote{In the supplementary material we provide different stability measures.} %We highlight the stability of these results to randomly chosen architecture of the Neural network.   % and to various sample sizes. 

		\section{Conclusion}\label{Sec:Conclusion}
		Different latent contexts may affect the relations between cause and outcome and hence can have significant ramification for causal relations and inference. 
		We develop a novel identification strategy as well as a new estimator for triangular models in the presence of non-separable disturbances, which unlike the common practice, does not rely on the strict monotonicity assumption. The key result of this identifiability approach is the explicit characterization of the distributional relationship between the latent context variable and the vector of observables through Fredholm integral equations governed by generator functions and an \textit{unknown kernel function}, inducing a non-monotonic inverse problem.  The role of these generator functions is two-fold: (i) to characterize the unknown kernel function induced by these generator functions and (ii) to ensure that the estimator of the unknown quantity is a continuous function of the data given this unknown kernel function. This very formulation facilitates the establishment of uniqueness of the interventional distribution given the observables. 
		Furthermore, we develop a novel CONGAN estimator based on a feed-forward Neural network architecture generating a synthetic counterfactual distribution. This synthetic distribution represents various combinations of actions, outcomes and contexts, rarely available in finite samples. In the simulations, the proposed estimator's performance in finite samples has been validated in several aspects by testing various data generation processes each associated with the commonly employed structural functions: Cobb-Douglas, AIDS, CES, Translog and backward-bending supply models. It can be seen that our model generates synthetic data mimicking the behavior in the real data (in terms of SWD similarity). This comparison has been done for different quantiles of the action variable. Our results are also compared to both the monotonic control variable as well as partial means estimators. In some of these models, the counterfactual results obtained by the conventionally used partial means estimator, largely deviate from the benchmark real data results in quantity and sign. This may in reality have important ramifications as to the proper public policy.
		
		In future work, our framework can be incorporated in dynamic models as well as panel models to capture the dependence of the action and outcome on past events through a different Neural network architecture, such as recurrent Neural networks. %Incorporating this feature, enables to achieve more general identifiability results than the prevailing ones, through the removal of various structural form restrictions in a similar manner to the approach taken in the present model.%by letting the fixed / random effects to be stochastic functions of the observables, unlike the prevailing approach which imposes a variant monotonicity.
		
		\pagebreak
		\appendix	
		\renewcommand{\thesection}{\Alph{section}.\arabic{section}}
		%	\addcontentsline{toc}{section}{Appendices} % ?
		%\appendixpage
		%\fancyhead{}
		%\pagestyle{fancy}
		%\section*{Appendices}
		%\renewcommand{\thesection}{A}
		%\renewcommand{\thesubsection}{A}
		%\setcounter{section}{0}
		
							\bibliography{congan}               

\begin{thebibliography}{}

\bibitem [\protect \citeauthoryear {%
Afriat%
}{%
Afriat%
}{%
{\protect \APACyear {1967}}%
}]{%
afriat1967}
\APACinsertmetastar {%
afriat1967}%
\begin{APACrefauthors}%
Afriat, S\BPBI N.%
\end{APACrefauthors}%
\unskip\
\newblock
\APACrefYearMonthDay{1967}{}{}.
\newblock
{\BBOQ}\APACrefatitle {The construction of utility functions from expenditure
  data} {The construction of utility functions from expenditure data}.{\BBCQ}
\newblock
\APACjournalVolNumPages{International economic review}{8}{1}{67--77}.
\PrintBackRefs{\CurrentBib}

\bibitem [\protect \citeauthoryear {%
Agarwal%
\ \protect \BOthers {.}}{%
Agarwal%
\ \protect \BOthers {.}}{%
{\protect \APACyear {2020}}%
}]{%
agarwal2020searching}
\APACinsertmetastar {%
agarwal2020searching}%
\begin{APACrefauthors}%
Agarwal, S.%
, Grigsby, J.%
, Horta{\c{c}}su, A.%
, Matvos, G.%
, Seru, A.%
\BCBL {}\ \BBA {} Yao, V.%
\end{APACrefauthors}%
\unskip\
\newblock
\APACrefYearMonthDay{2020}{}{}.
\newblock
{\BBOQ}\APACrefatitle {Searching for Approval} {Searching for approval}.{\BBCQ}
\newblock
\APACaddressPublisher{}{NBER Working Papers 27341. Available at:
  \url{http://www.nber.org/papers/w27341}}.
\PrintBackRefs{\CurrentBib}

\bibitem [\protect \citeauthoryear {%
Alamatsaz%
}{%
Alamatsaz%
}{%
{\protect \APACyear {1983}}%
}]{%
alamatsaz1983completeness}
\APACinsertmetastar {%
alamatsaz1983completeness}%
\begin{APACrefauthors}%
Alamatsaz, M.%
\end{APACrefauthors}%
\unskip\
\newblock
\APACrefYearMonthDay{1983}{}{}.
\newblock
{\BBOQ}\APACrefatitle {Completeness and self-decomposability of mixtures}
  {Completeness and self-decomposability of mixtures}.{\BBCQ}
\newblock
\APACjournalVolNumPages{Annals of the Institute of Statistical
  Mathematics}{35}{}{355--363}.
\PrintBackRefs{\CurrentBib}

\bibitem [\protect \citeauthoryear {%
Altonji%
\ \BBA {} Matzkin%
}{%
Altonji%
\ \BBA {} Matzkin%
}{%
{\protect \APACyear {2005}}%
}]{%
altonji2005cross}
\APACinsertmetastar {%
altonji2005cross}%
\begin{APACrefauthors}%
Altonji, J\BPBI G.%
\BCBT {}\ \BBA {} Matzkin, R\BPBI L.%
\end{APACrefauthors}%
\unskip\
\newblock
\APACrefYearMonthDay{2005}{}{}.
\newblock
{\BBOQ}\APACrefatitle {Cross section and panel data estimators for nonseparable
  models with endogenous regressors} {Cross section and panel data estimators
  for nonseparable models with endogenous regressors}.{\BBCQ}
\newblock
\APACjournalVolNumPages{Econometrica}{73}{4}{1053--1102}.
\PrintBackRefs{\CurrentBib}

\bibitem [\protect \citeauthoryear {%
Angrist%
, Imbens%
\BCBL {}\ \BBA {} Rubin%
}{%
Angrist%
\ \protect \BOthers {.}}{%
{\protect \APACyear {1996}}%
}]{%
angrist1996identification}
\APACinsertmetastar {%
angrist1996identification}%
\begin{APACrefauthors}%
Angrist, J\BPBI D.%
, Imbens, G\BPBI W.%
\BCBL {}\ \BBA {} Rubin, D\BPBI B.%
\end{APACrefauthors}%
\unskip\
\newblock
\APACrefYearMonthDay{1996}{}{}.
\newblock
{\BBOQ}\APACrefatitle {Identification of causal effects using instrumental
  variables} {Identification of causal effects using instrumental
  variables}.{\BBCQ}
\newblock
\APACjournalVolNumPages{Journal of the American statistical
  Association}{91}{434}{444--455}.
\PrintBackRefs{\CurrentBib}

\bibitem [\protect \citeauthoryear {%
Athey%
\ \BBA {} Imbens%
}{%
Athey%
\ \BBA {} Imbens%
}{%
{\protect \APACyear {2016}}%
}]{%
athey2016recursive}
\APACinsertmetastar {%
athey2016recursive}%
\begin{APACrefauthors}%
Athey, S.%
\BCBT {}\ \BBA {} Imbens, G.%
\end{APACrefauthors}%
\unskip\
\newblock
\APACrefYearMonthDay{2016}{}{}.
\newblock
{\BBOQ}\APACrefatitle {Recursive partitioning for heterogeneous causal effects}
  {Recursive partitioning for heterogeneous causal effects}.{\BBCQ}
\newblock
\APACjournalVolNumPages{Proceedings of the National Academy of
  Sciences}{113}{27}{7353--7360}.
\PrintBackRefs{\CurrentBib}

\bibitem [\protect \citeauthoryear {%
Athey%
, Imbens%
, Metzger%
\BCBL {}\ \BBA {} Munro%
}{%
Athey%
\ \protect \BOthers {.}}{%
{\protect \APACyear {2020}}%
}]{%
athey2020}
\APACinsertmetastar {%
athey2020}%
\begin{APACrefauthors}%
Athey, S.%
, Imbens, G.%
, Metzger, J.%
\BCBL {}\ \BBA {} Munro, E.%
\end{APACrefauthors}%
\unskip\
\newblock
\APACrefYearMonthDay{2020}{}{}.
\newblock
{\BBOQ}\APACrefatitle {Using Wasserstein Generative Adversarial Networks for
  the Design of Monte Carlo Simulations} {Using wasserstein generative
  adversarial networks for the design of monte carlo simulations}.{\BBCQ}
\newblock
\APACjournalVolNumPages{arXiv preprint arXiv:1909.02210}{}{}{}.
\PrintBackRefs{\CurrentBib}

\bibitem [\protect \citeauthoryear {%
Blundell%
\ \BBA {} Matzkin%
}{%
Blundell%
\ \BBA {} Matzkin%
}{%
{\protect \APACyear {2014}}%
}]{%
blundell2014control}
\APACinsertmetastar {%
blundell2014control}%
\begin{APACrefauthors}%
Blundell, R\BPBI W.%
\BCBT {}\ \BBA {} Matzkin, R\BPBI L.%
\end{APACrefauthors}%
\unskip\
\newblock
\APACrefYearMonthDay{2014}{}{}.
\newblock
{\BBOQ}\APACrefatitle {Control functions in nonseparable simultaneous equations
  models} {Control functions in nonseparable simultaneous equations
  models}.{\BBCQ}
\newblock
\APACjournalVolNumPages{Quantitative Economics}{5}{2}{271--295}.
\PrintBackRefs{\CurrentBib}

\bibitem [\protect \citeauthoryear {%
Carleson%
}{%
Carleson%
}{%
{\protect \APACyear {1966}}%
}]{%
carleson1966convergence}
\APACinsertmetastar {%
carleson1966convergence}%
\begin{APACrefauthors}%
Carleson, L.%
\end{APACrefauthors}%
\unskip\
\newblock
\APACrefYearMonthDay{1966}{}{}.
\newblock
{\BBOQ}\APACrefatitle {On convergence and growth of partial sums of Fourier
  series} {On convergence and growth of partial sums of fourier series}.{\BBCQ}
\newblock
\APACjournalVolNumPages{Acta Mathematica}{116}{}{135--157}.
\PrintBackRefs{\CurrentBib}

\bibitem [\protect \citeauthoryear {%
Chernozhukov%
, Fern{\'a}ndez-Val%
\BCBL {}\ \BBA {} Melly%
}{%
Chernozhukov%
\ \protect \BOthers {.}}{%
{\protect \APACyear {2013}}%
}]{%
chernozhukov2013inference}
\APACinsertmetastar {%
chernozhukov2013inference}%
\begin{APACrefauthors}%
Chernozhukov, V.%
, Fern{\'a}ndez-Val, I.%
\BCBL {}\ \BBA {} Melly, B.%
\end{APACrefauthors}%
\unskip\
\newblock
\APACrefYearMonthDay{2013}{}{}.
\newblock
{\BBOQ}\APACrefatitle {Inference on counterfactual distributions} {Inference on
  counterfactual distributions}.{\BBCQ}
\newblock
\APACjournalVolNumPages{Econometrica}{81}{6}{2205--2268}.
\PrintBackRefs{\CurrentBib}

\bibitem [\protect \citeauthoryear {%
Chernozhukov%
, Fern{\'a}ndez-Val%
, Newey%
, Stouli%
\BCBL {}\ \BBA {} Vella%
}{%
Chernozhukov%
\ \protect \BOthers {.}}{%
{\protect \APACyear {2020}}%
}]{%
chernozhukov2020}
\APACinsertmetastar {%
chernozhukov2020}%
\begin{APACrefauthors}%
Chernozhukov, V.%
, Fern{\'a}ndez-Val, I.%
, Newey, W.%
, Stouli, S.%
\BCBL {}\ \BBA {} Vella, F.%
\end{APACrefauthors}%
\unskip\
\newblock
\APACrefYearMonthDay{2020}{}{}.
\newblock
{\BBOQ}\APACrefatitle {Semiparametric estimation of structural functions in
  nonseparable triangular models} {Semiparametric estimation of structural
  functions in nonseparable triangular models}.{\BBCQ}
\newblock
\APACjournalVolNumPages{Quantitative Economics}{11}{2}{503--533}.
\PrintBackRefs{\CurrentBib}

\bibitem [\protect \citeauthoryear {%
Chernozhukov%
, Imbens%
\BCBL {}\ \BBA {} Newey%
}{%
Chernozhukov%
\ \protect \BOthers {.}}{%
{\protect \APACyear {2007}}%
}]{%
chernozhukov2007instrumental}
\APACinsertmetastar {%
chernozhukov2007instrumental}%
\begin{APACrefauthors}%
Chernozhukov, V.%
, Imbens, G\BPBI W.%
\BCBL {}\ \BBA {} Newey, W\BPBI K.%
\end{APACrefauthors}%
\unskip\
\newblock
\APACrefYearMonthDay{2007}{}{}.
\newblock
{\BBOQ}\APACrefatitle {Instrumental variable estimation of nonseparable models}
  {Instrumental variable estimation of nonseparable models}.{\BBCQ}
\newblock
\APACjournalVolNumPages{Journal of Econometrics}{139}{1}{4--14}.
\PrintBackRefs{\CurrentBib}

\bibitem [\protect \citeauthoryear {%
Chesher%
}{%
Chesher%
}{%
{\protect \APACyear {2003}}%
}]{%
chesher2003identification}
\APACinsertmetastar {%
chesher2003identification}%
\begin{APACrefauthors}%
Chesher, A.%
\end{APACrefauthors}%
\unskip\
\newblock
\APACrefYearMonthDay{2003}{}{}.
\newblock
{\BBOQ}\APACrefatitle {Identification in nonseparable models} {Identification
  in nonseparable models}.{\BBCQ}
\newblock
\APACjournalVolNumPages{Econometrica}{71}{5}{1405--1441}.
\PrintBackRefs{\CurrentBib}

\bibitem [\protect \citeauthoryear {%
Chetverikov%
}{%
Chetverikov%
}{%
{\protect \APACyear {2019}}%
}]{%
chetverikov2019testing}
\APACinsertmetastar {%
chetverikov2019testing}%
\begin{APACrefauthors}%
Chetverikov, D.%
\end{APACrefauthors}%
\unskip\
\newblock
\APACrefYearMonthDay{2019}{}{}.
\newblock
{\BBOQ}\APACrefatitle {Testing regression monotonicity in econometric models}
  {Testing regression monotonicity in econometric models}.{\BBCQ}
\newblock
\APACjournalVolNumPages{Econometric Theory}{35}{4}{729--776}.
\PrintBackRefs{\CurrentBib}

\bibitem [\protect \citeauthoryear {%
Christensen%
, Jorgenson%
\BCBL {}\ \BBA {} Lau%
}{%
Christensen%
\ \protect \BOthers {.}}{%
{\protect \APACyear {1975}}%
}]{%
christensen1975}
\APACinsertmetastar {%
christensen1975}%
\begin{APACrefauthors}%
Christensen, L\BPBI R.%
, Jorgenson, D\BPBI W.%
\BCBL {}\ \BBA {} Lau, L\BPBI J.%
\end{APACrefauthors}%
\unskip\
\newblock
\APACrefYearMonthDay{1975}{}{}.
\newblock
{\BBOQ}\APACrefatitle {Transcendental logarithmic utility functions}
  {Transcendental logarithmic utility functions}.{\BBCQ}
\newblock
\APACjournalVolNumPages{The American Economic Review}{65}{3}{367--383}.
\PrintBackRefs{\CurrentBib}

\bibitem [\protect \citeauthoryear {%
Darolles%
, Fan%
, Florens%
\BCBL {}\ \BBA {} Renault%
}{%
Darolles%
\ \protect \BOthers {.}}{%
{\protect \APACyear {2011}}%
}]{%
darolles2011nonparametric}
\APACinsertmetastar {%
darolles2011nonparametric}%
\begin{APACrefauthors}%
Darolles, S.%
, Fan, Y.%
, Florens, J\BHBI P.%
\BCBL {}\ \BBA {} Renault, E.%
\end{APACrefauthors}%
\unskip\
\newblock
\APACrefYearMonthDay{2011}{}{}.
\newblock
{\BBOQ}\APACrefatitle {Nonparametric instrumental regression} {Nonparametric
  instrumental regression}.{\BBCQ}
\newblock
\APACjournalVolNumPages{Econometrica}{79}{5}{1541--1565}.
\PrintBackRefs{\CurrentBib}

\bibitem [\protect \citeauthoryear {%
Deaton%
\ \BBA {} Muellbauer%
}{%
Deaton%
\ \BBA {} Muellbauer%
}{%
{\protect \APACyear {1980}}%
}]{%
deaton1980almost}
\APACinsertmetastar {%
deaton1980almost}%
\begin{APACrefauthors}%
Deaton, A.%
\BCBT {}\ \BBA {} Muellbauer, J.%
\end{APACrefauthors}%
\unskip\
\newblock
\APACrefYearMonthDay{1980}{}{}.
\newblock
{\BBOQ}\APACrefatitle {An almost ideal demand system} {An almost ideal demand
  system}.{\BBCQ}
\newblock
\APACjournalVolNumPages{The American economic review}{70}{3}{312--326}.
\PrintBackRefs{\CurrentBib}

\bibitem [\protect \citeauthoryear {%
Dembo%
, Kariv%
, Polisson%
\BCBL {}\ \BBA {} Quah%
}{%
Dembo%
\ \protect \BOthers {.}}{%
{\protect \APACyear {2021}}%
}]{%
dembo2021ever}
\APACinsertmetastar {%
dembo2021ever}%
\begin{APACrefauthors}%
Dembo, A.%
, Kariv, S.%
, Polisson, M.%
\BCBL {}\ \BBA {} Quah, J\BPBI K\BHBI H.%
\end{APACrefauthors}%
\unskip\
\newblock
\APACrefYearMonthDay{2021}{{\APACmonth{05}}}{}.
\newblock
\APACrefbtitle {Ever Since Allais} {Ever since allais}\ \APACbVolEdTR
  {}{Bristol Economics Discussion Papers\ \BNUM\ 21/745}.
\newblock
\APACaddressInstitution{}{School of Economics, University of Bristol, UK}.
\newblock
\begin{APACrefURL} \url{https://ideas.repec.org/p/bri/uobdis/21-745.html}
  \end{APACrefURL}
\PrintBackRefs{\CurrentBib}

\bibitem [\protect \citeauthoryear {%
Farrell%
, Liang%
\BCBL {}\ \BBA {} Misra%
}{%
Farrell%
\ \protect \BOthers {.}}{%
{\protect \APACyear {2021}}%
}]{%
farrell2021deep}
\APACinsertmetastar {%
farrell2021deep}%
\begin{APACrefauthors}%
Farrell, M\BPBI H.%
, Liang, T.%
\BCBL {}\ \BBA {} Misra, S.%
\end{APACrefauthors}%
\unskip\
\newblock
\APACrefYearMonthDay{2021}{}{}.
\newblock
{\BBOQ}\APACrefatitle {Deep neural networks for estimation and inference} {Deep
  neural networks for estimation and inference}.{\BBCQ}
\newblock
\APACjournalVolNumPages{Econometrica}{89}{1}{181--213}.
\PrintBackRefs{\CurrentBib}

\bibitem [\protect \citeauthoryear {%
Goodfellow%
\ \protect \BOthers {.}}{%
Goodfellow%
\ \protect \BOthers {.}}{%
{\protect \APACyear {2014}}%
}]{%
goodfellow2014generative}
\APACinsertmetastar {%
goodfellow2014generative}%
\begin{APACrefauthors}%
Goodfellow, I.%
, Pouget-Abadie, J.%
, Mirza, M.%
, Xu, B.%
, Warde-Farley, D.%
, Ozair, S.%
\BDBL {}Bengio, Y.%
\end{APACrefauthors}%
\unskip\
\newblock
\APACrefYearMonthDay{2014}{}{}.
\newblock
{\BBOQ}\APACrefatitle {Generative adversarial nets} {Generative adversarial
  nets}.{\BBCQ}
\newblock
\BIn{} \APACrefbtitle {Advances in neural information processing systems}
  {Advances in neural information processing systems}\ (\BPGS\ 2672--2680).
\PrintBackRefs{\CurrentBib}

\bibitem [\protect \citeauthoryear {%
Gulrajani%
, Ahmed%
, Arjovsky%
, Dumoulin%
\BCBL {}\ \BBA {} Courville%
}{%
Gulrajani%
\ \protect \BOthers {.}}{%
{\protect \APACyear {2017}}%
}]{%
gulrajani2017improved}
\APACinsertmetastar {%
gulrajani2017improved}%
\begin{APACrefauthors}%
Gulrajani, I.%
, Ahmed, F.%
, Arjovsky, M.%
, Dumoulin, V.%
\BCBL {}\ \BBA {} Courville, A.%
\end{APACrefauthors}%
\unskip\
\newblock
\APACrefYearMonthDay{2017}{}{}.
\newblock
{\BBOQ}\APACrefatitle {Improved training of wasserstein gans} {Improved
  training of wasserstein gans}.{\BBCQ}
\newblock
\APACjournalVolNumPages{arXiv preprint arXiv:1704.00028}{}{}{}.
\PrintBackRefs{\CurrentBib}

\bibitem [\protect \citeauthoryear {%
Hall%
\ \BBA {} Morton%
}{%
Hall%
\ \BBA {} Morton%
}{%
{\protect \APACyear {1993}}%
}]{%
hall1993estimation}
\APACinsertmetastar {%
hall1993estimation}%
\begin{APACrefauthors}%
Hall, P.%
\BCBT {}\ \BBA {} Morton, S\BPBI C.%
\end{APACrefauthors}%
\unskip\
\newblock
\APACrefYearMonthDay{1993}{}{}.
\newblock
{\BBOQ}\APACrefatitle {On the estimation of entropy} {On the estimation of
  entropy}.{\BBCQ}
\newblock
\APACjournalVolNumPages{Annals of the Institute of Statistical
  Mathematics}{45}{1}{69--88}.
\PrintBackRefs{\CurrentBib}

\bibitem [\protect \citeauthoryear {%
Hanoch%
}{%
Hanoch%
}{%
{\protect \APACyear {1965}}%
}]{%
Hanoch1965}
\APACinsertmetastar {%
Hanoch1965}%
\begin{APACrefauthors}%
Hanoch, G.%
\end{APACrefauthors}%
\unskip\
\newblock
\APACrefYearMonthDay{1965}{}{}.
\newblock
{\BBOQ}\APACrefatitle {The "Backward-bending" Supply of Labor} {The
  "backward-bending" supply of labor}.{\BBCQ}
\newblock
\APACjournalVolNumPages{Journal of Political Economy}{73}{6}{636--642}.
\PrintBackRefs{\CurrentBib}

\bibitem [\protect \citeauthoryear {%
Heckman%
\ \BBA {} Pinto%
}{%
Heckman%
\ \BBA {} Pinto%
}{%
{\protect \APACyear {2018}}%
}]{%
heckman2018unordered}
\APACinsertmetastar {%
heckman2018unordered}%
\begin{APACrefauthors}%
Heckman, J\BPBI J.%
\BCBT {}\ \BBA {} Pinto, R.%
\end{APACrefauthors}%
\unskip\
\newblock
\APACrefYearMonthDay{2018}{}{}.
\newblock
{\BBOQ}\APACrefatitle {Unordered monotonicity} {Unordered monotonicity}.{\BBCQ}
\newblock
\APACjournalVolNumPages{Econometrica}{86}{1}{1--35}.
\PrintBackRefs{\CurrentBib}

\bibitem [\protect \citeauthoryear {%
Heckman%
\ \BBA {} Robb%
}{%
Heckman%
\ \BBA {} Robb%
}{%
{\protect \APACyear {1985}}%
}]{%
heckman1985alternative}
\APACinsertmetastar {%
heckman1985alternative}%
\begin{APACrefauthors}%
Heckman, J\BPBI J.%
\BCBT {}\ \BBA {} Robb, R.%
\end{APACrefauthors}%
\unskip\
\newblock
\APACrefYearMonthDay{1985}{}{}.
\newblock
{\BBOQ}\APACrefatitle {Alternative methods for evaluating the impact of
  interventions: An overview} {Alternative methods for evaluating the impact of
  interventions: An overview}.{\BBCQ}
\newblock
\APACjournalVolNumPages{Journal of econometrics}{30}{1}{239--267}.
\PrintBackRefs{\CurrentBib}

\bibitem [\protect \citeauthoryear {%
Heckman%
\ \BBA {} Vytlacil%
}{%
Heckman%
\ \BBA {} Vytlacil%
}{%
{\protect \APACyear {2005}}%
}]{%
heckman2005structural}
\APACinsertmetastar {%
heckman2005structural}%
\begin{APACrefauthors}%
Heckman, J\BPBI J.%
\BCBT {}\ \BBA {} Vytlacil, E.%
\end{APACrefauthors}%
\unskip\
\newblock
\APACrefYearMonthDay{2005}{}{}.
\newblock
{\BBOQ}\APACrefatitle {Structural equations, treatment effects, and econometric
  policy evaluation 1} {Structural equations, treatment effects, and
  econometric policy evaluation 1}.{\BBCQ}
\newblock
\APACjournalVolNumPages{Econometrica}{73}{3}{669--738}.
\PrintBackRefs{\CurrentBib}

\bibitem [\protect \citeauthoryear {%
Ho%
, Jain%
\BCBL {}\ \BBA {} Abbeel%
}{%
Ho%
\ \protect \BOthers {.}}{%
{\protect \APACyear {2020}}%
}]{%
ho2020denoising}
\APACinsertmetastar {%
ho2020denoising}%
\begin{APACrefauthors}%
Ho, J.%
, Jain, A.%
\BCBL {}\ \BBA {} Abbeel, P.%
\end{APACrefauthors}%
\unskip\
\newblock
\APACrefYearMonthDay{2020}{}{}.
\newblock
{\BBOQ}\APACrefatitle {Denoising diffusion probabilistic models} {Denoising
  diffusion probabilistic models}.{\BBCQ}
\newblock
\APACjournalVolNumPages{Advances in neural information processing
  systems}{33}{}{6840--6851}.
\PrintBackRefs{\CurrentBib}

\bibitem [\protect \citeauthoryear {%
Hoderlein%
, Holzmann%
, Kasy%
\BCBL {}\ \BBA {} Meister%
}{%
Hoderlein%
\ \protect \BOthers {.}}{%
{\protect \APACyear {2017}}%
}]{%
hoderlein2017corrigendum}
\APACinsertmetastar {%
hoderlein2017corrigendum}%
\begin{APACrefauthors}%
Hoderlein, S.%
, Holzmann, H.%
, Kasy, M.%
\BCBL {}\ \BBA {} Meister, A.%
\end{APACrefauthors}%
\unskip\
\newblock
\APACrefYearMonthDay{2017}{}{}.
\newblock
{\BBOQ}\APACrefatitle {Corrigendum: Instrumental Variables with Unrestricted
  Heterogeneity and Continuous Treatment} {Corrigendum: Instrumental variables
  with unrestricted heterogeneity and continuous treatment}.{\BBCQ}
\newblock
\APACjournalVolNumPages{The Review of Economic Studies}{84}{2}{964--968}.
\PrintBackRefs{\CurrentBib}

\bibitem [\protect \citeauthoryear {%
Hoderlein%
\ \BBA {} Mammen%
}{%
Hoderlein%
\ \BBA {} Mammen%
}{%
{\protect \APACyear {2007}}%
}]{%
hoderlein2007identification}
\APACinsertmetastar {%
hoderlein2007identification}%
\begin{APACrefauthors}%
Hoderlein, S.%
\BCBT {}\ \BBA {} Mammen, E.%
\end{APACrefauthors}%
\unskip\
\newblock
\APACrefYearMonthDay{2007}{}{}.
\newblock
{\BBOQ}\APACrefatitle {Identification of marginal effects in nonseparable
  models without monotonicity} {Identification of marginal effects in
  nonseparable models without monotonicity}.{\BBCQ}
\newblock
\APACjournalVolNumPages{Econometrica}{75}{5}{1513--1518}.
\PrintBackRefs{\CurrentBib}

\bibitem [\protect \citeauthoryear {%
Hoderlein%
\ \BBA {} Mammen%
}{%
Hoderlein%
\ \BBA {} Mammen%
}{%
{\protect \APACyear {2009}}%
}]{%
hoderlein2009identification}
\APACinsertmetastar {%
hoderlein2009identification}%
\begin{APACrefauthors}%
Hoderlein, S.%
\BCBT {}\ \BBA {} Mammen, E.%
\end{APACrefauthors}%
\unskip\
\newblock
\APACrefYearMonthDay{2009}{}{}.
\newblock
{\BBOQ}\APACrefatitle {Identification and estimation of local average
  derivatives in non-separable models without monotonicity} {Identification and
  estimation of local average derivatives in non-separable models without
  monotonicity}.{\BBCQ}
\newblock
\APACjournalVolNumPages{The Econometrics Journal}{12}{1}{1--25}.
\PrintBackRefs{\CurrentBib}

\bibitem [\protect \citeauthoryear {%
Hoderlein%
, Su%
, White%
\BCBL {}\ \BBA {} Yang%
}{%
Hoderlein%
\ \protect \BOthers {.}}{%
{\protect \APACyear {2016}}%
}]{%
hoderlein2016testing}
\APACinsertmetastar {%
hoderlein2016testing}%
\begin{APACrefauthors}%
Hoderlein, S.%
, Su, L.%
, White, H.%
\BCBL {}\ \BBA {} Yang, T\BPBI T.%
\end{APACrefauthors}%
\unskip\
\newblock
\APACrefYearMonthDay{2016}{}{}.
\newblock
{\BBOQ}\APACrefatitle {Testing for monotonicity in unobservables under
  unconfoundedness} {Testing for monotonicity in unobservables under
  unconfoundedness}.{\BBCQ}
\newblock
\APACjournalVolNumPages{Journal of Econometrics}{193}{1}{183--202}.
\PrintBackRefs{\CurrentBib}

\bibitem [\protect \citeauthoryear {%
Hornik%
, Stinchcombe%
\BCBL {}\ \BBA {} White%
}{%
Hornik%
\ \protect \BOthers {.}}{%
{\protect \APACyear {1989}}%
}]{%
hornik1989multilayer}
\APACinsertmetastar {%
hornik1989multilayer}%
\begin{APACrefauthors}%
Hornik, K.%
, Stinchcombe, M.%
\BCBL {}\ \BBA {} White, H.%
\end{APACrefauthors}%
\unskip\
\newblock
\APACrefYearMonthDay{1989}{}{}.
\newblock
{\BBOQ}\APACrefatitle {Multilayer feedforward networks are universal
  approximators} {Multilayer feedforward networks are universal
  approximators}.{\BBCQ}
\newblock
\APACjournalVolNumPages{Neural Networks}{2}{5}{359--366}.
\PrintBackRefs{\CurrentBib}

\bibitem [\protect \citeauthoryear {%
Imbens%
}{%
Imbens%
}{%
{\protect \APACyear {2007}}%
}]{%
imbens2007nonadditive}
\APACinsertmetastar {%
imbens2007nonadditive}%
\begin{APACrefauthors}%
Imbens, G\BPBI W.%
\end{APACrefauthors}%
\unskip\
\newblock
\APACrefYearMonthDay{2007}{}{}.
\newblock
{\BBOQ}\APACrefatitle {Nonadditive models with endogenous regressors}
  {Nonadditive models with endogenous regressors}.{\BBCQ}
\newblock
\APACjournalVolNumPages{Econometric Society Monographs}{43}{}{17}.
\PrintBackRefs{\CurrentBib}

\bibitem [\protect \citeauthoryear {%
Imbens%
\ \BBA {} Angrist%
}{%
Imbens%
\ \BBA {} Angrist%
}{%
{\protect \APACyear {1994}}%
}]{%
angrist1994identification}
\APACinsertmetastar {%
angrist1994identification}%
\begin{APACrefauthors}%
Imbens, G\BPBI W.%
\BCBT {}\ \BBA {} Angrist, J\BPBI D.%
\end{APACrefauthors}%
\unskip\
\newblock
\APACrefYearMonthDay{1994}{}{}.
\newblock
{\BBOQ}\APACrefatitle {Identification and estimation of local average treatment
  effects} {Identification and estimation of local average treatment
  effects}.{\BBCQ}
\newblock
\APACjournalVolNumPages{Econometrica}{62}{2}{467-475}.
\PrintBackRefs{\CurrentBib}

\bibitem [\protect \citeauthoryear {%
Imbens%
\ \BBA {} Newey%
}{%
Imbens%
\ \BBA {} Newey%
}{%
{\protect \APACyear {2009}}%
}]{%
imbens2009identification}
\APACinsertmetastar {%
imbens2009identification}%
\begin{APACrefauthors}%
Imbens, G\BPBI W.%
\BCBT {}\ \BBA {} Newey, W\BPBI K.%
\end{APACrefauthors}%
\unskip\
\newblock
\APACrefYearMonthDay{2009}{}{}.
\newblock
{\BBOQ}\APACrefatitle {Identification and estimation of triangular simultaneous
  equations models without additivity} {Identification and estimation of
  triangular simultaneous equations models without additivity}.{\BBCQ}
\newblock
\APACjournalVolNumPages{Econometrica}{77}{5}{1481--1512}.
\PrintBackRefs{\CurrentBib}

\bibitem [\protect \citeauthoryear {%
Kahneman%
\ \BBA {} Tversky%
}{%
Kahneman%
\ \BBA {} Tversky%
}{%
{\protect \APACyear {1979}}%
}]{%
tversky1979analysis}
\APACinsertmetastar {%
tversky1979analysis}%
\begin{APACrefauthors}%
Kahneman, D.%
\BCBT {}\ \BBA {} Tversky, A.%
\end{APACrefauthors}%
\unskip\
\newblock
\APACrefYearMonthDay{1979}{}{}.
\newblock
{\BBOQ}\APACrefatitle {Prospect theory: An analysis of decisions under risk}
  {Prospect theory: An analysis of decisions under risk}.{\BBCQ}
\newblock
\APACjournalVolNumPages{Econometrica}{47}{2}{263--292}.
\PrintBackRefs{\CurrentBib}

\bibitem [\protect \citeauthoryear {%
Kim%
}{%
Kim%
}{%
{\protect \APACyear {2020}}%
}]{%
kim2020partial}
\APACinsertmetastar {%
kim2020partial}%
\begin{APACrefauthors}%
Kim, D.%
\end{APACrefauthors}%
\unskip\
\newblock
\APACrefYearMonthDay{2020}{}{}.
\newblock
{\BBOQ}\APACrefatitle {Partial identification in nonseparable count data
  instrumental variable models} {Partial identification in nonseparable count
  data instrumental variable models}.{\BBCQ}
\newblock
\APACjournalVolNumPages{The Econometrics Journal}{23}{2}{232--250}.
\PrintBackRefs{\CurrentBib}

\bibitem [\protect \citeauthoryear {%
Klein%
}{%
Klein%
}{%
{\protect \APACyear {2010}}%
}]{%
klein2010heterogeneous}
\APACinsertmetastar {%
klein2010heterogeneous}%
\begin{APACrefauthors}%
Klein, T\BPBI J.%
\end{APACrefauthors}%
\unskip\
\newblock
\APACrefYearMonthDay{2010}{}{}.
\newblock
{\BBOQ}\APACrefatitle {Heterogeneous treatment effects: Instrumental variables
  without monotonicity?} {Heterogeneous treatment effects: Instrumental
  variables without monotonicity?}{\BBCQ}
\newblock
\APACjournalVolNumPages{Journal of Econometrics}{155}{2}{99--116}.
\PrintBackRefs{\CurrentBib}

\bibitem [\protect \citeauthoryear {%
Mammen%
, Rothe%
\BCBL {}\ \BBA {} Schienle%
}{%
Mammen%
\ \protect \BOthers {.}}{%
{\protect \APACyear {2012}}%
}]{%
mammen2012nonparametric}
\APACinsertmetastar {%
mammen2012nonparametric}%
\begin{APACrefauthors}%
Mammen, E.%
, Rothe, C.%
\BCBL {}\ \BBA {} Schienle, M.%
\end{APACrefauthors}%
\unskip\
\newblock
\APACrefYearMonthDay{2012}{}{}.
\newblock
{\BBOQ}\APACrefatitle {Nonparametric regression with nonparametrically
  generated covariates} {Nonparametric regression with nonparametrically
  generated covariates}.{\BBCQ}
\newblock
\APACjournalVolNumPages{The Annals of Statistics}{40}{2}{1132--1170}.
\PrintBackRefs{\CurrentBib}

\bibitem [\protect \citeauthoryear {%
Mammen%
, Rothe%
\BCBL {}\ \BBA {} Schienle%
}{%
Mammen%
\ \protect \BOthers {.}}{%
{\protect \APACyear {2016}}%
}]{%
mammen2016semiparametric}
\APACinsertmetastar {%
mammen2016semiparametric}%
\begin{APACrefauthors}%
Mammen, E.%
, Rothe, C.%
\BCBL {}\ \BBA {} Schienle, M.%
\end{APACrefauthors}%
\unskip\
\newblock
\APACrefYearMonthDay{2016}{}{}.
\newblock
{\BBOQ}\APACrefatitle {Semiparametric estimation with generated covariates}
  {Semiparametric estimation with generated covariates}.{\BBCQ}
\newblock
\APACjournalVolNumPages{Econometric Theory}{32}{5}{1140--1177}.
\PrintBackRefs{\CurrentBib}

\bibitem [\protect \citeauthoryear {%
Manole%
, Balakrishnan%
\BCBL {}\ \BBA {} Wasserman%
}{%
Manole%
\ \protect \BOthers {.}}{%
{\protect \APACyear {2019}}%
}]{%
manole2019minimax}
\APACinsertmetastar {%
manole2019minimax}%
\begin{APACrefauthors}%
Manole, T.%
, Balakrishnan, S.%
\BCBL {}\ \BBA {} Wasserman, L.%
\end{APACrefauthors}%
\unskip\
\newblock
\APACrefYearMonthDay{2019}{}{}.
\newblock
{\BBOQ}\APACrefatitle {Minimax confidence intervals for the sliced Wasserstein
  distance} {Minimax confidence intervals for the sliced wasserstein
  distance}.{\BBCQ}
\newblock
\APACjournalVolNumPages{arXiv preprint arXiv:1909.07862}{}{}{}.
\PrintBackRefs{\CurrentBib}

\bibitem [\protect \citeauthoryear {%
Myerson%
}{%
Myerson%
}{%
{\protect \APACyear {1999}}%
}]{%
myerson1999nash}
\APACinsertmetastar {%
myerson1999nash}%
\begin{APACrefauthors}%
Myerson, R\BPBI B.%
\end{APACrefauthors}%
\unskip\
\newblock
\APACrefYearMonthDay{1999}{}{}.
\newblock
{\BBOQ}\APACrefatitle {Nash equilibrium and the history of economic theory}
  {Nash equilibrium and the history of economic theory}.{\BBCQ}
\newblock
\APACjournalVolNumPages{Journal of Economic Literature}{37}{3}{1067--1082}.
\PrintBackRefs{\CurrentBib}

\bibitem [\protect \citeauthoryear {%
Newey%
}{%
Newey%
}{%
{\protect \APACyear {1994}}%
}]{%
newey1994kernel}
\APACinsertmetastar {%
newey1994kernel}%
\begin{APACrefauthors}%
Newey, W\BPBI K.%
\end{APACrefauthors}%
\unskip\
\newblock
\APACrefYearMonthDay{1994}{}{}.
\newblock
{\BBOQ}\APACrefatitle {Kernel estimation of partial means and a general
  variance estimator} {Kernel estimation of partial means and a general
  variance estimator}.{\BBCQ}
\newblock
\APACjournalVolNumPages{Econometric Theory}{10}{2}{1--21}.
\PrintBackRefs{\CurrentBib}

\bibitem [\protect \citeauthoryear {%
Newey%
\ \BBA {} Powell%
}{%
Newey%
\ \BBA {} Powell%
}{%
{\protect \APACyear {2003}}%
}]{%
newey2003instrumental}
\APACinsertmetastar {%
newey2003instrumental}%
\begin{APACrefauthors}%
Newey, W\BPBI K.%
\BCBT {}\ \BBA {} Powell, J\BPBI L.%
\end{APACrefauthors}%
\unskip\
\newblock
\APACrefYearMonthDay{2003}{}{}.
\newblock
{\BBOQ}\APACrefatitle {Instrumental variable estimation of nonparametric
  models} {Instrumental variable estimation of nonparametric models}.{\BBCQ}
\newblock
\APACjournalVolNumPages{Econometrica}{71}{5}{1565--1578}.
\PrintBackRefs{\CurrentBib}

\bibitem [\protect \citeauthoryear {%
Pearl%
}{%
Pearl%
}{%
{\protect \APACyear {2019}}%
}]{%
pearl2019interpretation}
\APACinsertmetastar {%
pearl2019interpretation}%
\begin{APACrefauthors}%
Pearl, J.%
\end{APACrefauthors}%
\unskip\
\newblock
\APACrefYearMonthDay{2019}{}{}.
\newblock
{\BBOQ}\APACrefatitle {On the Interpretation of do (x)} {On the interpretation
  of do (x)}.{\BBCQ}
\newblock
\APACjournalVolNumPages{Journal of Causal Inference}{7}{1}{1--6}.
\PrintBackRefs{\CurrentBib}

\bibitem [\protect \citeauthoryear {%
Rosenzweig%
\ \BBA {} Wolpin%
}{%
Rosenzweig%
\ \BBA {} Wolpin%
}{%
{\protect \APACyear {2000}}%
}]{%
rosenzweig2000natural}
\APACinsertmetastar {%
rosenzweig2000natural}%
\begin{APACrefauthors}%
Rosenzweig, M\BPBI R.%
\BCBT {}\ \BBA {} Wolpin, K\BPBI I.%
\end{APACrefauthors}%
\unskip\
\newblock
\APACrefYearMonthDay{2000}{}{}.
\newblock
{\BBOQ}\APACrefatitle {Natural "{N}atural {E}xperiments" in {E}conomics}
  {Natural "{N}atural {E}xperiments" in {E}conomics}.{\BBCQ}
\newblock
\APACjournalVolNumPages{Journal of Economic Literature}{38}{4}{827--874}.
\PrintBackRefs{\CurrentBib}

\bibitem [\protect \citeauthoryear {%
Sato%
}{%
Sato%
}{%
{\protect \APACyear {1967}}%
}]{%
Sato1967}
\APACinsertmetastar {%
Sato1967}%
\begin{APACrefauthors}%
Sato, K.%
\end{APACrefauthors}%
\unskip\
\newblock
\APACrefYearMonthDay{1967}{}{}.
\newblock
{\BBOQ}\APACrefatitle {A Two-Level Constant-Elasticity-of-Substitution
  Production Function} {A two-level constant-elasticity-of-substitution
  production function}.{\BBCQ}
\newblock
\APACjournalVolNumPages{The Review of Economic Studies}{34}{2}{201--218}.
\PrintBackRefs{\CurrentBib}

\bibitem [\protect \citeauthoryear {%
{Schennach}%
}{%
{Schennach}%
}{%
{\protect \APACyear {2014}}%
}]{%
schennach2014entropic}
\APACinsertmetastar {%
schennach2014entropic}%
\begin{APACrefauthors}%
{Schennach}, S\BPBI M.%
\end{APACrefauthors}%
\unskip\
\newblock
\APACrefYearMonthDay{2014}{}{}.
\newblock
{\BBOQ}\APACrefatitle {Entropic latent variable integration via simulation}
  {Entropic latent variable integration via simulation}.{\BBCQ}
\newblock
\APACjournalVolNumPages{Econometrica}{82}{1}{345--385}.
\PrintBackRefs{\CurrentBib}

\bibitem [\protect \citeauthoryear {%
Singh%
\ \protect \BOthers {.}}{%
Singh%
\ \protect \BOthers {.}}{%
{\protect \APACyear {2018}}%
}]{%
singh2018nonparametric}
\APACinsertmetastar {%
singh2018nonparametric}%
\begin{APACrefauthors}%
Singh, S.%
, Uppal, A.%
, Li, B.%
, Li, C.%
, Zaheer, M.%
\BCBL {}\ \BBA {} P{\'o}czos, B.%
\end{APACrefauthors}%
\unskip\
\newblock
\APACrefYearMonthDay{2018}{}{}.
\newblock
{\BBOQ}\APACrefatitle {Nonparametric density estimation under adversarial
  losses} {Nonparametric density estimation under adversarial losses}.{\BBCQ}
\newblock
\BIn{} \APACrefbtitle {Advances in Neural Information Processing Systems}
  {Advances in neural information processing systems}\ (\BPGS\ 10225--10236).
\PrintBackRefs{\CurrentBib}

\bibitem [\protect \citeauthoryear {%
Sinn%
\ \BBA {} Rawat%
}{%
Sinn%
\ \BBA {} Rawat%
}{%
{\protect \APACyear {2018}}%
}]{%
sinn2018non}
\APACinsertmetastar {%
sinn2018non}%
\begin{APACrefauthors}%
Sinn, M.%
\BCBT {}\ \BBA {} Rawat, A.%
\end{APACrefauthors}%
\unskip\
\newblock
\APACrefYearMonthDay{2018}{}{}.
\newblock
{\BBOQ}\APACrefatitle {Non-parametric estimation of Jensen-Shannon Divergence
  in Generative Adversarial Network training} {Non-parametric estimation of
  jensen-shannon divergence in generative adversarial network training}.{\BBCQ}
\newblock
\BIn{} A.~Storkey\ \BBA {} F.~Perez-Cruz\ (\BEDS), \APACrefbtitle {Proceedings
  of the Twenty-First International Conference on Artificial Intelligence and
  Statistics} {Proceedings of the twenty-first international conference on
  artificial intelligence and statistics}\ (\BVOL~84, \BPGS\ 642--651).
\newblock
\APACaddressPublisher{}{PMLR}.
\PrintBackRefs{\CurrentBib}

\bibitem [\protect \citeauthoryear {%
Wei%
, Gong%
, Liu%
, Lu%
\BCBL {}\ \BBA {} Wang%
}{%
Wei%
\ \protect \BOthers {.}}{%
{\protect \APACyear {2018}}%
}]{%
wei2018improving}
\APACinsertmetastar {%
wei2018improving}%
\begin{APACrefauthors}%
Wei, X.%
, Gong, B.%
, Liu, Z.%
, Lu, W.%
\BCBL {}\ \BBA {} Wang, L.%
\end{APACrefauthors}%
\unskip\
\newblock
\APACrefYearMonthDay{2018}{}{}.
\newblock
{\BBOQ}\APACrefatitle {Improving the improved training of wasserstein gans: A
  consistency term and its dual effect} {Improving the improved training of
  wasserstein gans: A consistency term and its dual effect}.{\BBCQ}
\newblock
\APACjournalVolNumPages{arXiv preprint arXiv:1803.01541}{}{}{}.
\PrintBackRefs{\CurrentBib}

\end{thebibliography}
		
		\bibliographystyle{apacite}
		\pagebreak		
		\begin{appendices}
			\section*{Appendices}
			\begin{subappendices}%\label{Appendix:A}
				
				\section{An integral equation for the intervention}\label{Sec:integral}
				%\section{An alternative representation for the intervention via the solution of an integral equation}
				
				%\textcolor{red}{Roadmap for this section}
				In this section we derive an integral equation relating the unknown intervention $H(y|x)$ to known quantities. 
%
%%%%%%%%%%%%%%%%%%%%%%%%%%%%%%%%%%%%%%%%%%%%%%%%%%%%%%%%%%%%%%%%%%%%
%%%%%%%%%%%%%%%%%%%%%%%%%%%%%%%%%%%%%%%%%%%%%%%%%%%%%%%%%%%%%%%%%%%%
%\subsection{The Fredholm equation}\label{sec:inverse}
% you need to connect it to equation in lemma 3.2, you don't observe information, only reliaztions.
Recall that Eq. (\ref{eq:def_H}) for $H(y|x)$ is fully characterized by the quantities $f_{\upeta}(\cdot)$ and $F_{Y|X=x,\upeta=\eta}(y)$. 
Unfortunately, both of these quantities are
unknown, since $\upeta$ is unobserved. 

For future use, from this point on, instead of working with the disturbance $\eta$, we make a change of variables and work with
its normalized random variable 
\begin{equation}
	\upomega = F_{\upeta}(\upeta), 
	\label{eq:omega_eta}
\end{equation}
which is distributed as $U[0,1]$. 

Towards our goal of proving that $H(y|x)$ is identifiable, we next formulate the relationship between the unknown quantity $F_{Y|X,\upomega}$ and the (known) distribution $F_{Y|X,Z}$ through an integral equation.  
This relation is stated in the following lemma is proven in Appendix \ref{Appendix:Proof}.

\begin{lemma}\label{lemma:integral}	
	Let $(X,Y,Z,\upeta,\bm\epsilon)$ be distributed according to the triangular model in \eqref{X:Structural} and \eqref{Y:Structural}. Under assumptions \ref{asu:common} and \ref{asu:continuity}
	the following integral equation holds
	$\forall (x,y,z)\in\text{Supp}(X,Y,Z)$, 
	\begin{align}
		& F_{Y|X=x, Z=z}(y) =  \text{Pr}(Y\le y|X=x,Z=z) = \int_{0}^1 F_{Y|X=x, \upomega=\omega}(y)f_{\upomega|X=x,Z=z}(\omega)d\omega.  
		\label{eq:lemma:cumulative}
	\end{align}
\end{lemma}	
The entire variation in the estimable quantity $F_{Y|X=x,Z=z}(y)$ in the LHS of the system in \eqref{eq:lemma:cumulative} is propagated by $Z$. Such variation provides an indirect information regarding $\upeta$ through the integral equations  \eqref{eq:lemma:cumulative}. For sake of brevity, we rewrite Eq. \eqref{eq:lemma:cumulative} as follows, 
\begin{align}
	& G_{x,y}(z)=\int_{0}^1 \upgamma_{x,y}(\omega) k_{x}(z,\omega)d\omega \quad \forall (x,y)\in\text{Supp}(X,Y).
	\label{eq:integral:normalized}
\end{align}	
where $G_{x,y}(z):=F_{Y|X=x,Z=z}(y)$ is the observed quantity, $k_x(z,\omega) := f_{\upomega|X=x,Z=z}(\omega)$ is an unknown kernel function and 
\begin{equation}
	\upgamma_{x,y}(\omega):=F_{Y|X=x, \upomega=\omega}(y)
	\label{eq:gamma_def}
\end{equation}
is another unknown function. In terms of these quantities, Eq. (\ref{eq:def_H}) can be written as
\begin{align}
	%F_{Y|\text{Do}(X=x)}(y):=
	H(y|x) =  \int_{0}^1 \upgamma_{x,y}(\omega)d\omega,\label{eq:H:integral}
\end{align}			
Note that Eq. (\ref{eq:integral:normalized})
is an inhomogeneous Fredholm integral equation of the first kind. In this equation, the left hand side is known, whereas on the right hand side both the kernel as well as the function $\upgamma_{x,y}$ are unknown.
In addition, the system of integral equations may be singular due to an unbounded kernel in that for any $(x,z)$ the kernel may be a weighted sum of Dirac delta functions $\updelta(\cdot)$, aggravating the identifiability challenge. Our identifiability strategy for the interventional distribution bypasses these issues. This is done by characterizing an equivalence class of triangular models sharing the same interventional distribution, despite having different kernels, and is treated in the section to follow.

\subsection{Fourier Expansion}
Our main building block in achieving identifiability relies on 
representing the intervention in terms of Fourier series. For doing so, in the following lemma we rely on Carleson's theorem \citep{carleson1966convergence} for establishing the point-wise (Lebesgue) almost everywhere convergence of Fourier series of $L_2$ functions (see proof in appendix \ref{Appendix:Proof}).
\begin{lemma}\label{lemma:H:Fourier}
	Suppose that $f_{Y,X|Z=z}$ and $f_{Y,X|Z=z}^{\text{CF}}$ satisfy squares integrability. Namely,
	\[
	\int\int f_{Y,X|Z=z}^2(y,x)dydx <  \infty \quad \forall z \quad \text{and} \quad \int\int\left[f_{Y,X|Z=z}^{\text{CF}}(y,x)\right]^2 dydx <  \infty \quad\forall z.
	\]
	Thus, $H(y|x)$ can be characterized $\forall (x,y)\in\text{Supp}(X,Y)$ independently of $Z$ as follows,
	\begin{align}
		H(y|x)=\left\{f_{X|Z=z}(x)\right\}^{-1}\int _{y^{\prime}}^y\left\{
		\int_{\xi_1}\int_{\xi_2}\mathbb{E}_{\substack{Y\sim H(\cdot|X)\\ X\sim F_{X|Z=z}}}\left[\Psi_{\xi_1}(-Y)\Psi_{\xi_2}(-X)\right]\Psi_{\xi_2}(y^{\prime})\Psi_{\xi_3}(x) d\xi_2 d\xi_3 \right\}dy^{\prime},\label{eq:box}
	\end{align}
	and it holds that $\forall z\in\text{Supp}(Z|X=x)$,
	\begin{align}
		\mathbb{E}_{\substack{Y\sim H(\cdot|X)\\ X\sim F_{X|Z=z}}}\left[\Psi_{\xi_2}\left(-Y\right)\Psi_{\xi_3}\left(-X\right)\right]=\mathbb{E}_{\upomega,\upomega^{\prime},\bm\nu}\left[\Psi_{\xi_3}\left(-g^*\left(h^*(z,\upomega),\upomega^{\prime},\bm\upnu\right)\right)\Psi_{\xi_3}\left(-h^*(z,\upomega)\right)\right].\label{eq:moments:gen}
	\end{align}	 
	
\end{lemma}

				%\appendix
				\section{Algorithms}\label{Appendix:Algorithms}
				\renewcommand{\theequation}{\thesection.\arabic{equation}}
				\setcounter{equation}{0}
				We present our algorithms to estimate the counterfactual distribution: Initial parameter settings and Contaminated Generative Adversarial Network (CONGAN). 
				
				Algorithm \ref{algorithm:ContIDE} is employed to provide initial estimates for $\bm{\upbeta}$ and $\bm{\uptheta}$ in order to estimate $\widetilde{h}$ and $\widetilde{g}$, respectively. This is achieved by minimizing the negative entropy:
				\begin{align} 
					\mathcal{L}(\bm{\upbeta},\bm{\uptheta},\bm{\upsigma})=-N^{-1}\sum_{i=1}^N\log \frac{1}{\upsigma_y}\frac{1}{\upsigma_x}J^{-1}\sum_{j=1}^J K\left(\frac{\widetilde{g}\left(\widetilde{h}\left(Z_i,\omega_j\right),\omega_j,\bm\nu_j\right)-Y_i^{\text{O}}}{\upsigma_y}\right)K\left(\frac{\widetilde{h}\left(Z_i,\omega_j\right)-X_i^{\text{O}}}{\upsigma_x}\right),\nonumber
				\end{align} the empirical variant of the following negative entropy: $\mathcal{L}^{*}(\bm{\upbeta},\bm{\uptheta})=-\int \log f_{Y,X|Z}^{\text{O}}(y,x|z; \widetilde{h},\widetilde{g})f_Z(z)dy dx dz$,
				%\begin{align}
				%\mathcal{L}^{*}(\bm{\upbeta},\bm{\uptheta})=-\int \log f_{Y,X|Z}^{\text{O}}(y,x|z; \widetilde{h},\widetilde{g})f_Z(z)dy dx dz
				%\end{align}
				with $f_{Y,X|Z}^{\text{O}}(y,x|z; \widetilde{h},\widetilde{g})=\frac{\partial^2}{\partial y \partial x}F_{Y,X|Z}^{\text{O}}(y,x|z; \widetilde{h},\widetilde{g})$ in which the right hand side is the cross-partial derivative of the conditional distribution characterized in \eqref{equ:FYX:givenZ} with respect to $x$ and $y$.
				
				Algorithm \ref{algorithm:ContGAN} of the CONGAN model in \eqref{CONT:GAN:LOSS} is employed to yield the final estimates for $\bm{\upbeta}$ and $\bm{\uptheta}$ as well as the nuisance parameter $\bm{\upsigma}=(\upsigma_x,\upsigma_y,\upsigma_z)$. This is achieved by minimizing:
				\begin{align}
					\mathcal{K}_n(\bm{\uptheta},\bm{\upbeta},\bm{\upsigma}) = \frac{1}{N}\sum_{i=1}^N\log\frac{\widehat{p}_{\bm{\upsigma}}(X_i,Y_i,Z_i)}{\widehat{p}_{\bm{\upsigma}}(X_i,Y_i,Z_i)+\widehat{q}_{\bm{\upsigma}}(X_i,Y_i,Z_i)} + \frac{1}{N}\sum_{i=1}^N\log\frac{\widehat{q}_{\bm{\upsigma}}(\widehat{X}_i,\widehat{Y}_i,Z_i)}{\widehat{p}_{\bm{\upsigma}}(\widehat{X}_i,\widehat{Y}_i,Z_i)+\widehat{q}_{\bm{\upsigma}}(\widehat{X}_i,\widehat{Y}_i,Z_i)}.
				\end{align}

				\begin{center}
					\scalebox{1}{
						\begin{minipage}{1\linewidth}
							\begin{algorithm}[H]
								%	\SetAlFnt{\footnotesize}
								\setstretch{0.5}
								\SetAlgoLined
								\SetKwInOut{KwIn}{Input}
								\SetKwInOut{KwOut}{Output}	
								\KwIn{A sequence $\left\{(X_i,Y_i,Z_i)\right\}_{i=1}^N$ of independent realizations from $F_{X,Y,Z}$.}
								\KwOut{A pair of sequential generator functions $\left\{\widetilde{h}(\cdot,\cdot;\bm{\upbeta}),\widetilde{g}(\cdot,\cdot,\cdot;\bm{\uptheta})\right\}$.}
								%       \KwResult{}
								initialization of $\bm{\upbeta}$ and $\bm{\uptheta}, \upsigma_x, \upsigma_y$\;
								$\text{tol}  = 10^{-16}$\;      
								$\text{gap}  = 0$\;     
								$\text{iter}  = 0$\;    
								$\Lambda  = 0.01$\;
								$\text{Max.Iterations} = 10,000 $\;
								\While{$\text{gap} \ge \text{tol}$ \textbf{ and } $\text{iter}\le\text{Max.Iterations} $}{
									\For {$i=1:N$}{
										\For {$j=1:J$}{
											$\omega[i,j] = U_1, \hspace{1em} U_1\sim U[0,1]$\;       
											$\bm\nu[i,j] = U_2, \hspace{1em} U_2\sim U[0,1]$\;                                                       
											$\widehat{X}_{i,j}=\widetilde{h}(Z_i,\omega[i,j];\bm{\upbeta}^{\text{iter}})$\;
											$\widehat{Y}_{i,j}=\widetilde{g}(\widehat{X}_{i,j},\omega[i,j],\bm\nu[i,j];\bm{\uptheta}^{\text{iter}})$\;
										}              
										$\text{density}_i = \frac{1}{J}\sum_{j=1}^J \frac{1}{\upsigma_x} K\left(\frac{\widehat{X}_{i,j}-X_i}{\upsigma_x}\right)\frac{1}{\upsigma_y} K\left(\frac{\widehat{Y}_{i,j}-Y_i}{\upsigma_y}\right)$\;         
										
										$d_{i,j} =\max\left(\mathds{1}\left\{\widehat{X}_{i,j}\ne X_i\right\},\mathds{1}\left\{\widehat{Y}_{i,j}\ne Y_i\right\}\right)$\;
										$\widehat{p}_{\bm{\upsigma},i}(X_i,Y_i) = \left\{\sum_{j=1}^J d_{i,j}\right\}^{-1}\sum_{j=1}^J d_{i,j}\frac{1}{\upsigma_x} K\left(\frac{\widehat{X}_{i,j}-X_i}{\upsigma_x}\right)\frac{1}{\upsigma_y} K\left(\frac{\widehat{Y}_{i,j}-Y_i}{\upsigma_y}\right)$\;                        
									}
									
									$\mathcal{L}^{\text{iter}}=-\frac{1}{N}\sum_{i=1}^N\log\left(\max(\text{density}_i, 10^{-14})\right)$\;
									\vspace{1em}
									\textcolor{blue}{\textit{Employing leave-one out entropy estimator \citep{hall1993estimation}}}\;
									$\mathcal{L}_{\bm{\upsigma}}^{\text{iter}} = -\frac{1}{N}\sum_{i=1}^N \log\left(\max(\widehat{p}_{\bm{\upsigma},i}(X_i,Y_i), 10^{-14})\right)
									$\;
									\vspace{1em}
									%               $\mathcal{L}^{\text{iter}}=-\frac{1}{N}\sum_{i=1}^N\log\left(\max(\text{density}_i, 10^{-14})\right)$\;
									$\left[\begin{matrix}
										\bm{\upbeta}^{\text{iter}+1}\\
										\bm{\uptheta}^{\text{iter}+1}\\
										\bm{\upsigma}^{\text{iter}+1}           
									\end{matrix}\right] = \left[\begin{matrix}
										\bm{\upbeta}^{\text{iter}}\\
										\bm{\uptheta}^{\text{iter}}\\
										\bm{\upsigma}^{\text{iter}}
									\end{matrix}\right] - \Lambda\left[\begin{matrix}
										\nabla_{\bm{\upbeta}}\mathcal{L}^{\text{iter}}\\
										\nabla_{\bm{\uptheta}}\mathcal{L}^{\text{iter}}\\
										\nabla_{\bm{\upsigma}}\mathcal{L}_{\bm{\upsigma}}^{\text{iter}}         
									\end{matrix}\right]$\;
									\vspace{1em}            
									$(\upsigma_x, \upsigma_y) = \bm{\upsigma}^{\text{iter}+1}$\;
									\vspace{1em}            
									\If{$\text{iter} > 0$}{
										$\text{gap} =  \mathcal{L}^{\text{iter}}-\mathcal{L}^{\text{iter-1}}$ \;
									}
									$\text{iter} = \text{iter}+1$
								}
								\caption{Initial parameter settings}\label{algorithm:ContIDE}
							\end{algorithm}
						\end{minipage}%
					}
				\end{center}
				
				%\subsection{Algorithms}
				
				%The bandwidth vector $\bm{\upsigma}:=(\upsigma_x,\upsigma_y,\upsigma_z)$ is obtained a parameter of the model by using the following leave-one out entropy estimator \citep{hall1993estimation}: \begin{align}
					%& \bm{\upsigma} = \underset{\bm{\upsigma}}{\arg\min} -\frac{1}{N}\sum_{i=1}^N \log(\widehat{f}_{\bm{\upsigma},i}(X_i,Y_i,Z_i))\\
					%& \widehat{f}_{\bm{\upsigma},i}(X_i,Y_i,Z_i) = \frac{1}{N-1}\sum_{j\ne i} \frac{1}{\upsigma_x}K\left(\frac{X_j-X_i}{\upsigma_x}\right)  \frac{1}{\upsigma_y}K\left(\frac{Y_j-Y_i}{\upsigma_y}\right) \frac{1}{\upsigma_z}K\left(\frac{Z_j-Z_i}{\upsigma_z}\right)
					%\end{align}

					\begin{center}
						\scalebox{1}{
							\begin{minipage}{1\linewidth}
								\begin{algorithm}[H]
									\setstretch{0.5}
									\SetAlgoLined
									\SetKwInOut{KwIn}{Input}
									\SetKwInOut{KwOut}{Output}
									%		\algsetup{linenosize=\tiny}	
									\KwIn{A sequence $\left\{(X_i,Y_i,Z_i)\right\}_{i=1}^N$ of independent realizations from $F_{X,Y,Z}$.}
									\KwOut{A pair of sequential generator functions $\left\{\widetilde{h}(\cdot,\cdot;\bm{\upbeta}),\widetilde{g}(\cdot,\cdot,\cdot;\bm{\uptheta})\right\}$.}
									%       \KwResult{}
									initialization of $\bm{\upbeta}$ and $\bm{\uptheta}, \upsigma_x, \upsigma_y, \upsigma_z$\;
									$\text{tol}  = 10^{-16}$\;      
									$\text{gap}  = 0$\;     
									$\text{iter}  = 0$\;    
									$\Lambda  = 0.01$\;
									$\text{Max.Iterations} = 10,000 $\;
									\While{$\text{gap} \ge \text{tol}$ \textbf{ and } $\text{iter}\le\text{Max.Iterations} $}{
										\For {$i=1:N$}{
											$\omega[i] = U_1, \hspace{1em} U_1\sim U[0,1]$\;       
											$\bm\nu[i] = U_2, \hspace{1em} U_2\sim U[0,1]$\;                                       
											$\widehat{X}_{i}=\widetilde{h}(Z_i,\omega[i];\bm{\upbeta}^{\text{iter}})$\;
											$\widehat{Y}_{i}=\widetilde{g}(\widehat{X}_{i},\omega[i],\bm\nu[i];\bm{\uptheta}^{\text{iter}})$\;
										}
										\For {$i=1:N$}{
											$\widehat{p}_{\bm{\upsigma}}(X_i,Y_i,Z_i) = \max\left(\frac{1}{N}\sum_{j=1}^N \frac{1}{\upsigma_x} K\left(\frac{X_{j}-X_i}{\upsigma_x}\right)\frac{1}{\upsigma_y} K\left(\frac{Y_{j}-Y_i}{\upsigma_y}\right)\frac{1}{\upsigma_z}K\left(\frac{Z_j-Z_i}{\upsigma_z}\right),10^{-14}\right)$\;         
											$\widehat{p}_{\bm{\upsigma}}(\widehat{X}_i,\widehat{Y}_i,\widehat{Z}_i) = \max\left(\frac{1}{N}\sum_{j=1}^N \frac{1}{\upsigma_x} K\left(\frac{X_{j}-\widehat{X}_i}{\upsigma_x}\right)\frac{1}{\upsigma_y} K\left(\frac{Y_{j}-\widehat{Y}_i}{\upsigma_y}\right)\frac{1}{\upsigma_z}K\left(\frac{Z_j-Z_i}{\upsigma_z}\right),10^{-14}\right)$\;                             
											$\widehat{q}_{\bm{\upsigma}}(X_i,Y_i,Z_i) = \max\left(\frac{1}{N}\sum_{j=1}^N \frac{1}{\upsigma_x} K\left(\frac{\widehat{X}_{j}-X_i}{\upsigma_x}\right)\frac{1}{\upsigma_y} K\left(\frac{\widehat{Y}_{j}-Y_i}{\upsigma_y}\right)\frac{1}{\upsigma_z}K\left(\frac{Z_j-Z_i}{\upsigma_z}\right),10^{-14}\right)$\;                     
											$\widehat{q}_{\bm{\upsigma}}(\widehat{X}_i,\widehat{Y}_i,\widehat{Z}_i) = \max\left(\frac{1}{N}\sum_{j=1}^N \frac{1}{\upsigma_x} K\left(\frac{\widehat{X}_{j}-\widehat{X}_i}{\upsigma_x}\right)\frac{1}{\upsigma_y} K\left(\frac{\widehat{Y}_{j}-\widehat{Y}_i}{\upsigma_y}\right)\frac{1}{\upsigma_z}K\left(\frac{Z_j-Z_i}{\upsigma_z}\right),10^{-14}\right)$\;                     
											$\widehat{p}_{i,\bm{\upsigma}}(X_i,Y_i,Z_i) = \max\left(\frac{1}{N-1}\sum_{j\ne i} \frac{1}{\upsigma_x} K\left(\frac{X_{j}-X_i}{\upsigma_x}\right)\frac{1}{\upsigma_y} K\left(\frac{Y_{j}-Y_i}{\upsigma_y}\right)\frac{1}{\upsigma_z}K\left(\frac{Z_j-Z_i}{\upsigma_z}\right),10^{-14}\right)$\;                     
											$\widehat{q}_{i,\bm{\upsigma}}(\widehat{X}_i,\widehat{Y}_i,\widehat{Z}_i) = \max\left(\frac{1}{N-1}\sum_{j\ne i} \frac{1}{\upsigma_x} K\left(\frac{\widehat{X}_{j}-\widehat{X}_i}{\upsigma_x}\right)\frac{1}{\upsigma_y} K\left(\frac{\widehat{Y}_{j}-\widehat{Y}_i}{\upsigma_y}\right)\frac{1}{\upsigma_z}K\left(\frac{Z_j-Z_i}{\upsigma_z}\right),10^{-14}\right)$\;                             
										}

										$\mathcal{L}^{\text{iter}} = \frac{1}{N}\sum_{i=1}^N\log\frac{\widehat{p}_{\bm{\upsigma}}(X_i,Y_i,Z_i)}{\widehat{p}_{\bm{\upsigma}}(X_i,Y_i,Z_i)+\widehat{q}_{\bm{\upsigma}}(X_i,Y_i,Z_i)} + \frac{1}{N}\sum_{i=1}^N\log\frac{\widehat{q}_{\bm{\upsigma}}(\widehat{X}_i,\widehat{Y}_i,Z_i)}{\widehat{p}_{\bm{\upsigma}}(\widehat{X}_i,\widehat{Y}_i,Z_i)+\widehat{q}_{\bm{\upsigma}}(\widehat{X}_i,\widehat{Y}_i,Z_i)}$\;          
										\vspace{1em}
										\textcolor{blue}{\textit{Employing leave-one out entropy estimator \citep{hall1993estimation}}}\;
										$\mathcal{L}_{\bm{\upsigma}}^{\text{iter}} = -\frac{1}{N}\sum_{i=1}^N \log(\widehat{p}_{\bm{\upsigma},i}(X_i,Y_i,Z_i))-\frac{1}{N}\sum_{i=1}^N \log(\widehat{q}_{i,\bm{\upsigma}}(\widehat{X}_i,\widehat{Y}_i,\widehat{Z}_i)
										)$\;
										\vspace{1em}
										%               $\mathcal{L}^{\text{iter}}=-\frac{1}{N}\sum_{i=1}^N\log\left(\max(\text{density}_i, 10^{-14})\right)$\;
										$\left[\begin{matrix}
											\bm{\upbeta}^{\text{iter}+1}\\
											\bm{\uptheta}^{\text{iter}+1}\\
											\bm{\upsigma}^{\text{iter}+1}           
										\end{matrix}\right] = \left[\begin{matrix}
											\bm{\upbeta}^{\text{iter}}\\
											\bm{\uptheta}^{\text{iter}}\\
											\bm{\upsigma}^{\text{iter}}
										\end{matrix}\right] - \Lambda\left[\begin{matrix}
											\nabla_{\bm{\upbeta}}\mathcal{L}^{\text{iter}}\\
											\nabla_{\bm{\uptheta}}\mathcal{L}^{\text{iter}}\\
											\nabla_{\bm{\upsigma}}\mathcal{L}_{\bm{\upsigma}}^{\text{iter}}         
										\end{matrix}\right]$\;
										\vspace{1em}            
										$(\upsigma_x, \upsigma_y, \upsigma_z) = \bm{\upsigma}^{\text{iter}+1}$\;
										\vspace{1em}            
										\If{$\text{iter} > 0$}{
											$\text{gap} =  \mathcal{L}^{\text{iter}}-\mathcal{L}^{\text{iter-1}}$ \;
										}
										$\text{iter} = \text{iter}+1$
									}
									\caption{Contaminated Generative Adversarial Network (CONGAN)}\label{algorithm:ContGAN}
								\end{algorithm}
							\end{minipage}%
						}
					\end{center}
					
					%\leavevmode\thispagestyle{empty}\clearpage %\newpage
					%\newpage\clearpage 
					%\newpage
					% \let\cleardoublepage\clearpage
					%\clearpage 
					%\cleardoublepage
					%\pagebreak
					\vspace{5em}
					\section{Tables}\label{Appendix:Tables}
					%\clearpage 

					\begin{table}[H]\centering\caption{\label{tab:CES:Triang:Average}Do-intervention (CES supply function \citep{Sato1967})\\ \small Entries represent average (standard error) of estimates\\ obtained from $100$ samples of 10,000 observations each\\ 		
							$X=h_{3}\left(\zeta(Z), \zeta(\eta)\right)^{1}$; \hspace{1em} $Y=g_5(X,\bm\epsilon)^{2}$; \hspace{1em} $\upzeta(x)=1.5+1.5\tanh(0.15 x)$\\\vspace{1mm}					
							$[\upalpha_1,\upalpha_2,\upalpha_3,\uprho]=[6,0.5,0.5,0.5]$\\ $[\upbeta_1,\upbeta_2,\upbeta_3]=[2,1,-0.25]$,\hspace{1em} $[\upgamma_1,\upgamma_2,\upgamma_3]=[1,1,1]$}\vspace{1mm}
						\scaleobj{0.75}{\begin{threeparttable}	
								\begin{tabular}{|c|c|c|c|c|c|c|c|}
									\hline
									\multirow{2}{*}{\begin{tabular}[c]{@{}c@{}}$\overset{(a)}{\text{Quantile}}$\\  of  X\end{tabular}} & \multirow{2}{*}{\begin{tabular}[c]{@{}c@{}}\\ $\overset{(b)}{X}$\end{tabular}} & \multicolumn{2}{c|}{$\mathbb{E}\left[Y|X=x\right]$} & \multicolumn{4}{c|}{$\mathbb{E}\left[Y|\text{do}(X=x)\right]$}    \\ \cline{3-8} 
									&                    & $\overset{(c)}{\text{real}}$           & $\overset{(d)}{\text{synt}}$             & $\overset{(e)}{\text{real}}$  & $\overset{(f)}{\text{synt}}$    & $\overset{(g)}{\text{IN}}$ & $\overset{(h)}{\text{PM}}$ \\ \hline\multirow{2}{*}{5}                                                                                  & 2.79                                                                               & 10.45                         & 10.57$^{\dag}$                        & \textbf{9.55}                 & \textbf{10.06*}               & 6.98                                                                                                                                                                      & 6.96                                                                                                                                    \\ 
									\cline{2-8}
									& (0.066)                                                                            & (0.042)                       & (0.092)                       & (0.080)                       & (0.143)                       & (0.338)                                                                                                                                                                   & (0.276)                                                                                                                                 \\ 
									\hline
									\multirow{2}{*}{10}                                                                                 & 3.29                                                                               & 10.67                         & 10.80$^{\dag}$                          & \textbf{9.81}                 & \textbf{10.30*}               & 7.60                                                                                                                                                                      & 7.59                                                                                                                                    \\ 
									\cline{2-8}
									& (0.060)                                                                            & (0.046)                       & (0.081)                       & (0.086)                       & (0.128)                       & (0.288)                                                                                                                                                                   & (0.234)                                                                                                                                 \\ 
									\hline
									\multirow{2}{*}{15}                                                                                 & 3.56                                                                               & 10.81                         & 10.93$^{\dag}$                          & \textbf{9.95}                 & \textbf{10.43*}               & 7.98                                                                                                                                                                      & 7.97                                                                                                                                    \\ 
									\cline{2-8}
									& (0.062)                                                                            & (0.048)                       & (0.074)                       & (0.086)                       & (0.113)                       & (0.245)                                                                                                                                                                   & (0.212)                                                                                                                                 \\ 
									\hline
									\multirow{2}{*}{20}                                                                                 & 3.77                                                                               & 10.91                         & 11.02$^{\dag}$                          & \textbf{10.07}                & \textbf{10.54*}               & 8.27                                                                                                                                                                      & 8.25                                                                                                                                    \\ 
									\cline{2-8}
									& (0.056)                                                                            & (0.045)                       & (0.073)                       & (0.084)                       & (0.105)                       & (0.231)                                                                                                                                                                   & (0.198)                                                                                                                                 \\ 
									\hline
									\multirow{2}{*}{25}                                                                                 & 3.94                                                                               & 11.00                         & 11.11$^{\dag}$                          & \textbf{10.16}                & \textbf{10.62*}               & 8.52                                                                                                                                                                      & 8.51                                                                                                                                    \\ 
									\cline{2-8}
									& (0.053)                                                                            & (0.046)                       & (0.075)                       & (0.086)                       & (0.101)                       & (0.213)                                                                                                                                                                   & (0.188)                                                                                                                                 \\ 
									\hline
									\multirow{2}{*}{30}                                                                                 & 4.11                                                                               & 11.10                         & 11.20$^{\dag}$                          & \textbf{10.26}                & \textbf{10.71*}               & 8.79                                                                                                                                                                      & 8.77*                                                                                                                                   \\ 
									\cline{2-8}
									& (0.064)                                                                            & (0.051)                       & (0.080)                       & (0.096)                       & (0.100)                       & (0.201)                                                                                                                                                                   & (0.189)                                                                                                                                 \\ 
									\hline
									\multirow{2}{*}{35}                                                                                 & 4.37                                                                               & 11.24                         & 11.33$^{\dag}$                          & \textbf{10.42}                & \textbf{10.86*}               & 9.19                                                                                                                                                                      & 9.17*                                                                                                                                   \\ 
									\cline{2-8}
									& (0.097)                                                                            & (0.076)                       & (0.100)                       & (0.115)                       & (0.109)                       & (0.194)                                                                                                                                                                   & (0.228)                                                                                                                                 \\ 
									\hline
									\multirow{2}{*}{40}                                                                                 & 4.87                                                                               & 11.55                         & 11.62$^{\dag}$                          & \textbf{10.74}                & \textbf{11.16*}               & 10.00                                                                                                                                                                     & 9.98*                                                                                                                                   \\ 
									\cline{2-8}
									& (0.177)                                                                            & (0.125)                       & (0.160)                       & (0.169)                       & (0.145)                       & (0.273)                                                                                                                                                                   & (0.329)                                                                                                                                 \\ 
									\hline
									\multirow{2}{*}{45}                                                                                 & 5.24                                                                               & 11.79                         & 11.85$^{\dag}$                          & \textbf{10.99}                & \textbf{11.39*}               & 10.60                                                                                                                                                                     & 10.58*                                                                                                                                  \\ 
									\cline{2-8}
									& (0.144)                                                                            & (0.111)                       & (0.179)                       & (0.138)                       & (0.151)                       & (0.270)                                                                                                                                                                   & (0.260)                                                                                                                                 \\ 
									\hline
									\multirow{2}{*}{50}                                                                                 & 5.53                                                                               & 12.00                         & 12.04$^{\dag}$                          & \textbf{11.20}                & \textbf{11.58*}               & 11.12                                                                                                                                                                     & 11.08*                                                                                                                                  \\ 
									\cline{2-8}
									& (0.189)                                                                            & (0.150)                       & (0.221)                       & (0.190)                       & (0.184)                       & (0.350)                                                                                                                                                                   & (0.320)                                                                                                                                 \\ 
									\hline
									\multirow{2}{*}{55}                                                                                 & 6.04                                                                               & 12.39                         & 12.39$^{\dag}$                          & \textbf{11.58}                & \textbf{11.93*}               & 12.01                                                                                                                                                                     & 11.96                                                                                                                                   \\ 
									\cline{2-8}
									& (0.137)                                                                            & (0.125)                       & (0.226)                       & (0.174)                       & (0.194)                       & (0.271)                                                                                                                                                                   & (0.253)                                                                                                                                 \\ 
									\hline
									\multirow{2}{*}{60}                                                                                 & 6.44                                                                               & 12.71                         & 12.68$^{\dag}$                          & \textbf{11.89}                & \textbf{12.22*}               & 12.72                                                                                                                                                                     & 12.68                                                                                                                                   \\ 
									\cline{2-8}
									& (0.125)                                                                            & (0.127)                       & (0.257)                       & (0.166)                       & (0.218)                       & (0.247)                                                                                                                                                                   & (0.257)                                                                                                                                 \\ 
									\hline
									\multirow{2}{*}{65}                                                                                 & 6.86                                                                               & 13.07                         & 13.00$^{\dag}$                          & \textbf{12.23}                & \textbf{12.54*}               & 13.48                                                                                                                                                                     & 13.45                                                                                                                                   \\ 
									\cline{2-8}
									& (0.127)                                                                            & (0.135)                       & (0.315)                       & (0.181)                       & (0.264)                       & (0.282)                                                                                                                                                                   & (0.277)                                                                                                                                 \\ 
									\hline
									\multirow{2}{*}{70}                                                                                 & 7.37                                                                               & 13.56                         & 13.43$^{\dag}$                         & \textbf{12.69}                & \textbf{12.97*}               & 14.43                                                                                                                                                                     & 14.41                                                                                                                                   \\ 
									\cline{2-8}
									& (0.158)                                                                            & (0.176)                       & (0.403)                       & (0.227)                       & (0.335)                       & (0.346)                                                                                                                                                                   & (0.334)                                                                                                                                 \\ 
									\hline
									\multirow{2}{*}{75}                                                                                 & 7.98                                                                               & 14.20                         & 13.96$^{\dag}$                          & \textbf{13.25}                & \textbf{13.49*}               & 15.58                                                                                                                                                                     & 15.59                                                                                                                                   \\ 
									\cline{2-8}
									& (0.153)                                                                            & (0.198)                       & (0.502)                       & (0.260)                       & (0.425)                       & (0.385)                                                                                                                                                                   & (0.341)                                                                                                                                 \\ 
									\hline
									\multirow{2}{*}{80}                                                                                 & 8.67                                                                               & 15.02                         & 14.62$^{\dag}$                          & \textbf{13.94}                & \textbf{14.13*}               & 16.89                                                                                                                                                                     & 16.98                                                                                                                                   \\ 
									\cline{2-8}
									& (0.197)                                                                            & (0.277)                       & (0.650)                       & (0.344)                       & (0.553)                       & (0.501)                                                                                                                                                                   & (0.448)                                                                                                                                 \\ 
									\hline
									\multirow{2}{*}{85}                                                                                 & 9.44                                                                               & 16.03                         & 15.40$^{\dag}$                          & \textbf{14.74}                & \textbf{14.88}                & 18.37                                                                                                                                                                     & 18.55                                                                                                                                   \\ 
									\cline{2-8}
									& (0.294)                                                                            & (0.428)                       & (0.868)                       & (0.497)                       & (0.744)                       & (0.714)                                                                                                                                                                   & (0.643)                                                                                                                                 \\ 
									\hline
									\multirow{2}{*}{90}                                                                                 & 10.35                                                                              & 17.41                         & 16.37$^{\dag}$                          & \textbf{15.73}                & \textbf{15.80}                & 20.13                                                                                                                                                                     & 20.43                                                                                                                                   \\ 
									\cline{2-8}
									& (0.438)                                                                            & (0.701)                       & (1.181)                       & (0.728)                       & (1.010)                       & (1.137)                                                                                                                                                                   & (0.960)                                                                                                                                 \\ 
									\hline
									\multirow{2}{*}{95}                                                                                 & 11.83                                                                              & 20.08                         & 17.93$^{\dag}$                          & \textbf{17.34}                & \textbf{17.26}                & 23.18                                                                                                                                                                     & 23.76                                                                                                                                   \\ 
									\cline{2-8}
									& (0.606)                                                                            & (1.111)                       & (1.642)                       & (1.044)                       & (1.379)                       & (1.746)                                                                                                                                                                   & (1.481)                                                                                                                                 \\
									\hline
								\end{tabular}
								\begin{tablenotes}
									\item[(*)]\small No significant difference when comparing the  empirical distributions of $((f),(g),(h))$ and benchmark $(e)$ sequences of estimates in quantiles $[\Lambda,1-\Lambda]$ with $\Lambda=0.05$ at significance level of $95\%$ (Sliced Wasserstein Distance).
									\item[($\dag$)]  No significant difference when comparing the empirical distributions of $(d)$ and benchmark $(c)$ sequences of estimates in quantiles $[\Lambda,1-\Lambda]$ with $\Lambda=0.05$  at significance level of $95\%$ (Sliced Wasserstein Distance).				
									\item[1] $\scaleobj{0.9}{X=h_3(\zeta(Z), \zeta(\eta)) = \upalpha_1\left[\upalpha_2 \zeta(Z)^{-\uprho}+\upalpha_3 \zeta(\eta)^{-\uprho}\right]^{-\upalpha_4/\uprho}.  }$
									\item[2] $\scaleobj{0.9}{Y=g_5(X,\bm\epsilon)  =  \upbeta_1 X + \upbeta_2\epsilon +\upbeta_3 X\epsilon.}$	
									\item[(c)] $\scaleobj{0.9}{\widehat{\mathbb{E}}\left[Y_{\text{real}}^{\text{O}}|X_{\text{real}}^{\text{O}}=x\right]}$ is the empirical variant of $\mathds{E}\left[Y|X=x\right]$ in the real data.				
									\item[(d)] $\scaleobj{0.9}{\widehat{\mathbb{E}}\left[Y_{\text{synt}}^{\text{O}}|X_{\text{synt}}^{\text{O}}=x\right]}$ is the empirical variant of $\mathds{E}\left[Y|X=x\right]$ in the synthetic data.
									\item[(e)] $\scaleobj{0.9}{\widehat{\mathbb{E}}\left[Y_{\text{real}}^{\text{CF}}|X_{\text{real}}^{\text{CF}}=x\right]}$ is the empirical variant of $\mathds{E}\left[Y|\text{do}(X=x)\right]$ in the real data.
									\item[(f)] $\scaleobj{0.9}{\widehat{\mathbb{E}}\left[Y_{\text{synt}}^{\text{CF}}|X_{\text{synt}}^{\text{CF}}=x\right]}$ is the empirical variant of $\mathds{E}\left[Y|\text{do}(X=x)\right]$ in the synthetic data.		
									\item [(g)] Control variable estimator (IN)  \citep{imbens2009identification}.
									\item [(h)] Partial Means estimator (PM) \citep{newey1994kernel}.
								\end{tablenotes}
						\end{threeparttable}        }
					\end{table}
					
					\begin{table}[H]\centering\caption{\label{tab:tranglog:Average}Do-intervention (Translog utility \citep{christensen1975})\\ \small Entries represent average (standard error) of estimates\\ obtained from $100$ samples of 10,000 observations each\\ 
							$X=h_{1}\left(\zeta(Z), \zeta(\eta)\right)^{1}$; \hspace{1em} $Y=g_5(X,\bm\epsilon)^{2}$; \hspace{1em}
							$\upzeta(x)=1.5+1.5\tanh(0.15 x)$\\\vspace{1mm}		
							$[\upalpha_0,\upalpha_1,\upalpha_2,\upalpha_3,\upalpha_4,\upalpha_5]=[0,5,10,0,-26.25,3.25]$\\ $[\upbeta_1,\upbeta_2,\upbeta_3]=[8,-1,6]$,\hspace{1em} $[\upgamma_1,\upgamma_2,\upgamma_3]=[1,1,-3]$}\vspace{1mm}
						\scaleobj{0.80}{\begin{threeparttable}	
								\begin{tabular}{|c|c|c|c|c|c|c|c|}
									\hline
									\multirow{2}{*}{\begin{tabular}[c]{@{}c@{}}$\overset{(a)}{\text{Quantile}}$\\  of  X\end{tabular}} & \multirow{2}{*}{\begin{tabular}[c]{@{}c@{}}\\ $\overset{(b)}{X}$\end{tabular}} & \multicolumn{2}{c|}{$\mathbb{E}\left[Y|X=x\right]$} & \multicolumn{4}{c|}{$\mathbb{E}\left[Y|\text{do}(X=x)\right]$}    \\ \cline{3-8} 
									&                    & $\overset{(c)}{\text{real}}$           & $\overset{(d)}{\text{synt}}$             & $\overset{(e)}{\text{real}}$  & $\overset{(f)}{\text{synt}}$    & $\overset{(g)}{\text{IN}}$ & $\overset{(h)}{\text{PM}}$ \\ \hline\multirow{2}{*}{5}                                                                                  & 1.06                                                                               & 15.92                         & 16.29$^{\dag}$                          & \textbf{14.32}                & \textbf{14.64*}               & 10.17                                                                                                                                                                     & 8.55                                                                                                                                    \\ 
									\cline{2-8}
									& (0.050)                                                                            & (0.393)                       & (0.589)                       & (0.272)                       & (0.543)                       & (1.809)                                                                                                                                                                   & (1.234)                                                                                                                                 \\ 
									\hline
									\multirow{2}{*}{10}                                                                                 & 1.40                                                                               & 17.66                         & 17.81$^{\dag}$                          & \textbf{15.31}                & \textbf{15.49*}               & 15.05                                                                                                                                                                     & 11.66                                                                                                                                   \\ 
									\cline{2-8}
									& (0.030)                                                                            & (0.229)                       & (0.481)                       & (0.207)                       & (0.460)                       & (1.365)                                                                                                                                                                   & (0.976)                                                                                                                                 \\ 
									\hline
									\multirow{2}{*}{15}                                                                                 & 1.61                                                                               & 18.50                         & 18.55$^{\dag}$                          & \textbf{15.88}                & \textbf{15.97*}               & 18.39                                                                                                                                                                     & 13.53                                                                                                                                   \\ 
									\cline{2-8}
									& (0.028)                                                                            & (0.185)                       & (0.458)                       & (0.190)                       & (0.439)                       & (1.234)                                                                                                                                                                   & (0.959)                                                                                                                                 \\ 
									\hline
									\multirow{2}{*}{20}                                                                                 & 1.75                                                                               & 19.00                         & 19.00$^{\dag}$                          & \textbf{16.27}                & \textbf{16.30*}               & 20.95                                                                                                                                                                     & 14.71                                                                                                                                   \\ 
									\cline{2-8}
									& (0.023)                                                                            & (0.168)                       & (0.432)                       & (0.179)                       & (0.420)                       & (1.087)                                                                                                                                                                   & (0.857)                                                                                                                                 \\ 
									\hline
									\multirow{2}{*}{25}                                                                                 & 1.88                                                                               & 19.41                         & 19.37$^{\dag}$                          & \textbf{16.63}                & \textbf{16.61*}               & 22.62                                                                                                                                                                     & 15.68                                                                                                                                   \\ 
									\cline{2-8}
									& (0.027)                                                                            & (0.168)                       & (0.417)                       & (0.175)                       & (0.407)                       & (0.963)                                                                                                                                                                   & (0.803)                                                                                                                                 \\ 
									\hline
									\multirow{2}{*}{30}                                                                                 & 1.99                                                                               & 19.71                         & 19.64$^{\dag}$                          & \textbf{16.93}                & \textbf{16.86*}               & 23.22                                                                                                                                                                     & 16.43                                                                                                                                   \\ 
									\cline{2-8}
									& (0.030)                                                                            & (0.165)                       & (0.406)                       & (0.167)                       & (0.395)                       & (0.880)                                                                                                                                                                   & (0.823)                                                                                                                                 \\ 
									\hline
									\multirow{2}{*}{35}                                                                                 & 2.05                                                                               & 19.86                         & 19.78$^{\dag}$                          & \textbf{17.09}                & \textbf{17.00*}               & 23.28                                                                                                                                                                     & 16.81*                                                                                                                                  \\ 
									\cline{2-8}
									& (0.017)                                                                            & (0.133)                       & (0.402)                       & (0.157)                       & (0.392)                       & (0.880)                                                                                                                                                                   & (0.835)                                                                                                                                 \\ 
									\hline
									\multirow{2}{*}{40}                                                                                 & 2.07                                                                               & 19.93                         & 19.84$^{\dag}$                          & \textbf{17.18}                & \textbf{17.07*}               & 23.24                                                                                                                                                                     & 16.99*                                                                                                                                  \\ 
									\cline{2-8}
									& (0.017)                                                                            & (0.124)                       & (0.402)                       & (0.162)                       & (0.396)                       & (0.893)                                                                                                                                                                   & (0.845)                                                                                                                                 \\ 
									\hline
									\multirow{2}{*}{45}                                                                                 & 2.11                                                                               & 20.01                         & 19.91$^{\dag}$                          & \textbf{17.27}                & \textbf{17.15*}               & 23.16                                                                                                                                                                     & 17.18*                                                                                                                                  \\ 
									\cline{2-8}
									& (0.018)                                                                            & (0.121)                       & (0.399)                       & (0.163)                       & (0.399)                       & (0.916)                                                                                                                                                                   & (0.852)                                                                                                                                 \\ 
									\hline
									\multirow{2}{*}{50}                                                                                 & 2.16                                                                               & 20.13                         & 20.02$^{\dag}$                          & \textbf{17.42}                & \textbf{17.28*}               & 22.99                                                                                                                                                                     & 17.48*                                                                                                                                  \\ 
									\cline{2-8}
									& (0.026)                                                                            & (0.118)                       & (0.406)                       & (0.177)                       & (0.418)                       & (0.946)                                                                                                                                                                   & (0.835)                                                                                                                                 \\ 
									\hline
									\multirow{2}{*}{55}                                                                                 & 2.25                                                                               & 20.32                         & 20.20$^{\dag}$                          & \textbf{17.68}                & \textbf{17.50*}               & 22.58                                                                                                                                                                     & 18.04*                                                                                                                                  \\ 
									\cline{2-8}
									& (0.039)                                                                            & (0.135)                       & (0.419)                       & (0.202)                       & (0.444)                       & (0.938)                                                                                                                                                                   & (0.793)                                                                                                                                 \\ 
									\hline
									\multirow{2}{*}{60}                                                                                 & 2.37                                                                               & 20.59                         & 20.44$^{\dag}$                          & \textbf{18.08}                & \textbf{17.84*}               & 22.05                                                                                                                                                                     & 18.85                                                                                                                                   \\ 
									\cline{2-8}
									& (0.028)                                                                            & (0.127)                       & (0.395)                       & (0.196)                       & (0.424)                       & (0.898)                                                                                                                                                                   & (0.728)                                                                                                                                 \\ 
									\hline
									\multirow{2}{*}{65}                                                                                 & 2.50                                                                               & 20.83                         & 20.66$^{\dag}$                          & \textbf{18.49}                & \textbf{18.19*}               & 21.82                                                                                                                                                                     & 19.68                                                                                                                                   \\ 
									\cline{2-8}
									& (0.031)                                                                            & (0.120)                       & (0.382)                       & (0.203)                       & (0.419)                       & (0.865)                                                                                                                                                                   & (0.696)                                                                                                                                 \\ 
									\hline
									\multirow{2}{*}{70}                                                                                 & 2.66                                                                               & 21.11                         & 20.93$^{\dag}$                          & \textbf{19.07}                & \textbf{18.67*}               & 22.01                                                                                                                                                                     & 20.84                                                                                                                                   \\ 
									\cline{2-8}
									& (0.032)                                                                            & (0.123)                       & (0.370)                       & (0.232)                       & (0.440)                       & (0.830)                                                                                                                                                                   & (0.701)                                                                                                                                 \\ 
									\hline
									\multirow{2}{*}{75}                                                                                 & 2.82                                                                               & 21.37                         & 21.17$^{\dag}$                          & \textbf{19.68}                & \textbf{19.18*}               & 22.61                                                                                                                                                                     & 21.99                                                                                                                                   \\ 
									\cline{2-8}
									& (0.041)                                                                            & (0.119)                       & (0.360)                       & (0.253)                       & (0.440)                       & (0.776)                                                                                                                                                                   & (0.666)                                                                                                                                 \\ 
									\hline
									\multirow{2}{*}{80}                                                                                 & 3.07                                                                               & 21.81                         & 21.58$^{\dag}$                          & \textbf{20.79}                & \textbf{20.10*}               & 24.04                                                                                                                                                                     & 23.83                                                                                                                                   \\ 
									\cline{2-8}
									& (0.052)                                                                            & (0.134)                       & (0.392)                       & (0.343)                       & (0.584)                       & (0.747)                                                                                                                                                                   & (0.750)                                                                                                                                 \\ 
									\hline
									\multirow{2}{*}{85}                                                                                 & 3.34                                                                               & 22.33                         & 22.07$^{\dag}$                          & \textbf{22.16}                & \textbf{21.20}                & 25.70                                                                                                                                                                     & 25.59                                                                                                                                   \\ 
									\cline{2-8}
									& (0.064)                                                                            & (0.165)                       & (0.374)                       & (0.468)                       & (0.608)                       & (0.775)                                                                                                                                                                   & (0.782)                                                                                                                                 \\ 
									\hline
									\multirow{2}{*}{90}                                                                                 & 3.79                                                                               & 23.58                         & 23.24$^{\dag}$                          & \textbf{25.03}                & \textbf{23.41}                & 28.79                                                                                                                                                                     & 28.73                                                                                                                                   \\ 
									\cline{2-8}
									& (0.066)                                                                            & (0.257)                       & (0.584)                       & (0.578)                       & (0.971)                       & (0.903)                                                                                                                                                                   & (0.862)                                                                                                                                 \\ 
									\hline
									\multirow{2}{*}{95}                                                                                 & 4.42                                                                               & 26.48                         & 25.55$^{\dag}$                          & \textbf{29.77}                & \textbf{26.55}                & 33.16                                                                                                                                                                     & 33.21                                                                                                                                   \\ 
									\cline{2-8}
									& (0.107)                                                                            & (0.670)                       & (1.090)                       & (1.021)                       & (1.531)                       & (1.312)                                                                                                                                                                   & (1.342)                                                                                                                                 \\
									\hline
								\end{tabular}%
								\begin{tablenotes}
									\item[(*)]\small No significant difference when comparing the  empirical distributions of $((f),(g),(h))$ and benchmark $(e)$ sequences of estimates in quantiles $[\Lambda,1-\Lambda]$ with $\Lambda=0.05$  at significance level of $95\%$ (Sliced Wasserstein Distance).
									\item[($\dag$)]  No significant difference when comparing the empirical distributions of $(d)$ and benchmark $(c)$ sequences of estimates in quantiles $[\Lambda,1-\Lambda]$ with $\Lambda=0.05$  at significance level of $95\%$ (Sliced Wasserstein Distance).				
									\item[1] $\scaleobj{0.9}{X=h_1(\zeta(Z), \zeta(\eta)) =  \upalpha_0 + \upalpha_1 \log(\zeta(Z))  + \upalpha_2\log(\zeta(\eta) ) + \upalpha_3\log^2(\zeta(Z))+ \upalpha_4\log^2(\zeta(\eta))}$
									\[ \scaleobj{0.9}{+  \upalpha_5\log(\zeta(Z))\log(\zeta(\eta)).}\]\vspace{-1.5em}
									\item[2] $\scaleobj{0.9}{Y=g_5(X,\bm\epsilon)  =  \upbeta_1 X + \upbeta_2\epsilon +\upbeta_3 X\epsilon.}$	
									\item[(c)] $\scaleobj{0.9}{\widehat{\mathbb{E}}\left[Y_{\text{real}}^{\text{O}}|X_{\text{real}}^{\text{O}}=x\right]}$ is the empirical variant of $\mathds{E}\left[Y|X=x\right]$ in the real data.								
									\item[(d)] $\scaleobj{0.9}{\widehat{\mathbb{E}}\left[Y_{\text{synt}}^{\text{O}}|X_{\text{synt}}^{\text{O}}=x\right]}$ is the empirical variant of $\mathds{E}\left[Y|X=x\right]$ in the synthetic data.
									\item[(e)] $\scaleobj{0.9}{\widehat{\mathbb{E}}\left[Y_{\text{real}}^{\text{CF}}|X_{\text{real}}^{\text{CF}}=x\right]}$ is the empirical variant of $\mathds{E}\left[Y|\text{do}(X=x)\right]$ in the real data.
									\item[(f)] $\scaleobj{0.9}{\widehat{\mathbb{E}}\left[Y_{\text{synt}}^{\text{CF}}|X_{\text{synt}}^{\text{CF}}=x\right]}$ is the empirical variant of $\mathds{E}\left[Y|\text{do}(X=x)\right]$ in the synthetic data.		
									\item [(g)] Control variable estimator (IN)  \citep{imbens2009identification}.
									\item [(h)] Partial Means estimator (PM) \citep{newey1994kernel}.						
								\end{tablenotes}
						\end{threeparttable}        }
					\end{table}

					% Please add the following required packages to your document preamble: https://www.latex-tables.com/
					% \usepackage{multirow}
					% \usepackage{graphicx}
					\begin{table}[H]\centering\caption{\label{tab:tanhtanh:Average}Do-intervention ($\tanh$-based non-monotonic function)\\ \small Entries represent average (standard error) of estimates\\ obtained from $100$ samples of 10,000 observations each\\ 
							$X=h_{6}\left(Z,\eta\right)^{1}$; \hspace{1em} $Y=g_6(X,\bm\epsilon)^{2}$\\\vspace{1mm}		
							$[\upalpha_1,\upalpha_2,\upalpha_3,\upalpha_4,\upalpha_5]=[6,0.25,-0.5,5,0.5]$\\ $[\upbeta_1,\upbeta_2,\upbeta_3,\upbeta_4,\upbeta_5]=[6,0.25,-0.5,10,0.5]$,\hspace{1em} $[\upgamma_1,\upgamma_2,\upgamma_3]=[1,1,-3]$}\vspace{1mm}
						\scaleobj{0.8}{\begin{threeparttable}	
								\begin{tabular}{|c|c|c|c|c|c|c|c|}
									\hline
									\multirow{2}{*}{\begin{tabular}[c]{@{}c@{}}$\overset{(a)}{\text{Quantile}}$\\  of  X\end{tabular}} & \multirow{2}{*}{\begin{tabular}[c]{@{}c@{}}\\ $\overset{(b)}{X}$\end{tabular}} & \multicolumn{2}{c|}{$\mathbb{E}\left[Y|X=x\right]$} & \multicolumn{4}{c|}{$\mathbb{E}\left[Y|\text{do}(X=x)\right]$}    \\ \cline{3-8} 
									&                    & $\overset{(c)}{\text{real}}$           & $\overset{(d)}{\text{synt}}$             & $\overset{(e)}{\text{real}}$  & $\overset{(f)}{\text{synt}}$    & $\overset{(g)}{\text{IN}}$ & $\overset{(h)}{\text{PM}}$ \\ \hline\multirow{2}{*}{5}                                                                                   & -        2.02                                                                      & -        1.05                 & -        0.95$^{\dag}$                 & \textbf{- 0.36}               & \textbf{- 0.31*}              & -        1.35                                                                                                                                                             & -        1.11                                                                                                                           \\ 
									\cline{2-8}
									& (0.034)                                                                            & (0.042)                       & (0.064)                       & (0.042)                       & (0.116)                       & (0.748)                                                                                                                                                                   & (0.261)                                                                                                                                 \\ 
									\hline
									\multirow{2}{*}{10}                                                                                  & -        1.63                                                                      & -        0.85                 & -        0.75$^{\dag}$                 & \textbf{- 0.25}               & \textbf{- 0.22*}              & -        1.02                                                                                                                                                             & -        0.88                                                                                                                           \\ 
									\cline{2-8}
									& (0.028)                                                                            & (0.043)                       & (0.065)                       & (0.040)                       & (0.107)                       & (0.538)                                                                                                                                                                   & (0.242)                                                                                                                                 \\ 
									\hline
									\multirow{2}{*}{15}                                                                                  & -        1.34                                                                      & -        0.67                 & -        0.59$^{\dag}$                  & \textbf{- 0.17}               & \textbf{- 0.14*}              & -        0.56                                                                                                                                                             & -        0.70                                                                                                                           \\ 
									\cline{2-8}
									& (0.028)                                                                            & (0.044)                       & (0.065)                       & (0.038)                       & (0.098)                       & (0.387)                                                                                                                                                                   & (0.212)                                                                                                                                 \\ 
									\hline
									\multirow{2}{*}{20}                                                                                  & -        1.10                                                                      & -        0.53                 & -        0.46$^{\dag}$                  & \textbf{- 0.10}               & \textbf{- 0.08*}              & -        0.18                                                                                                                                                             & -        0.56                                                                                                                           \\ 
									\cline{2-8}
									& (0.025)                                                                            & (0.042)                       & (0.062)                       & (0.037)                       & (0.089)                       & (0.291)                                                                                                                                                                   & (0.189)                                                                                                                                 \\ 
									\hline
									\multirow{2}{*}{25}                                                                                  & -        0.87                                                                      & -        0.38                 & -        0.32$^{\dag}$                  & \textbf{- 0.04}               & \textbf{- 0.01*}              & 0.05                                                                                                                                                                      & -        0.40                                                                                                                           \\ 
									\cline{2-8}
									& (0.026)                                                                            & (0.045)                       & (0.065)                       & (0.036)                       & (0.086)                       & (0.233)                                                                                                                                                                   & (0.183)                                                                                                                                 \\ 
									\hline
									\multirow{2}{*}{30}                                                                                  & -        0.68                                                                      & -        0.26                 & -        0.20$^{\dag}$                  & \textbf{0.02}                 & \textbf{0.04*}                & 0.14                                                                                                                                                                      & -        0.27                                                                                                                           \\ 
									\cline{2-8}
									& (0.028)                                                                            & (0.045)                       & (0.066)                       & (0.037)                       & (0.082)                       & (0.221)                                                                                                                                                                   & (0.186)                                                                                                                                 \\ 
									\hline
									\multirow{2}{*}{35}                                                                                  & -        0.49                                                                      & -        0.13                 & -        0.08$^{\dag}$                  & \textbf{0.07}                 & \textbf{0.09*}                & 0.18                                                                                                                                                                      & -        0.13                                                                                                                           \\ 
									\cline{2-8}
									& (0.035)                                                                            & (0.046)                       & (0.065)                       & (0.037)                       & (0.078)                       & (0.222)                                                                                                                                                                   & (0.186)                                                                                                                                 \\ 
									\hline
									\multirow{2}{*}{40}                                                                                  & -        0.34                                                                      & -        0.03                 & 0.01$^{\dag}$                           & \textbf{0.12}                 & \textbf{0.14*}                & 0.19                                                                                                                                                                      & -        0.03                                                                                                                           \\ 
									\cline{2-8}
									& (0.029)                                                                            & (0.043)                       & (0.064)                       & (0.039)                       & (0.070)                       & (0.220)                                                                                                                                                                   & (0.181)                                                                                                                                 \\ 
									\hline
									\multirow{2}{*}{45}                                                                                  & -        0.15                                                                      & 0.10                          & 0.14$^{\dag}$                           & \textbf{0.17}                 & \textbf{0.19*}                & 0.22                                                                                                                                                                      & 0.10*                                                                                                                                   \\ 
									\cline{2-8}
									& (0.040)                                                                            & (0.052)                       & (0.071)                       & (0.039)                       & (0.072)                       & (0.215)                                                                                                                                                                   & (0.183)                                                                                                                                 \\ 
									\hline
									\multirow{2}{*}{50}                                                                                  & 0.03                                                                               & 0.23                          & 0.26$^{\dag}$                           & \textbf{0.23}                 & \textbf{0.25*}                & 0.25                                                                                                                                                                      & 0.24*                                                                                                                                  \\ 
									\cline{2-8}
									& (0.028)                                                                            & (0.043)                       & (0.068)                       & (0.039)                       & (0.067)                       & (0.217)                                                                                                                                                                   & (0.180)                                                                                                                                 \\ 
									\hline
									\multirow{2}{*}{55}                                                                                  & 0.12                                                                               & 0.29                          & 0.31$^{\dag}$                           & \textbf{0.26}                 & \textbf{0.27*}                & 0.26                                                                                                                                                                      & 0.30*                                                                                                                                   \\ 
									\cline{2-8}
									& (0.035)                                                                            & (0.044)                       & (0.065)                       & (0.041)                       & (0.062)                       & (0.218)                                                                                                                                                                   & (0.180)                                                                                                                                 \\ 
									\hline
									\multirow{2}{*}{60}                                                                                  & 0.32                                                                               & 0.43                          & 0.44$^{\dag}$                           & \textbf{0.32}                 & \textbf{0.33*}                & 0.31                                                                                                                                                                      & 0.45*                                                                                                                                   \\ 
									\cline{2-8}
									& (0.034)                                                                            & (0.045)                       & (0.070)                       & (0.038)                       & (0.064)                       & (0.222)                                                                                                                                                                   & (0.183)                                                                                                                                 \\ 
									\hline
									\multirow{2}{*}{65}                                                                                  & 0.52                                                                               & 0.56                          & 0.57$^{\dag}$                           & \textbf{0.37}                 & \textbf{0.39*}                & 0.34                                                                                                                                                                      & 0.60                                                                                                                                    \\ 
									\cline{2-8}
									& (0.027)                                                                            & (0.042)                       & (0.067)                       & (0.039)                       & (0.063)                       & (0.225)                                                                                                                                                                   & (0.181)                                                                                                                                 \\ 
									\hline
									\multirow{2}{*}{70}                                                                                  & 0.70                                                                               & 0.69                          & 0.68$^{\dag}$                           & \textbf{0.43}                 & \textbf{0.44*}                & 0.40                                                                                                                                                                      & 0.73                                                                                                                                    \\ 
									\cline{2-8}
									& (0.030)                                                                            & (0.040)                       & (0.067)                       & (0.039)                       & (0.059)                       & (0.237)                                                                                                                                                                   & (0.183)                                                                                                                                 \\ 
									\hline
									\multirow{2}{*}{75}                                                                                  & 0.90                                                                               & 0.82                          & 0.80$^{\dag}$                           & \textbf{0.49}                 & \textbf{0.49*}                & 0.50                                                                                                                                                                      & 0.88                                                                                                                                    \\ 
									\cline{2-8}
									& (0.024)                                                                            & (0.040)                       & (0.066)                       & (0.040)                       & (0.061)                       & (0.268)                                                                                                                                                                   & (0.187)                                                                                                                                 \\ 
									\hline
									\multirow{2}{*}{80}                                                                                  & 1.09                                                                               & 0.94                          & 0.92$^{\dag}$                           & \textbf{0.54}                 & \textbf{0.55*}                & 0.69                                                                                                                                                                      & 1.01                                                                                                                                    \\ 
									\cline{2-8}
									& (0.026)                                                                            & (0.040)                       & (0.068)                       & (0.040)                       & (0.062)                       & (0.341)                                                                                                                                                                   & (0.196)                                                                                                                                 \\ 
									\hline
									\multirow{2}{*}{85}                                                                                  & 1.33                                                                               & 1.09                          & 1.05$^{\dag}$                           & \textbf{0.61}                 & \textbf{0.61*}                & 1.04                                                                                                                                                                      & 1.16                                                                                                                                    \\ 
									\cline{2-8}
									& (0.027)                                                                            & (0.041)                       & (0.071)                       & (0.042)                       & (0.066)                       & (0.454)                                                                                                                                                                   & (0.216)                                                                                                                                 \\ 
									\hline
									\multirow{2}{*}{90}                                                                                  & 1.63                                                                               & 1.26                          & 1.21$^{\dag}$                           & \textbf{0.70}                 & \textbf{0.69*}                & 1.47                                                                                                                                                                      & 1.35                                                                                                                                    \\ 
									\cline{2-8}
									& (0.030)                                                                            & (0.041)                       & (0.072)                       & (0.042)                       & (0.069)                       & (0.549)                                                                                                                                                                   & (0.255)                                                                                                                                 \\ 
									\hline
									\multirow{2}{*}{95}                                                                                  & 2.02                                                                               & 1.47                          & 1.40$^{\dag}$                           & \textbf{0.81}                 & \textbf{0.79*}                & 1.84                                                                                                                                                                      & 1.64                                                                                                                                    \\ 
									\cline{2-8}
									& (0.029)                                                                            & (0.045)                       & (0.073)                       & (0.045)                       & (0.076)                       & (0.671)                                                                                                                                                                   & (0.300)                                                                                                                                 \\
									\hline
								\end{tabular}%
								\begin{tablenotes}
									\item[(*)]\small No significant difference when comparing the  empirical distributions of $((f),(g),(h))$ and benchmark $(e)$ sequences of estimates in quantiles $[\Lambda,1-\Lambda]$ with $\Lambda=0.05$  at significance level of $95\%$ (Sliced Wasserstein Distance).
									\item[($\dag$)]  No significant difference when comparing the empirical distributions of $(d)$ and benchmark $(c)$ sequences of estimates in quantiles $[\Lambda,1-\Lambda]$ with $\Lambda=0.05$  at significance level of $95\%$ (Sliced Wasserstein Distance).						
									\item[1] $\scaleobj{0.9}{X=h_6(Z,\eta) =  \upalpha_1\tanh(\upalpha_2 Z+\upalpha_3 \upeta)+\upalpha_4\tanh(\upalpha_5 \upeta).}$
									\item[2] $\scaleobj{0.9}{Y=g_6(X,\bm\epsilon)  =  \upbeta_1\tanh(\upbeta_2 Z+\upbeta_3 \epsilon)+\upbeta_4\tanh(\upbeta_5 \epsilon). }$	
									\item[(c)] $\scaleobj{0.9}{\widehat{\mathbb{E}}\left[Y_{\text{real}}^{\text{O}}|X_{\text{real}}^{\text{O}}=x\right]}$ is the empirical variant of $\mathds{E}\left[Y|X=x\right]$ in the real data.				
									
									\item[(d)] $\scaleobj{0.9}{\widehat{\mathbb{E}}\left[Y_{\text{synt}}^{\text{O}}|X_{\text{synt}}^{\text{O}}=x\right]}$ is the empirical variant of $\mathds{E}\left[Y|X=x\right]$ in the synthetic data.
									\item[(e)] $\scaleobj{0.9}{\widehat{\mathbb{E}}\left[Y_{\text{real}}^{\text{CF}}|X_{\text{real}}^{\text{CF}}=x\right]}$ is the empirical variant of $\mathds{E}\left[Y|\text{do}(X=x)\right]$ in the real data.
									\item[(f)] $\scaleobj{0.9}{\widehat{\mathbb{E}}\left[Y_{\text{synt}}^{\text{CF}}|X_{\text{synt}}^{\text{CF}}=x\right]}$ is the empirical variant of $\mathds{E}\left[Y|\text{do}(X=x)\right]$ in the synthetic data.	
									\item [(g)] Control variable estimator (IN)  \citep{imbens2009identification}.
									\item [(h)] Partial Means estimator (PM) \citep{newey1994kernel}.	
								\end{tablenotes}
						\end{threeparttable}        }
					\end{table}

					\begin{table}[H]\centering\caption{\label{tab:AIDS:Average}Do-intervention (AIDS function \citep{deaton1980almost})\\ \small Entries represent average (standard error) of estimates\\ obtained from $100$ samples of 10,000 observations each\\ $X=h_{2}\left(\zeta(Z), \zeta(\eta)\right)^{1}$; \hspace{1em} $Y=g_5(X,\bm\epsilon)^{2}$; \hspace{1em} $\upzeta(x)=1.5+1.5\tanh(0.15 x)$\\\vspace{1mm}
							$[\upalpha_0,\upalpha_1,\upalpha_2,\upalpha_3,\upalpha_4,\upalpha_5,\upalpha_6,\uprho]=[0,5,10,0,-26.25,3.25,-0.15,0.5]$\\ $[\upbeta_1,\upbeta_2,\upbeta_3]=[8,-1,6]$,\hspace{1em} $[\upgamma_1,\upgamma_2,\upgamma_3]=[1,1,1]$}\vspace{1mm}
						\scaleobj{0.8}{\begin{threeparttable}	
								\begin{tabular}{|c|c|c|c|c|c|c|c|}
									\hline
									\multirow{2}{*}{\begin{tabular}[c]{@{}c@{}}$\overset{(a)}{\text{Quantile}}$\\  of  X\end{tabular}} & \multirow{2}{*}{\begin{tabular}[c]{@{}c@{}}\\ $\overset{(b)}{X}$\end{tabular}} & \multicolumn{2}{c|}{$\mathbb{E}\left[Y|X=x\right]$} & \multicolumn{4}{c|}{$\mathbb{E}\left[Y|\text{do}(X=x)\right]$}    \\ \cline{3-8} 
									&                    & $\overset{(c)}{\text{real}}$           & $\overset{(d)}{\text{synt}}$             & $\overset{(e)}{\text{real}}$  & $\overset{(f)}{\text{synt}}$    & $\overset{(g)}{\text{IN}}$ & $\overset{(h)}{\text{PM}}$ \\ \hline\multirow{2}{*}{5}                                                                                  & 0.82                                                                               & 13.76                         & 14.34$^{\dag}$                          & \textbf{12.29}                & \textbf{12.56*}               & 8.17                                                                                                                                                                      & 6.59                                                                                                                                    \\ 
									\cline{2-8}
									& (0.041)                                                                            & (0.389)                       & (0.581)                       & (0.262)                       & (0.489)                       & (1.816)                                                                                                                                                                   & (1.252)                                                                                                                                 \\ 
									\hline
									\multirow{2}{*}{10}                                                                                 & 1.14                                                                               & 15.69                         & 16.00$^{\dag}$                          & \textbf{13.32}                & \textbf{13.60*}               & 12.45*                                                                                                                                                                    & 9.53                                                                                                                                    \\ 
									\cline{2-8}
									& (0.036)                                                                            & (0.254)                       & (0.515)                       & (0.225)                       & (0.436)                       & (1.416)                                                                                                                                                                   & (1.133)                                                                                                                                 \\ 
									\hline
									\multirow{2}{*}{15}                                                                                 & 1.37                                                                               & 16.73                         & 16.92$^{\dag}$                          & \textbf{13.98}                & \textbf{14.23*}               & 16.13                                                                                                                                                                     & 11.54                                                                                                                                   \\ 
									\cline{2-8}
									& (0.033)                                                                            & (0.210)                       & (0.469)                       & (0.206)                       & (0.407)                       & (1.175)                                                                                                                                                                   & (1.016)                                                                                                                                 \\ 
									\hline
									\multirow{2}{*}{20}                                                                                 & 1.51                                                                               & 17.30                         & 17.42$^{\dag}$                          & \textbf{14.40}                & \textbf{14.61*}               & 18.94                                                                                                                                                                     & 12.77                                                                                                                                   \\ 
									\cline{2-8}
									& (0.026)                                                                            & (0.183)                       & (0.435)                       & (0.189)                       & (0.382)                       & (0.911)                                                                                                                                                                   & (0.909)                                                                                                                                 \\ 
									\hline
									\multirow{2}{*}{25}                                                                                 & 1.64                                                                               & 17.72                         & 17.80$^{\dag}$                          & \textbf{14.75}                & \textbf{14.94*}               & 20.87                                                                                                                                                                     & 13.72                                                                                                                                   \\ 
									\cline{2-8}
									& (0.027)                                                                            & (0.182)                       & (0.403)                       & (0.176)                       & (0.352)                       & (0.819)                                                                                                                                                                   & (0.852)                                                                                                                                 \\ 
									\hline
									\multirow{2}{*}{30}                                                                                 & 1.74                                                                               & 18.04                         & 18.08$^{\dag}$                          & \textbf{15.05}                & \textbf{15.20*}               & 21.74                                                                                                                                                                     & 14.47*                                                                                                                                  \\ 
									\cline{2-8}
									& (0.022)                                                                            & (0.168)                       & (0.398)                       & (0.169)                       & (0.342)                       & (0.824)                                                                                                                                                                   & (0.802)                                                                                                                                 \\ 
									\hline
									\multirow{2}{*}{35}                                                                                 & 1.82                                                                               & 18.27                         & 18.29$^{\dag}$                          & \textbf{15.28}                & \textbf{15.41*}               & 21.95                                                                                                                                                                     & 15.03*                                                                                                                                  \\ 
									\cline{2-8}
									& (0.017)                                                                            & (0.155)                       & (0.396)                       & (0.165)                       & (0.339)                       & (0.911)                                                                                                                                                                   & (0.793)                                                                                                                                 \\ 
									\hline
									\multirow{2}{*}{40}                                                                                 & 1.85                                                                               & 18.37                         & 18.38$^{\dag}$                          & \textbf{15.39}                & \textbf{15.50*}               & 21.94                                                                                                                                                                     & 15.28*                                                                                                                                  \\ 
									\cline{2-8}
									& (0.010)                                                                            & (0.147)                       & (0.391)                       & (0.161)                       & (0.336)                       & (0.952)                                                                                                                                                                   & (0.781)                                                                                                                                 \\ 
									\hline
									\multirow{2}{*}{45}                                                                                 & 1.89                                                                               & 18.45                         & 18.46$^{\dag}$                          & \textbf{15.50}                & \textbf{15.60*}               & 21.88                                                                                                                                                                     & 15.50*                                                                                                                                  \\ 
									\cline{2-8}
									& (0.013)                                                                            & (0.147)                       & (0.388)                       & (0.165)                       & (0.339)                       & (0.991)                                                                                                                                                                   & (0.777)                                                                                                                                 \\ 
									\hline
									\multirow{2}{*}{50}                                                                                 & 1.95                                                                               & 18.60                         & 18.59$^{\dag}$                          & \textbf{15.68}                & \textbf{15.76*}               & 21.68                                                                                                                                                                     & 15.89*                                                                                                                                  \\ 
									\cline{2-8}
									& (0.026)                                                                            & (0.152)                       & (0.390)                       & (0.187)                       & (0.349)                       & (1.051)                                                                                                                                                                   & (0.781)                                                                                                                                 \\ 
									\hline
									\multirow{2}{*}{55}                                                                                 & 2.07                                                                               & 18.88                         & 18.85$^{\dag}$                          & \textbf{16.06}                & \textbf{16.09*}               & 21.16                                                                                                                                                                     & 16.67                                                                                                                                   \\ 
									\cline{2-8}
									& (0.028)                                                                            & (0.144)                       & (0.382)                       & (0.195)                       & (0.344)                       & (1.101)                                                                                                                                                                   & (0.748)                                                                                                                                 \\ 
									\hline
									\multirow{2}{*}{60}                                                                                 & 2.17                                                                               & 19.09                         & 19.04$^{\dag}$                          & \textbf{16.38}                & \textbf{16.37*}               & 20.76                                                                                                                                                                     & 17.31                                                                                                                                   \\ 
									\cline{2-8}
									& (0.023)                                                                            & (0.133)                       & (0.375)                       & (0.191)                       & (0.340)                       & (1.079)                                                                                                                                                                   & (0.739)                                                                                                                                 \\ 
									\hline
									\multirow{2}{*}{65}                                                                                 & 2.27                                                                               & 19.27                         & 19.21$^{\dag}$                          & \textbf{16.71}                & \textbf{16.65*}               & 20.54                                                                                                                                                                     & 17.95                                                                                                                                   \\ 
									\cline{2-8}
									& (0.021)                                                                            & (0.129)                       & (0.370)                       & (0.188)                       & (0.348)                       & (0.992)                                                                                                                                                                   & (0.731)                                                                                                                                 \\ 
									\hline
									\multirow{2}{*}{70}                                                                                 & 2.42                                                                               & 19.53                         & 19.46$^{\dag}$                          & \textbf{17.26}                & \textbf{17.12*}               & 20.57                                                                                                                                                                     & 18.99                                                                                                                                   \\ 
									\cline{2-8}
									& (0.041)                                                                            & (0.145)                       & (0.364)                       & (0.259)                       & (0.383)                       & (0.874)                                                                                                                                                                   & (0.778)                                                                                                                                 \\ 
									\hline
									\multirow{2}{*}{75}                                                                                 & 2.62                                                                               & 19.85                         & 19.76$^{\dag}$                          & \textbf{18.04}                & \textbf{17.79*}               & 21.25                                                                                                                                                                     & 20.38                                                                                                                                   \\ 
									\cline{2-8}
									& (0.035)                                                                            & (0.133)                       & (0.366)                       & (0.251)                       & (0.410)                       & (0.846)                                                                                                                                                                   & (0.746)                                                                                                                                 \\ 
									\hline
									\multirow{2}{*}{80}                                                                                 & 2.82                                                                               & 20.16                         & 20.07$^{\dag}$                          & \textbf{18.92}                & \textbf{18.53*}               & 22.29                                                                                                                                                                     & 21.77                                                                                                                                   \\ 
									\cline{2-8}
									& (0.041)                                                                            & (0.136)                       & (0.359)                       & (0.298)                       & (0.437)                       & (0.805)                                                                                                                                                                   & (0.784)                                                                                                                                 \\ 
									\hline
									\multirow{2}{*}{85}                                                                                 & 3.15                                                                               & 20.74                         & 20.65$^{\dag}$                          & \textbf{20.61}                & \textbf{19.93*}               & 24.26                                                                                                                                                                     & 24.09                                                                                                                                   \\ 
									\cline{2-8}
									& (0.047)                                                                            & (0.151)                       & (0.389)                       & (0.381)                       & (0.564)                       & (0.853)                                                                                                                                                                   & (0.859)                                                                                                                                 \\ 
									\hline
									\multirow{2}{*}{90}                                                                                 & 3.50                                                                               & 21.60                         & 21.51$^{\dag}$                          & \textbf{22.79}                & \textbf{21.66}                & 26.62                                                                                                                                                                     & 26.58                                                                                                                                   \\ 
									\cline{2-8}
									& (0.076)                                                                            & (0.279)                       & (0.487)                       & (0.669)                       & (0.758)                       & (0.944)                                                                                                                                                                   & (0.975)                                                                                                                                 \\ 
									\hline
									\multirow{2}{*}{95}                                                                                 & 4.17                                                                               & 24.41                         & 24.07$^{\dag}$                          & \textbf{27.78}                & \textbf{25.33}                & 31.08                                                                                                                                                                     & 31.09                                                                                                                                   \\ 
									\cline{2-8}
									& (0.088)                                                                            & (0.522)                       & (0.920)                       & (0.908)                       & (1.373)                       & (1.262)                                                                                                                                                                   & (1.209)                                                                                                                                 \\
									\hline
								\end{tabular}%
								\begin{tablenotes}
									\item[(*)]\small No significant difference when comparing the  empirical distributions of $((f),(g),(h))$ and benchmark $(e)$ sequences of estimates in quantiles $[\Lambda,1-\Lambda]$ with $\Lambda=0.05$  at significance level of $95\%$ (Sliced Wasserstein Distance).
									\item[($\dag$)]  No significant difference when comparing the empirical distributions of $(d)$ and benchmark $(c)$ sequences of estimates in quantiles $[\Lambda,1-\Lambda]$ with $\Lambda=0.05$  at significance level of $95\%$ (Sliced Wasserstein Distance).				
									\item[1] $\scaleobj{0.9}{X=h_2(\zeta(Z), \zeta(\eta)) =  \upalpha_0 + \upalpha_1 \log(\zeta(Z))  + \upalpha_2\log(\zeta(\eta) ) + \upalpha_3\log^2(\zeta(Z))+ \upalpha_4\log^2(\zeta(\eta))}$
									\[ \scaleobj{0.9}{+  \upalpha_5\log(\zeta(Z))\log(\zeta(\eta)) + \upalpha_6 (\zeta(Z)\zeta(\eta))^{\uprho}.}\]\vspace{-1.5em}
									\item[2] $\scaleobj{0.9}{Y=g_5(X,\bm\epsilon)  =  \upbeta_1 X + \upbeta_2\epsilon +\upbeta_3 X\epsilon.}$	
									\item[(c)] $\scaleobj{0.9}{\widehat{\mathbb{E}}\left[Y_{\text{real}}^{\text{O}}|X_{\text{real}}^{\text{O}}=x\right]}$ is the empirical variant of $\mathds{E}\left[Y|X=x\right]$ in the real data.				
									
									\item[(d)] $\scaleobj{0.9}{\widehat{\mathbb{E}}\left[Y_{\text{synt}}^{\text{O}}|X_{\text{synt}}^{\text{O}}=x\right]}$ is the empirical variant of $\mathds{E}\left[Y|X=x\right]$ in the synthetic data.
									\item[(e)] $\scaleobj{0.9}{\widehat{\mathbb{E}}\left[Y_{\text{real}}^{\text{CF}}|X_{\text{real}}^{\text{CF}}=x\right]}$ is the empirical variant of $\mathds{E}\left[Y|\text{do}(X=x)\right]$ in the real data.
									\item[(f)] $\scaleobj{0.9}{\widehat{\mathbb{E}}\left[Y_{\text{synt}}^{\text{CF}}|X_{\text{synt}}^{\text{CF}}=x\right]}$ is the empirical variant of $\mathds{E}\left[Y|\text{do}(X=x)\right]$ in the synthetic data.			
									\item [(g)] Control variable estimator (IN)  \citep{imbens2009identification}.
									\item [(h)] Partial Means estimator (PM) \citep{newey1994kernel}.
								\end{tablenotes}
						\end{threeparttable}        }
					\end{table}

					\begin{table}[H]\centering\caption{\label{tab:BackBending:Average}Do-intervention (Backward-bending function  \citep{Hanoch1965})\\ \small Entries represent average (standard error) of estimates\\ obtained from $100$ samples of 10,000 observations each\\ 
							$X=h_{4}\left(Z, \eta\right)^{1}$; \hspace{1em} $Y=g_5(X,\bm\epsilon)^{2}$\\\vspace{1mm}		
							$[\upalpha_1,\upalpha_2]=[1,3]$\\ $[\upbeta_1,\upbeta_2,\upbeta_3]=[0.5,0.6,0.1]$,\hspace{1em} $[\upgamma_1,\upgamma_2,\upgamma_3]=[1,1,-3]$}\vspace{1mm}
						\scaleobj{0.8}{\begin{threeparttable}	
								\begin{tabular}{|c|c|c|c|c|c|c|c|}
									\hline
									\multirow{2}{*}{\begin{tabular}[c]{@{}c@{}}$\overset{(a)}{\text{Quantile}}$\\  of  X\end{tabular}} & \multirow{2}{*}{\begin{tabular}[c]{@{}c@{}}\\ $\overset{(b)}{X}$\end{tabular}} & \multicolumn{2}{c|}{$\mathbb{E}\left[Y|X=x\right]$} & \multicolumn{4}{c|}{$\mathbb{E}\left[Y|\text{do}(X=x)\right]$}    \\ \cline{3-8} 
									&                    & $\overset{(c)}{\text{real}}$           & $\overset{(d)}{\text{synt}}$             & $\overset{(e)}{\text{real}}$  & $\overset{(f)}{\text{synt}}$    & $\overset{(g)}{\text{IN}}$ & $\overset{(h)}{\text{PM}}$ \\ \hline\multirow{2}{*}{5}                                                                                  & 0.07                                                                               & 0.13                          & 0.08$^\dag$                        & \textbf{0.44}                 & \textbf{0.42*}                & -        0.17                                                                                                                                                             & 0.13                                                                                                                                    \\ 
									\cline{2-8}
									& (0.059)                                                                            & (0.021)                       & (0.060)                       & (0.025)                       & (0.114)                       & (0.131)                                                                                                                                                                   & (0.090)                                                                                                                                 \\ 
									\hline
									\multirow{2}{*}{10}                                                                                 & 0.10                                                                               & 0.13                          & 0.09$^{\dag}$                          & \textbf{0.44}                 & \textbf{0.42*}                & -        0.18                                                                                                                                                             & 0.14                                                                                                                                    \\ 
									\cline{2-8}
									& (0.057)                                                                            & (0.021)                       & (0.060)                       & (0.025)                       & (0.114)                       & (0.131)                                                                                                                                                                   & (0.090)                                                                                                                                 \\ 
									\hline
									\multirow{2}{*}{15}                                                                                 & 0.12                                                                               & 0.14                          & 0.09$^{\dag}$                          & \textbf{0.44}                 & \textbf{0.42*}                & -        0.18                                                                                                                                                             & 0.14                                                                                                                                    \\ 
									\cline{2-8}
									& (0.070)                                                                            & (0.020)                       & (0.059)                       & (0.027)                       & (0.113)                       & (0.130)                                                                                                                                                                   & (0.090)                                                                                                                                 \\ 
									\hline
									\multirow{2}{*}{20}                                                                                 & 0.18                                                                               & 0.14                          & 0.09$^{\dag}$                          & \textbf{0.45}                 & \textbf{0.43*}                & -        0.19                                                                                                                                                             & 0.16                                                                                                                                    \\ 
									\cline{2-8}
									& (0.091)                                                                            & (0.021)                       & (0.058)                       & (0.029)                       & (0.112)                       & (0.134)                                                                                                                                                                   & (0.090)                                                                                                                                 \\ 
									\hline
									\multirow{2}{*}{25}                                                                                 & 0.35                                                                               & 0.16                          & 0.11$^{\dag}$                          & \textbf{0.47}                 & \textbf{0.44*}                & -        0.24                                                                                                                                                             & 0.21                                                                                                                                    \\ 
									\cline{2-8}
									& (0.120)                                                                            & (0.021)                       & (0.058)                       & (0.033)                       & (0.109)                       & (0.139)                                                                                                                                                                   & (0.087)                                                                                                                                 \\ 
									\hline
									\multirow{2}{*}{30}                                                                                 & 0.52                                                                               & 0.17                          & 0.12$^{\dag}$                          & \textbf{0.48}                 & \textbf{0.46*}                & -        0.28                                                                                                                                                             & 0.26                                                                                                                                    \\ 
									\cline{2-8}
									& (0.132)                                                                            & (0.022)                       & (0.056)                       & (0.033)                       & (0.105)                       & (0.138)                                                                                                                                                                   & (0.089)                                                                                                                                 \\ 
									\hline
									\multirow{2}{*}{35}                                                                                 & 0.72                                                                               & 0.19                          & 0.15$^{\dag}$                           & \textbf{0.50}                 & \textbf{0.48*}                & -        0.33                                                                                                                                                             & 0.33                                                                                                                                    \\ 
									\cline{2-8}
									& (0.139)                                                                            & (0.024)                       & (0.057)                       & (0.034)                       & (0.102)                       & (0.135)                                                                                                                                                                   & (0.091)                                                                                                                                 \\ 
									\hline
									\multirow{2}{*}{40}                                                                                 & 0.93                                                                               & 0.22                          & 0.17$^{\dag}$                           & \textbf{0.53}                 & \textbf{0.50*}                & -        0.35                                                                                                                                                             & 0.41                                                                                                                                    \\ 
									\cline{2-8}
									& (0.110)                                                                            & (0.023)                       & (0.058)                       & (0.032)                       & (0.099)                       & (0.126)                                                                                                                                                                   & (0.085)                                                                                                                                 \\ 
									\hline
									\multirow{2}{*}{45}                                                                                 & 1.06                                                                               & 0.24                          & 0.19$^{\dag}$                          & \textbf{0.54}                 & \textbf{0.52*}                & -        0.33                                                                                                                                                             & 0.46                                                                                                                                    \\ 
									\cline{2-8}
									& (0.082)                                                                            & (0.022)                       & (0.058)                       & (0.030)                       & (0.098)                       & (0.125)                                                                                                                                                                   & (0.081)                                                                                                                                 \\ 
									\hline
									\multirow{2}{*}{50}                                                                                 & 1.17                                                                               & 0.25                          & 0.20$^{\dag}$                          & \textbf{0.56}                 & \textbf{0.53*}                & -        0.30                                                                                                                                                             & 0.50                                                                                                                                    \\ 
									\cline{2-8}
									& (0.103)                                                                            & (0.026)                       & (0.061)                       & (0.033)                       & (0.095)                       & (0.130)                                                                                                                                                                   & (0.087)                                                                                                                                 \\ 
									\hline
									\multirow{2}{*}{55}                                                                                 & 1.37                                                                               & 0.28                          & 0.23$^{\dag}$                           & \textbf{0.59}                 & \textbf{0.56*}                & -        0.17                                                                                                                                                             & 0.58*                                                                                                                                   \\ 
									\cline{2-8}
									& (0.163)                                                                            & (0.035)                       & (0.069)                       & (0.041)                       & (0.094)                       & (0.162)                                                                                                                                                                   & (0.108)                                                                                                                                 \\ 
									\hline
									\multirow{2}{*}{60}                                                                                 & 1.55                                                                               & 0.31                          & 0.26$^{\dag}$                           & \textbf{0.61}                 & \textbf{0.59*}                & -        0.03                                                                                                                                                             & 0.66*                                                                                                                                   \\ 
									\cline{2-8}
									& (0.136)                                                                            & (0.031)                       & (0.073)                       & (0.040)                       & (0.093)                       & (0.172)                                                                                                                                                                   & (0.101)                                                                                                                                 \\ 
									\hline
									\multirow{2}{*}{65}                                                                                 & 1.70                                                                               & 0.34                          & 0.29$^{\dag}$                           & \textbf{0.64}                 & \textbf{0.61*}                & 0.12                                                                                                                                                                      & 0.73*                                                                                                                                   \\ 
									\cline{2-8}
									& (0.157)                                                                            & (0.036)                       & (0.083)                       & (0.046)                       & (0.095)                       & (0.193)                                                                                                                                                                   & (0.110)                                                                                                                                 \\ 
									\hline
									\multirow{2}{*}{70}                                                                                 & 2.08                                                                               & 0.42                          & 0.37$^{\dag}$                           & \textbf{0.70}                 & \textbf{0.68*}                & 0.49                                                                                                                                                                      & 0.92                                                                                                                                    \\ 
									\cline{2-8}
									& (0.226)                                                                            & (0.060)                       & (0.108)                       & (0.072)                       & (0.103)                       & (0.241)                                                                                                                                                                   & (0.135)                                                                                                                                 \\ 
									\hline
									\multirow{2}{*}{75}                                                                                 & 2.59                                                                               & 0.55                          & 0.51$^{\dag}$                          & \textbf{0.81}                 & \textbf{0.79*}                & 0.95*                                                                                                                                                                     & 1.20                                                                                                                                    \\ 
									\cline{2-8}
									& (0.322)                                                                            & (0.109)                       & (0.170)                       & (0.128)                       & (0.138)                       & (0.266)                                                                                                                                                                   & (0.174)                                                                                                                                 \\ 
									\hline
									\multirow{2}{*}{80}                                                                                 & 3.28                                                                               & 0.77                          & 0.73$^{\dag}$                           & \textbf{0.99}                 & \textbf{0.97*}                & 1.45                                                                                                                                                                      & 1.57                                                                                                                                    \\ 
									\cline{2-8}
									& (0.366)                                                                            & (0.161)                       & (0.227)                       & (0.188)                       & (0.180)                       & (0.249)                                                                                                                                                                   & (0.196)                                                                                                                                 \\ 
									\hline
									\multirow{2}{*}{85}                                                                                 & 4.19                                                                               & 1.15                          & 1.12$^{\dag}$                           & \textbf{1.31}                 & \textbf{1.29*}                & 1.99                                                                                                                                                                      & 2.02                                                                                                                                    \\ 
									\cline{2-8}
									& (0.534)                                                                            & (0.272)                       & (0.333)                       & (0.331)                       & (0.276)                       & (0.274)                                                                                                                                                                   & (0.274)                                                                                                                                 \\ 
									\hline
									\multirow{2}{*}{90}                                                                                 & 5.57                                                                               & 1.85                          & 1.83$^{\dag}$                           & \textbf{1.97}                 & \textbf{1.90*}                & 2.64                                                                                                                                                                      & 2.66                                                                                                                                    \\ 
									\cline{2-8}
									& (0.951)                                                                            & (0.371)                       & (0.397)                       & (0.620)                       & (0.357)                       & (0.537)                                                                                                                                                                   & (0.366)                                                                                                                                 \\ 
									\hline
									\multirow{2}{*}{95}                                                                                 & 8.11                                                                               & 3.00                          & 2.99$^{\dag}$                           & \textbf{3.54}                 & \textbf{2.98*}                & 3.84*                                                                                                                                                                     & 3.67*                                                                                                                                   \\ 
									\cline{2-8}
									& (0.929)                                                                            & (0.278)                       & (0.272)                       & (0.639)                       & (0.294)                       & (0.770)                                                                                                                                                                   & (0.359)                                                                                                                                 \\
									\hline
								\end{tabular}%
								\begin{tablenotes}
									\item[(*)]\small No significant difference when comparing the  empirical distributions of $((f),(g),(h))$ and benchmark $(e)$ sequences of estimates in quantiles $[\Lambda,1-\Lambda]$ with $\Lambda=0.05$  at significance level of $95\%$ (Sliced Wasserstein Distance).
									\item[($\dag$)]  No significant difference when comparing the empirical distributions of $(d)$ and benchmark $(c)$ sequences of estimates in quantiles $[\Lambda,1-\Lambda]$ with $\Lambda=0.05$  at significance level of $95\%$ (Sliced Wasserstein Distance).
									\item[1] $\scaleobj{0.9}{X=h_4(Z,\eta) = \exp(Z\upeta - \upalpha_1 Z)-\upalpha_2 (Z\upeta-\upalpha_1 Z).  }$
									\item[2] $\scaleobj{0.9}{Y=g_5(X,\bm\epsilon)  =  \upbeta_1 X + \upbeta_2\epsilon +\upbeta_3 X\epsilon.}$	
									\item[(c)] $\scaleobj{0.9}{\widehat{\mathbb{E}}\left[Y_{\text{real}}^{\text{O}}|X_{\text{real}}^{\text{O}}=x\right]}$ is the empirical variant of $\mathds{E}\left[Y|X=x\right]$ in the real data.				
									
									\item[(d)] $\scaleobj{0.9}{\widehat{\mathbb{E}}\left[Y_{\text{synt}}^{\text{O}}|X_{\text{synt}}^{\text{O}}=x\right]}$ is the empirical variant of $\mathds{E}\left[Y|X=x\right]$ in the synthetic data.
									\item[(e)] $\scaleobj{0.9}{\widehat{\mathbb{E}}\left[Y_{\text{real}}^{\text{CF}}|X_{\text{real}}^{\text{CF}}=x\right]}$ is the empirical variant of $\mathds{E}\left[Y|\text{do}(X=x)\right]$ in the real data.
									\item[(f)] $\scaleobj{0.9}{\widehat{\mathbb{E}}\left[Y_{\text{synt}}^{\text{CF}}|X_{\text{synt}}^{\text{CF}}=x\right]}$ is the empirical variant of $\mathds{E}\left[Y|\text{do}(X=x)\right]$ in the synthetic data.		
									\item [(g)] Control variable estimator (IN)  \citep{imbens2009identification}.
									\item [(h)] Partial Means estimator (PM) \citep{newey1994kernel}.
								\end{tablenotes}
						\end{threeparttable}        }
					\end{table}
					
					\section{Proofs}\label{Appendix:Proof}
					
					\renewcommand{\theequation}{\thesection.\arabic{equation}}
					\setcounter{equation}{0}
					%\subsection*{Proof of lemma \ref{lemma:integral}}
					\begin{proof}[\textbf{Proof of Lemma \ref{lemma:eps:support}}]
					By assumption \ref{asu:common} $\forall x$  $\text{Supp}(\left.\upeta\right|X=x) = \text{Supp}(\upeta)$,
					\[
					\text{Supp}(\left.\bm\epsilon\right|X=x)=\bigcup_{\eta\in\text{Supp}(\upeta|X=x)}\text{Supp}(\left.\bm\epsilon\right|X=x,\upeta=\eta)=\bigcup_{\eta\in\text{Supp}(\upeta)}\text{Supp}(\left.\bm\epsilon\right|X=x,\upeta=\eta) \quad \forall x\in\text{Supp}(X),
					\]
					and by definition of the triangular model $\bm\upepsilon|X,\upeta \overset{d}{=}\bm\upepsilon|\upeta$ (equivalence in distribution). Consequently
					\[
					\bigcup_{\eta\in\text{Supp}(\upeta)}\text{Supp}(\left.\bm\epsilon\right|X=x,\upeta=\eta)=\bigcup_{\eta\in\text{Supp}(\upeta)}\text{Supp}(\left.\bm\epsilon\right|\upeta=\eta)=\text{Supp}(\bm\epsilon) \quad \forall x\in\text{Supp}(X).
					\]
					Hence, Eq.  \eqref{eq:Eps:Support} holds.
				\end{proof}
			
					\begin{proof}[\textbf{Proof of Lemma \ref{lemma:integral}}]
						By the law of total expectation,
						\begin{align}
							& F_{Y|X=x, Z=z}(y) =  \int F_{Y|X=x, Z=z,\upeta=\eta}(y)f_{\upeta|X=x,Z=z}\left(\eta\right)d\eta. \nonumber
						\end{align}	
						Since $Z$ does not appear in eq. 
						\eqref{Y:Structural} and $\bm\epsilon\ci Z|\upeta$, then $F_{Y|X=x,Z=z,\upeta=\eta}(y)=F_{Y|X=x,\upeta=\eta}(y)$. Hence,	
						\begin{align}
							& F_{Y|X=x, Z=z}(y) =  \int F_{Y|X=x, \upeta=\eta}(y)f_{\upeta|X=x,Z=z}\left(\eta\right)d\eta.\label{eq:Fredholm}
						\end{align}	
						By assumption \ref{asu:continuity}, $f_{X|Z=z}(x)$ is non-zero for any $(x,z)\in\text{Supp}(X,Z)$. Thus, by Bayes law 
						\[
						f_{\upeta|X=x,Z=z}\left(\eta\right)=\frac{f_{X|Z=z,\upeta=\eta}(x)f_{\upeta|Z=z}(\eta)}{f_{X|Z=z}(x)}.
						\] 
						Insert this expression into eq. \eqref{eq:Fredholm} gives,
						\begin{align}
							& F_{Y|X=x, Z=z}(y) =  \int F_{Y|X=x, \upeta=\eta}(y)\frac{f_{X|Z=z,\upeta=\eta}(x)f_{\upeta|Z=z}(\eta)}{f_{X|Z=z}(x)}d\eta. \nonumber
						\end{align}	
						By assumption $Z\ci\upeta$ implying that $f_{\upeta|Z=z}(\eta)=f_{\upeta}(\eta)$. Changing variables $\upomega=F_{\upeta}(\upeta)\sim U[0,1]$ and applying the probability integral transform we get,
						\begin{align}
							& F_{Y|X=x, Z=z}(y) =  \int_{0}^1 F_{Y|X=x, \upeta=F_{\upeta}^{-1}(\omega)}(y)\frac{f_{X|Z=z,\upeta=F_{\upeta}^{-1}(\omega)}(x)}{f_{X|Z=z}(x)}d\omega. 
							\label{eq:F}
						\end{align}
						By definition of the random variable $\upomega$ we have that $f_{X|Z=z,\upeta=F_{\upeta}^{-1}(\omega)}(x) = f_{X|Z=z,F_{\upeta}(\upeta)=\omega}(x) = f_{X|Z=z,\upomega=\omega}(x)$, $F_{Y|X=x,\upeta=F_{\upeta}^{-1}(\omega)}(y) = F_{Y|X=x,F_{\upeta}(\upeta)=\omega}(y) = F_{Y|X=x,\upomega=\omega}(y)$ and $f_{\upomega|Z=z}(\omega)=f_{\upomega}(\omega)=1$. Thus,
						\begin{align}
							& F_{Y|X=x, Z=z}(y) =  \int_{0}^1  F_{Y|X=x, \upomega=\omega}(y)\frac{f_{X|\upomega=\omega,Z=z}(x)f_{\upomega}(\omega)}{f_{X|Z=z}(x)}d\omega. 
						\end{align}
						By Bayes law we get,
						\begin{align}
							& F_{Y|X=x, Z=z}(y) =  \int_{0}^1 F_{Y|X=x, \upomega=\omega}(y)f_{\upomega|X=x,Z=z}(\omega)d\omega. 
						\end{align}
						which concludes the proof.	
					\end{proof}

							\begin{proof}[\textbf{Proof of Lemma \ref{Lemma:Decoupling}}]   
						By the inverse probability integral transform, $\Lambda_0(\upomega)$ has the same distribution as $\upeta$. 
						Since by assumption $Z$ and $\upeta$ are independent, and by construction $Z$ and $\Lambda_0$ are independent,  it follows 
						from Eq. (\ref{eq:tilde_X}) 
						that $(X^*,Z)\overset{d}{=}(X,Z)$. Also,
						$(X^*,\Lambda_0(\upomega)) \overset{d}=(X,\upeta)$. 
						Similarly, by the inverse probability integral transform,  $(\Lambda_0(\upomega),\Lambda_1(\upomega,\nu_1),\ldots,\Lambda_d(\upomega,\nu_1,\ldots,\nu_d))$ has the same distribution as $(\upeta,\bm{\epsilon})$. Since by the model assumption, $Z$ is independent of both $\upeta$ and $\bm \epsilon$, then
						$(X^*$, $\Lambda_0(\upomega)$, $\Lambda_1(\nu_1)$, $\ldots$, $\Lambda_d(\nu_d))$ has the same distribution as $(X,\upeta,{\bm \epsilon})$. Hence, by 
						Eq. (\ref{eq:tilde_Y}), 
						$(X^*,Y^*,Z)\overset{d}= (X,Y,Z)$. 
						%following the one-to-one mapping between  $(\upeta,\bm\epsilon)$ and $(\upomega,\bm\nu)$.
					\end{proof}
								
					\begin{proof}[\textbf{Proof of Lemma \ref{lemma:Invariant:Do}}]	
						By the definition of the intervention, Eq. (\ref{eq:do_X}), 
						\begin{align}
							& T(y|x) = P\close{Y\le y|\text{do}(X=x)}= \mathbb{E}_{\bm\epsilon\sim F_{\bm\epsilon}}\left[\mathds{1}\left\{g(x,\bm{\epsilon})\le y\right\}\right]
							= \E_\upeta \left[   \E_{\be|\upeta}  
							[ \mathds{1}\left\{g(x,\be)\le y\right\}      ]   \right]
							\nonumber
						\end{align}
						First, we make a change of variables from 
						$\upeta$ to $\upomega=F_\upeta(\upeta)$. Since $\upeta$ is assumed continuous, then by the probability integral transform $\upomega\sim U[0,1]$. 
						Next, given a value of $\upeta$ or equivalently of $\upomega$, we make a change of variables
						from $\be$ to $\bm\nu$ as follows, inverting
						Eqs. (\ref{eq:Lambda_j}): $\nu_1 = F_{\epsilon_1|\Lambda_0(\upomega) }(\epsilon_1)$,
						and recursively
						$\nu_j =  F_{\epsilon_j|\Lambda_0(\upomega), 
							\Lambda_1(\upomega,\nu_1),\ldots,\Lambda_{j-1}(\upomega,\nu_1,\ldots,\nu_{j-1}) }(\epsilon_j)$.
						By the assumption that $\be|\upeta$ is continuous and the probability integral transform it follows that $(\nu_1,\ldots,\nu_d)\sim U[0,1]^d$ and is also independent of $\upomega$. 
						Then, in terms of the transformed variables, 
						\[
						T(y|x) = \E_{\upomega} \left[
						\E_{\bm\nu}[
						\mathds{1}\left\{g(x,\Lambda_1(\upomega,\nu_1),\ldots, \Lambda_d(\upomega,\bm \nu))\le y\right\}
						]
						\right]
						= \E_{\upomega} \left[
						\E_{\bm\nu}[
						\mathds{1}\left\{
						g^*(x,\upomega,\bm\nu) \leq y
						\right\}
						]
						\right].
						\]
					\end{proof}
%We define the counterfactual distribution (CF) of $(Y,X)$ given $Z$ in terms of Eq. \eqref{eq:def_H} as,
%\begin{align}
%	F_{Y,X|Z=z}^{\text{CF}}(y,x):=\int^{x}H(y|t)f_{X|Z=z}(t)dt. \label{eq:CF:def:H}
%\end{align}
%Hence, for any $z$ the random vector $(Y^{\text{CF}}, X)|Z=z$ is distributed $Y^{\text{CF}}\sim H(\cdot|X)$ and $X\sim X|Z=z$.
\begin{proof}[\textbf{proof of Lemma \ref{lemma:Fourier:CF:representation}}]
	The Fourier transform of the joint density function $f_{y,x|Z=z}^{\text{CF}}(y,x)$ is,
	\[
	\upalpha(z,\xi_2,\xi_3) := \int_x\int_y\Psi_{\xi_1}(-Y)\Psi_{\xi_2}(-X)f_{Y,X|Z=z}(y,x) dy dx = \mathbb{E}_{(Y,X)\sim F_{Y,X|Z=z}^{\text{CF}}}\left[\Psi_{\xi_1}(-Y)\Psi_{\xi_2}(-X)\right].
	\]
	The inverse Fourier transform is,
	\begin{align}
		f_{Y,X|Z=z}^{\text{CF}}(y,x)=   
		\int_{\xi_1}\int_{\xi_2}\upalpha(z,\xi_2,\xi_3)\Psi_{\xi_1}(y)\Psi_{\xi_2}(x) d\xi_2 d\xi_3.\nonumber
	\end{align}	
	Hence, the conditional density $f_{Y,X|Z=z}(y,x)$ admits the varying coefficient-Fourier representation,
	\begin{align}
		f_{Y,X|Z=z}^{\text{CF}}(y,x)=   
		\int_{\xi_1}\int_{\xi_2}\mathbb{E}_{(Y,X)\sim f_{Y,X|Z=z}^{\text{CF}}}\left[\Psi_{\xi_1}(-Y)\Psi_{\xi_2}(-X)\right]\Psi_{\xi_1}(y)\Psi_{\xi_2}(x) d\xi_2 d\xi_3.\label{eq:CF}
	\end{align}
By probability integral transform the counterfactual density can be represented equivalently as,
\begin{align}	
	\Gamma_{\xi_1, \xi_2}(z;(\widetilde{h}, \widetilde{g}))=\mathbb{E}_{\substack{\upomega\sim U[0,1]\\X\sim F_{X|Z=z}}}\left[\int \exp(-i\xi_1 \cdot X) \exp(-i\xi_2 \cdot \widetilde{g}(X, \upomega, \bm\nu))d\bm\nu\right].\nonumber
\end{align}
This gives Eq. \eqref{eq:box:counter:X:Y}.

\end{proof}

\begin{proof}[\textbf{proof of Lemma \ref{lemma:Fourier:representation}}]
		Let $\Psi_{\xi}(t):=\exp\left(-2\pi\cdot t\cdot\xi\cdot i\right)$.
	The Fourier transform of the joint density function $f_{y,x|Z=z}(y,x)$ is,
	\[
	\upalpha(z,\xi_2,\xi_3) := \int_x\int_y\Psi_{\xi_1}(-Y)\Psi_{\xi_2}(-X)f_{Y,X|Z=z}(y,x) dy dx = \mathbb{E}_{(Y,X)\sim F_{Y,X|Z=z}}\left[\Psi_{\xi_1}(-Y)\Psi_{\xi_2}(-X)\right].
	\]
	The inverse Fourier transform is,
	\begin{align}
		f_{Y,X|Z=z}(y,x)=   
		\int_{\xi_1}\int_{\xi_2}\upalpha(z,\xi_2,\xi_3)\Psi_{\xi_1}(y)\Psi_{\xi_2}(x) d\xi_2 d\xi_3.\nonumber
	\end{align}	
	Hence, the conditional density $f_{Y,X|Z=z}(y,x)$ admits the varying coefficient-Fourier representation,
	\begin{align}
		f_{Y,X|Z=z}(y,x)=   
		\int_{\xi_1}\int_{\xi_2}\mathbb{E}_{(Y,X)\sim f_{Y,X|Z=z}}\left[\Psi_{\xi_1}(-Y)\Psi_{\xi_2}(-X)\right]\Psi_{\xi_1}(y)\Psi_{\xi_2}(x) d\xi_2 d\xi_3.\label{eq:obs}
	\end{align}
	Eq. \eqref{eq:obs} involves an unknown functions $g$ and $h$ or equivalently,
	\begin{align}%[box=\tcbhighmath]{equation}
		f_{Y,X|Z=z}(y,x)=   
		\int_{\xi_1}\int_{\xi_2}\mathbb{E}_{(X,\upomega)\sim \mathcal{F}_{X,\upomega|Z=z}}\left[\int \Psi_{\xi_1}(-g(X,\upomega,\bm\nu))\Psi_{\xi_2}(-X)d\bm\nu\right]\Psi_{\xi_1}(y)\Psi_{\xi_2}(x) d\xi_2 d\xi_3.\label{eq:obs:pair}
	\end{align}
	Using  $m(x,\omega,\bm\nu):=\Psi_{\xi_1}(-g(x,\omega,\bm\nu))\Psi_{\xi_2}(-x)$,
	\begin{align}%[box=\tcbhighmath]{equation}
		f_{Y,X|Z=z}(y,x)=   
		\int_{\xi_1}\int_{\xi_2}\mathbb{E}_{(X,\upomega)\sim \mathcal{F}_{X,\upomega|Z=z}}\left[\int m(X,\upomega,\bm\nu)d\bm\nu\right]\Psi_{\xi_1}(y)\Psi_{\xi_2}(x) d\xi_2 d\xi_3.\label{eq:obs:pair}
	\end{align}
We know each of the Fourier-Transform coefficients, \[\Gamma_{\xi_1, \xi_2}(z):=\mathbb{E}_{(X,Y)\sim \mathcal{F}_{X,Y|Z=z}}\left[\int \Psi_{\xi_1}(-Y)\Psi_{\xi_2}(-X)d\bm\nu\right]\] of the observed density, 
Hence Eq. \eqref{eq:coeff:obs:pair:Fourier} holds, namely,
\begin{align}%[box=\tcbhighmath]{equation}
	\Gamma_{\xi_1, \xi_2}(z)=   
	\mathbb{E}_{(X,\upomega)\sim \mathcal{F}_{X,\upomega|Z=z}}\left[\int \exp(-i\xi_1 \cdot X) \exp(-i\xi_2 \cdot g(X, \omega, \nu))d\bm\nu\right].
\end{align}
\end{proof}
\newcommand\U[2]{u_{ \xi_2,\xi_3}^*(z,{#1},{#2})}
\newcommand\Utilde[2]{\widetilde{u}_{ \xi_2,\xi_3}(z,{#1},{#2})}
\begin{proof}[\textbf{Proof of Theorem \ref{Theorem:Identifiability}}]
Define \[m_{\xi_1, \xi_2}(x,\upomega,\bm\nu):=\int \exp(-i\xi_1 \cdot x) \exp(-i\xi_2 \cdot g^*(x, \omega, \nu))d\bm\nu,\] and similarly,
\[m_{\xi_1, \xi_2}^{\prime}(x,\upomega,\bm\nu):=\int \exp(-i\xi_1 \cdot x) \exp(-i\xi_2 \cdot \widetilde{g}(x, \omega, \nu))d\bm\nu.\] 
Suppose that Eq. \eqref{eq:coeff:obs:pair:Fourier} is satisfied for $(m_{\xi_1, \xi_2}, \mathcal{F})$,
\begin{align}%[box=\tcbhighmath]{equation}
	\Gamma_{\xi_1, \xi_2}(z)=   
	\mathbb{E}_{(X,\upomega)\sim \mathcal{F}_{X,\upomega|Z=z}}\left[\int m_{\xi_1, \xi_2}(X,\upomega,\bm\nu)d\bm\nu\right].\label{eq:coeff:obs:pair}
\end{align}
Similarly, suppose that Eq. \eqref{eq:coeff:obs:pair:Fourier} is satisfied for $(m_{\xi_1, \xi_2}^{\prime}, \mathcal{F}^{\prime})$
\begin{align}%[box=\tcbhighmath]{equation}
	\Gamma_{\xi_1, \xi_2}(z)=   
	\mathbb{E}_{(X,\upomega)\sim \mathcal{F}_{X,\upomega|Z=z}}^{\prime}\left[\int m_{\xi_1, \xi_2}^{\prime}(X,\upomega,\bm\nu)d\bm\nu\right].\nonumber
\end{align}
In fact, Eq. \eqref{eq:coeff:obs:pair} is a variant of Fredholm integral equation of the first kind in terms, with a kernel $\frac{\partial}{\partial \omega}\mathcal{F}_{\upomega|X=x, Z=z}(\omega)$ and an unknown function $\int m_{\xi_1, \xi_2}(X,\upomega,\bm\nu)d\bm\nu$. The known function is $\Gamma_{\xi_1, \xi_2}(z)$.
Following theorem \ref{theorem:coefficients}, if there exist two distinct pairs of functions $(m, \mathcal{F})$ and $(m_{\xi_1, \xi_2}^{\prime}, \mathcal{F}^{\prime})$ satisfying the observational equivalence condition in Eq. \eqref{eq:coeff:obs:pair} it follows that \[\int m(X_{\xi_1, \xi_2},\upomega,\bm\nu)d\bm\nu = \int m_{\xi_1, \xi_2}^{\prime}(X,\upomega,\bm\nu)d\bm\nu \text{ almost everywhere}.\]

This equality is a sufficient condition for identifiability of $f_{Y,X|Z=z}^{\text{CF}}(y,x)$, because it implies that the Fourier transform coefficients of the counterfactual density are uniquely determined.
\end{proof}

	Let $g_1(x, \eta)$ and $g_2(x, \eta)$ be two different functions, and let $f_1(\eta | x, z)$ and $f_2(\eta | x, z)$ be two different densities expressible as infinite Gaussian mixtures characterized as, 
\begin{align}
	f_i(X,\upeta, \upeta^{\prime}|Z=z) = \sum_{j=1}^{\infty}\pi_j^{i}\phi\left(\left[x,\eta,\eta^{\prime}\right]; \left[\mu_{1, j}^{i}(z), \mu_{2, j}^{i}, \mu_{3, j}^{i}], \Sigma_j^{i}\right]\right)  \quad \forall i\in\left\{1,2\right\}. \label{eq:Mixture}	
\end{align}	
where $\left\{\Sigma_j^i\right\}_{j=1}^{\infty}$ determines the dependence among $(Z,\eta,\eta^{\prime})$, where the sequence of $\left\{\sigma_{j,\eta,\eta^{\prime}}^i\right\}_{j=1}^{\infty}$ determines the dependence between $(\upeta,\upeta^{\prime})$. For brevity of exposition, we define two specific settings, in which $(\upeta,\upeta^{\prime})$  are fully linearly dependent and another one in which they are linearly independent by $\rho=1$ and $\rho=0$, respectively.
\begin{lemma}\label{lemma:GaussianMixtureEquivalence}
	Suppose that the following observational equivalence holds for two mechanisms $(h^*, g^*)$ and $(\widetilde{h}, \widetilde{g})$ generating $(X^*,Y^*)$ and $(\widetilde{X},\widetilde{Y})$, respectively, in regime $\rho\in\left\{0, 1\right\}$,
	\begin{align}%[\right]\mathbb{E}_{(\upeta,\upeta^{\prime})\sim F_{\eta,\eta^{\prime}}(\cdot,\cdot;\rho=1)}} \left[  \right]
		& \mathbb{E}\left[ g(\widetilde{X}, \widetilde{Y})|Z=z,\rho=1\right]=\mathbb{E}\left[ g(X^*, Y^*)|Z=z,\rho=1\right] = \varphi(z) \quad\forall i\in \left\{1,2\right\} \text{ and } \forall z, \label{eq:obs:equiv}
	\end{align}%	$g_1(x, \eta) - g_2(x, \eta) = 0$
such that, $(X^*,Y^*)|Z, \rho$ and $(\widetilde{X},\widetilde{Y})|Z, \rho$ are strongly complete mixtures of Gaussians  with respect to $Z$.
	Then it also holds that,
		\begin{align}
		&\mathbb{E}\left[ g(\widetilde{X}, \widetilde{Y})|Z=z,\rho=0\right]=\mathbb{E}\left[ g(X^*, Y^*)|Z=z,\rho=0\right]   \quad\forall i\in \left\{1,2\right\} \text{ and } \forall z. \label{eq:CF:equiv}
	\end{align}
% This implies that $g_1(x, \eta, \upeta^{\prime}) - g_2(x, \eta, \upeta^{\prime}) = 0$ almost everywhere.
\end{lemma}	

\begin{proof}[\textbf{Proof of Lemma \ref{lemma:GaussianMixtureEquivalence}}]
	The expression of conditional expectations in Eq.\eqref{eq:obs:equiv} can be simplified to,
	\begin{align}
		\mathbb{E}[g(\widetilde{X}, \widetilde{Y})-g(X^*, Y^*) | Z = z, \rho=1] = 0  \quad \forall z.  \label{eq:g_i:expected}
	\end{align}	

	Following Theorem 1 \citep{alamatsaz1983completeness}, exponential family mixtures are strongly complete, which include a mixture of Gaussians as in our present case.
    By definition of strong completeness,  Eq. \eqref{eq:g_i:expected} implies that, 
	\begin{align}
	\mathbb{E}[g(\widetilde{X}, \widetilde{Y})-g(X^*, Y^*) | Z = z, \rho=0] = 0   \quad \forall z.  \label{eq:g_i:expected:other}
\end{align}	
    The equality in Eq. \eqref{eq:g_i:expected:other} is related to the fact that the joint distribution of the independent random vectors $(X^*,Y^*)$ and $(\widetilde{X},\widetilde{Y})$ given $Z$ and $\rho$, each is a strongly complete mixture of Gaussians, is a Gaussian mixture by itself. Because Gaussian mixtures are strongly complete with respect to common parameters determined by the value of $Z$, the invariance under setting $\rho=1$ in Eq. \eqref{eq:g_i:expected} implies invariance under setting $\rho=0$. 
\end{proof}
\begin{proof}[\textbf{Proof of Theorem \ref{theorem:coefficients}}]
Let $\mathcal{G}_j(x,\omega, \omega^{\prime}):=\int m_j(x,\omega, \omega^{\prime},\nu)d\nu \quad j\in\left\{1, 2\right\}$.	By supposition it holds equivalently for all $z$, 
	\[\mathbb{E}\left[\int \text{Re}(m_1(X,\upomega, \upomega^{\prime}, \nu)d\nu)|Z=z, \rho=1\right] = \mathbb{E}\left[\int \text{Re}(m_2(X,\upomega,\upomega^{\prime}, \nu)d\nu)|Z=z,  \rho=1\right],\] and \[\mathbb{E}\left[\int \text{Im}(m_1(X,\upomega,\upomega^{\prime}, \nu)d\nu)|Z=z, \rho=1\right] = \mathbb{E}\left[\int \text{Im}(m_2(X,\upomega, \upomega^{\prime},\nu)d\nu)|Z=z, \rho=1\right].\] 
	Following lemma \ref{lemma:GaussianMixtureEquivalence} these the equations above give,
	\[
	\int\text{Re}(m_1(x,\omega, \omega^{\prime}, \nu)d\nu=\int\text{Re}(m_2(x,\omega, \omega^{\prime} \nu)d\nu \text{ almost everywhere}
	\]
	and
	\[
	\int\text{Im}(m_1(x,\omega, \omega^{\prime}, \nu)d\nu=\int\text{Im}(m_2(x,\omega, \omega^{\prime}, \nu)d\nu \text{ almost everywhere}
	\]
	Thus, almost everywhere we get,
	\[\int m_1(x,\omega, \omega^{\prime}, \nu)d\nu=\int m_1(x,\omega, \omega^{\prime},\nu)d\nu\] almost everywhere. Hence, we have that almost everywhere,
	\[
	\mathcal{G}_1(x,\omega, \omega^{\prime}) =  \mathcal{G}_2(x,\omega, \omega^{\prime}).
	\]
\end{proof}
It follows that for any two distinct pairs of functions $(m, \mathcal{F})$ and $(m^{\prime}, \mathcal{F}^{\prime})$ satisfying the observational equivalence condition in Eq. \eqref{eq:obs:pair} it follows that \[\int m(X,\upomega,\bm\nu)d\bm\nu = \int m^{\prime}(X,\upomega,\bm\nu)d\bm\nu \text{ almost everywhere}.\]
%For example, if  
%\begin{align}
%	X=h(Z,\upeta,\varepsilon) = \phi(\upeta, Z) + \varepsilon, \label{eq:stochastic:h}
%\end{align} with $(Z, \upeta, \varepsilon)$ consisting of independently standard normally distributed variables. The last requirement is just a normalization. We have that $X|Z,\upeta$ is normally distributed and $\upeta|Z\overset{d}{=}\upeta$ is also normally distributed. Thus, $(X, \upeta)|Z=z$ can be exactly represented by an infinite Gaussian mixture due to the presence of the stochastic component $\varepsilon$ in Eq. \eqref{eq:stochastic:h}.

\section{The implications of wrongly relying on monotonicity}\label{Appendix:nonMonotonicity:Proof}
In what follows we state a necessary and sufficient condition for $Q$ to be a member in the equivalence class $\widetilde{\mathcal{C}}$ through an auxiliary lemma. For this purpose we define $Q(z,t):=F_{X|Z=z}^{-1}\left(t\right)$ $\forall (z,t)\in\text{Supp}(Z)\times[0,1]$.
\begin{lemma}\label{lemma:Auxiliary}
	There exists a function  $g^Q$ such that $(Q,g^Q)\in\widetilde{\mathcal{C}}$ if and only if for any pairs $(x,z),(x,z')\in\text{Supp}(X,Z)$ the following holds,
	\begin{align}
		F_{X|Z=z}(x)=F_{X|Z=z'}(x) \Rightarrow F_{Y|X=x,Z=z}=F_{Y|X=x,Z=z'}.
		\label{eq:existence:condition}
	\end{align}
	%$F_{\upomega|\widetilde{X}^*,\widetilde{Y},Z} = F_{\upomega|\widetilde{X}^*,Z}$.
\end{lemma}

\begin{proof}%Let $(\widetilde{h},\widetilde{g})\in\widetilde{\mathcal{C}}$ and
	Let $X^Q = Q(Z,\upomega)$. %From lemma \ref{lemma:Quantile:decomp} we get $\widetilde{h}=Q(Z,\Uplambda(Z,\upomega))$. 
	The first direction of the proof requires to show that if $(Q,g^Q)\in\widetilde{\cal C}$  then \eqref{eq:existence:condition} holds. Suppose by the contradiction that condition \eqref{eq:existence:condition} doesn't hold and there exists a function $g^Q:\mathbb{R}\times[0,1]^2\mapsto\mathbb{R}$ such that $(Q,g^Q)$ belong to equivalence class $\widetilde{\cal C}$. This gives that for all $(x,t)\in\text{Supp}(X)\times[0,1]$,
	\begin{align}
		F_{\widetilde{Y}^*|\widetilde{X}^*=x,Z=z}(y)=\int_{[0,1]^d} \mathds{1}\left\{g^Q(x,F_{X|Z=z}(x), \bm\nu)\le y\right\}d\bm\nu.\nonumber
	\end{align}
	Define $m(x,t):=\int_{[0,1]^d} \mathds{1}\left\{g^Q(x,t, \bm\nu)\le y\right\}d\bm\nu$. This gives,
	\begin{align}
		F_{\widetilde{Y}^*|\widetilde{X}^*=x,Z=z}(y)=m(x,F_{X|Z=z}(x)).\label{eq:GQ}
	\end{align} The violation of condition \eqref{eq:existence:condition} implies there exist pairs $(x,z),(x,z')\in\text{Supp}(X,Z)$ for which,
	\begin{align}
		F_{X|Z=z}(x)=F_{X|Z=z'}(x),\label{eq:X:Ztag}
	\end{align}
	and
	\begin{align}
		F_{\widetilde{Y}^*|\widetilde{X}^*=x,Z=z}\ne F_{\widetilde{Y}^*|\widetilde{X}^*=x,Z=z'}.\label{eq:FYX:Ztag}
	\end{align}
	However, \eqref{eq:X:Ztag} implies that $m(x,F_{X|Z=z}(x))=m(x,F_{X|Z=z'}(x))$. This and \eqref{eq:GQ} yield,
	\[
	F_{\widetilde{Y}^*|\widetilde{X}^*=x,Z=z}(y)=F_{\widetilde{Y}^*|\widetilde{X}^*=x,Z=z'}(y),
	\]
	which contradicts \eqref{eq:FYX:Ztag}. Thus, it follows that if  $(Q,g^Q)\in\widetilde{\cal C}$ condition \eqref{eq:existence:condition} holds. 
	
	The other direction of the proof requires to show that if \eqref{eq:existence:condition} holds, then $(Q,g^Q)\in\widetilde{\cal C}$. Note that if \eqref{eq:existence:condition} holds, then  $\widetilde{Y}^*\ci\upomega|\widetilde{X}^*,F_{X|Z}(\widetilde{X}^*)$. This condition implies that there exists $g^Q$ satisfying  \eqref{eq:GQ}. Thus, $(Q,g^Q)\in\widetilde{\cal C}$. 
\end{proof}

\begin{comment}
A sufficient condition for absolute integrability of $\mathbb{E}_{\bm\upepsilon|\upeta^{\prime}}\left[\Psi_{\xi}\left(-g\left(h(z,\upeta),\bm\upepsilon\right)\right)\right]$  over each linear combination of $(\upeta,\upeta^{\prime})$ is the existence of the conditional density $f_{Y_{\text{CF}}|\bm\epsilon,\upeta,\upeta^{\prime},Z}(y)$. This is so because of the following relationship,
	\begin{align}
& f_{Y_{\text{CF}}, X|\bm\epsilon,\upeta,\upeta^{\prime},Z}(y, x) = 
	\int_{\xi_1}	\int_{\xi_2}\mathbb{E}_{\bm\upepsilon|\upeta^{\prime}}\left[\Psi_{\xi_1}\left(-g\left(h(z,\upeta),\bm\upepsilon\right)\right)\right]\Psi_{\xi_2}\left(-h(z,\upeta)\right)\Psi_{\xi_1}(y)\Psi_{\xi_2}(x) d\xi_1 d\xi_2.\label{eq:counterfactual}%\\ & f_{X|Z=z}(x) =  \int_{\xi}\mathbb{E}_{\upeta}\left[\Psi_{\xi}\left(-h(z,\upeta)\right)\right]\Psi_{\xi}(x) d\xi.\nonumber	
\end{align}
This is so because if the Fourier coefficient $\mathbb{E}_{\bm\upepsilon|\upeta^{\prime}}\left[\Psi_{\xi}\left(-g\left(h(z,\upeta),\bm\upepsilon\right)\right)\right]$ exists, the expectation is finite. Consequently, we have integrability which implies also an absolute integrability.
\end{comment}

	\pagebreak

\end{subappendices}

\end{appendices}

\end{document}